\documentclass[10pt, onecolumn, letterpaper]{IEEEtran}
\renewcommand{\baselinestretch}{0.9}
\usepackage{subfigure}
\usepackage{times}
\usepackage{latexsym}
\usepackage{morefloats}
\usepackage{amssymb}
\usepackage{amsmath}

\usepackage{enumitem}
\usepackage{mathrsfs}
\usepackage{amsgen,amsfonts,amsbsy,amsthm}
\usepackage{graphicx}
\usepackage{wasysym}
\usepackage{pxfonts}
\usepackage{stackrel}
\usepackage{array}
\newcommand{\remark}[1]{{\emph{Remark:}} #1}
\newcommand{\remarks}{\emph{Remarks:}}

\newcommand{\BEQA}{\begin{eqnarray}}
\newcommand{\EEQA}{\end{eqnarray}}
\newcommand{\define}{\stackrel{\triangle}{=}}
\newcommand{\EXP}[1]{\mathsf{E}\!\left(#1\right)}
\newcommand{\prob}[1]{\mathsf{Pr}\left(#1\right)}

\newcommand{\ind}{\mathbf{1}}

\newcommand{\bs}{\beta_s}
\newcommand{\bc}{\beta_c}
\newcommand{\bd}{\beta_d}
\newcommand{\bx}{\beta_x}
\newcommand{\Ccal}{\mathcal{C}}
\newcommand{\Acal}{\mathcal{A}}
\newcommand{\Bcal}{\mathcal{B}}
\newcommand{\Tcal}{\mathcal{T}}

\newtheorem{theorem}{Theorem}

\newtheorem{lemma}{Lemma}

\begin{document}
\title{Analytical Modeling of IEEE~802.11 \\Type CSMA/CA Networks with\\ Short Term Unfairness}
%\title{Modeling, Performance Analysis,\\and Optimization of Long Distance,\\ Single Hop, Homogeneous WiFi Networks:\\ Challenges and Solutions}
%\title{An Approximate Analytical Model\\for WiFi Star Topology\\with Non-negligible Equal Propagation Delays\\and its implications}
\author{Abhijit Bhattacharya and Anurag Kumar\\Dept.\ of Electrical Communication Engineering\\Indian Institute of Science, Bangalore, 560012, India\\ email: \{abhijit, anurag\}@ece.iisc.ernet.in}

%\IEEEauthorblockN{Rachit Srivastava}
		%\IEEEauthorblockA{Deptt. of Electrical Communication Engineering \\
		%				Indian Institute of Science \\
		%				Bangalore, India 560012\\
		%				\{, abhibhattacharya2007, anurageceiisc\}@gmail.com}
%						rachitsri@gmail.com}
%\and
%		\IEEEauthorblockN{Anurag Kumar}
%		\IEEEauthorblockA{Deptt. of Electrical Communication Engineering \\
%						Indian Institute of Science\\
%						Bangalore, India 560012 \\
%						anurag@ece.iisc.ernet.in}

\maketitle

\begin{abstract}
We consider single-hop topologies with saturated transmitting nodes, using IEEE~802.11 DCF for medium access. However, unlike the conventional WiFi, we study systems where one or more of the protocol parameters are different from the standard, and/or where the propagation delays among the nodes are not negligible compared to the duration of a backoff slot. We observe that for several classes of protocol parameters, and for large propagation delays, such systems exhibit a certain performance anomaly known as short term unfairness, which may lead to severe performance degradation. The standard fixed point analysis technique (and its simple extensions) do not predict the system behavior well in such cases; a mean field model based asymptotic approach also is not adequate to predict the performance for networks of practical sizes in such cases. We provide a detailed stochastic model that accurately captures the system evolution. Since an exact analysis of this model is computationally intractable, we develop a novel approximate, but accurate, analysis that uses a parsimonious state representation for computational tractability. Apart from providing insights into the system behavior, the analytical method is also able to quantify the extent of short term unfairness in the system, and can therefore be used for tuning the protocol parameters to achieve desired throughput and fairness objectives. 

% , where the propagation delays among the nodes are not negligible compared to the slot duration. In this situation, we find that there is misaligned sensing of channel idleness, and also short-term unfairness in access to the medium. We demonstrate that existing analysis techniques (or, extensions thereof) are unable to account for these features, resulting in inaccurate prediction of the performance. Focusing on the case in which transmitters are equidistant from one another, and also each receiver is equidistant from all the transmitters,  Numerical experiments show, the approximate analysis predicts the system throughput to an accuracy of 2-3\%, and collision probabilities to an accuracy of 3-8\% compared to simulations. Interestingly, we observe that as propagation delay increases, the collision probability of a node initially increases, but then flattens out, contrary to what one might intuitively expect. Finally, we also demonstrate how to optimize slot duration using the approximate analysis for maximizing system throughput.
\end{abstract}

\section{Introduction}
\label{sec:intro}

The IEEE~802.11 Distributed Coordination Function (DCF) is perhaps the most widely known CSMA/CA standard, thanks to its ubiquitous presence in ``WiFi'' networks. Due to its simple implementation, and the advent of inexpensive chipsets, however, the DCF is being considered for applications beyond conventional WiFi, e.g., energy harvesting wireless sensor networks \cite{fafoutis-etal14}, Unmanned Aerial Vehicle (UAV) communications \cite{brown-etal04uav}, etc. UAV systems are becoming a popular choice for aerial remote sensing applications \cite{colomina-molina14}, thus further widening the range of possibilities with DCF.

The most popular performance analysis of IEEE~802.11 CSMA/CA (WiFi) networks was provided by Bianchi in the seminal work \cite{bianchi00performance}, and was later generalized by Kumar et al. \cite{kumar-etal04new-insights}. We shall provide a brief overview of this technique later in this chapter. However, it is now well-known that this analysis might not work if the DCF backoff parameters are different from those in the standard; in particular, Ramaiyan et al. \cite{ramaiyan-etalYYfp-analysis} demonstrated via some examples that the analysis may not capture the system performance well when the backoff sequences are such that the system exhibits short-term unfairness, i.e., one node or the other repeatedly succeeds in acquiring the channel for a long random time period, while the other nodes languish at large backoff durations, followed by another, randomly selected node acquiring the privileged status, and so on. We shall present these examples, as well as some new examples of short term unfairness in Section~\ref{sec:stu-examples}. Further, we have found that the phenomenon of short-term unfairness is also observed under the practical setting where the backoff sequences are as per the standard, but the \emph{propagation delays} among the participating nodes are large compared to the duration of a backoff slot; this situation arises in a variety of applications such as providing broadband connectivity to remote rural areas using WiFi based \emph{long distance} networks \cite{raman-chebrolu07wifi-rural}, or network formation among UAVs, or between UAVs and a ground station over distances of several kilometres \cite{brown-etal04uav}. Furthermore, with the evolution of WiFi standards, the slot durations\footnote{Throughout the paper, we use the terms ``slot'' and ``backoff slot'' interchangeably.} are decreasing; e.g., the latest WiFi standard IEEE~802.11ac adopts a slot duration of 9~$\mu secs$, as compared to 20~$\mu secs$ in the widely used IEEE~802.11b. Thus, even the propagation delays that were negligible compared to the slot duration in earlier WiFi standards may occupy multiple slot durations in future. In this case also, the analysis in \cite{bianchi00performance} (or simple extensions thereof) does not work well. 

As pointed out above, the situations where the analysis in \cite{bianchi00performance} and \cite{kumar-etal04new-insights} does not work are those where there is significant short term unfairness. The analysis of \cite{bianchi00performance} and \cite{kumar-etal04new-insights} makes the key modeling simplification that, in steady state, during contention periods, the nodes make attempts as \emph{equal rate} independent Bernoulli processes embedded at the backoff slot boundaries. Since the node attempt model is state-independent, such a model does not capture the possibly advantageous position that a successful node might be in, as compared to the unsuccessful nodes, and hence cannot yield the short term unfairness that has been observed. Thus, \emph{a good, parsimonious analytical model to understand and predict the system behavior when the system evolution exhibits high correlation (manifested as short-term unfairness) is still lacking}. Our work is intended as a first step in that direction. In this work, we address this problem for the case of a single-hop topology consisting of \emph{saturated} transmitting nodes and their receivers, using the IEEE~802.11 DCF basic access mechanism for medium access. %We require that the transmitters are equidistant from one another, and also each receiver is equidistant from all the transmitters; note that this includes the case where the propagation delays are negligible.  
We use the theory of Markov Regenerative Processes to develop a tractable generalization of the Bianchi analysis that incorporates general backoff sequences, as well as large propagation delays. Comparison against extensive simulations have shown that the analysis captures the system performance well even in the presence of high correlation in system evolution.

\noindent
\textbf{Summary of contributions}
\vspace{1mm}

\noindent
Based on a study of the evolution of the system, and a stochastic simulation, we find that the phenomenon of short term unfairness in IEEE~802.11 DCF networks renders the state-less, constant attempt rate approach adopted in \cite{kumar-etal04new-insights}, \cite{bianchi00performance}, and later in \cite{simo-reigadas10wild}, inaccurate (see Section~\ref{sec:stu-examples}). In our analytical approach, we maintain some state information, and introduce state-dependent attempt rates. How we do this in a parsimonious and computationally tractable manner, while developing an accurate approximation, is the primary contribution of this work (Section~\ref{sec:mrp-state-dependent}). Furthermore, our analysis can be used to quantify the extent of short term unfairness in a system, and thus allows tuning of the protocol parameters for performance optimization (Section~\ref{sec:optimization-throughput-fairness}). In addition, our results reveal several interesting, sometimes counterintuitive, traits in the dependence of the system performance on propagation delay (Sections~\ref{sec:numerical} and \ref{sec:opt-slot-duration}).

\subsection{Related Work} %and Our Contributions}
\label{subsec:related-work}
%\vspace{-5mm}
There is a considerable body of literature on performance analysis of IEEE~802.11 DCF, starting with the seminal work by Bianchi \cite{bianchi00performance}, which was later generalized by Kumar et al. \cite{kumar-etal04new-insights} to incorporate general backoff parameters. Several extensions have been proposed since then. For example, Jindal and Psounis \cite{jindal09} proposed a throughput analysis for multi-hop IEEE~802.11 networks with non-saturated nodes. Nardelli and Knightly \cite{nardelli-knightley12} proposed a closed form analysis for the saturation throughput in the presence of hidden terminals, but under several simplifying assumptions. Considerable attention has also been given to performance analysis of IEEE~802.11e EDCA; see, for example, \cite{tinnirello-bianchi10edca, ramaiyan-etalYYfp-analysis}, and the references therein. However, none of this work is suitable for predicting the performance of systems that exhibit short term unfairness, and the same has been explicitly pointed out in \cite{tinnirello-bianchi10edca}. We will shed more light on this as we proceed further. 

Short term as well as long term unfairness have been observed (and modeled) before in the presence of \emph{hidden terminals} in WLANs by Garetto et al. \cite{garetto-etal05short-term-unfairness}. However, they assume negligible propagation delay throughout their work, and parts of their analysis rely on the assumption that under no hidden nodes, the system is fair, and existing techniques predict system behavior accurately, which is not quite correct as demonstrated in \cite{ramaiyan-etalYYfp-analysis}, and also our current work. Therefore, in the light of the findings in our current work and in \cite{ramaiyan-etalYYfp-analysis}, the problems in \cite{garetto-etal05short-term-unfairness} need a relook.
   
Rademacher et al.\ \cite{radamacher14wild} attempted to address the problem of large propagation delays in WiFi networks by \emph{making the slot duration at least as large as the propagation delay}, and then using the existing analysis techniques. However, this approach does not provide any insight into the system behavior under the default standard, and is suboptimal in general in terms of throughput (see Section~\ref{sec:opt-slot-duration}). 

Simo-Reigadas et al.\ \cite{simo-reigadas10wild} aimed to develop an extension of the Bianchi model to predict the performance of IEEE~802.11 DCF with non-negligible propagation delays. However, we shall argue in Section~\ref{subsubsec:simo-reigadas-performance} that the analysis in \cite{simo-reigadas10wild} \emph{does not capture two distinct features of such systems}, and as a consequence, the \emph{collision/success probabilities computed using the analysis are inaccurate compared to simulation results} obtained from a detailed stochastic model, as well as from the Qualnet simulator\cite{qualnet}\footnote{This anomaly does not show up significantly in the numerical results presented in their work primarily because they \emph{do not compare the collision/success probabilities obtained from their analysis against any experimental or simulation results}, and provide comparison results only for system throughput, which, as our numerical results later on demonstrate, is \emph{less sensitive to} (but not unaffected by) inaccuracies in the analysis than other performance measures such as collision probability.}.%\footnote{See \cite{techreport} for a discussion on why this anomaly does not show up significantly in the numerical results presented in their work.}.

Our work is thus intended as a first step towards an accurate analytical model for such systems. Our key contribution is the development of a principled approach for analyzing IEEE~802.11 DCF based systems with short term unfairness.

\vspace{2em}
\begin{center}
 \textbf{\Large{Part~I: Generalized Backoff Sequences}}
\end{center}
\normalsize
\section{A brief description of IEEE~802.11 DCF (adapted from \cite{kumar-etal04new-insights})}
\label{subsec:dcf-description}

\begin{figure}[ht]
  \begin{center}
  \includegraphics[scale=0.45]{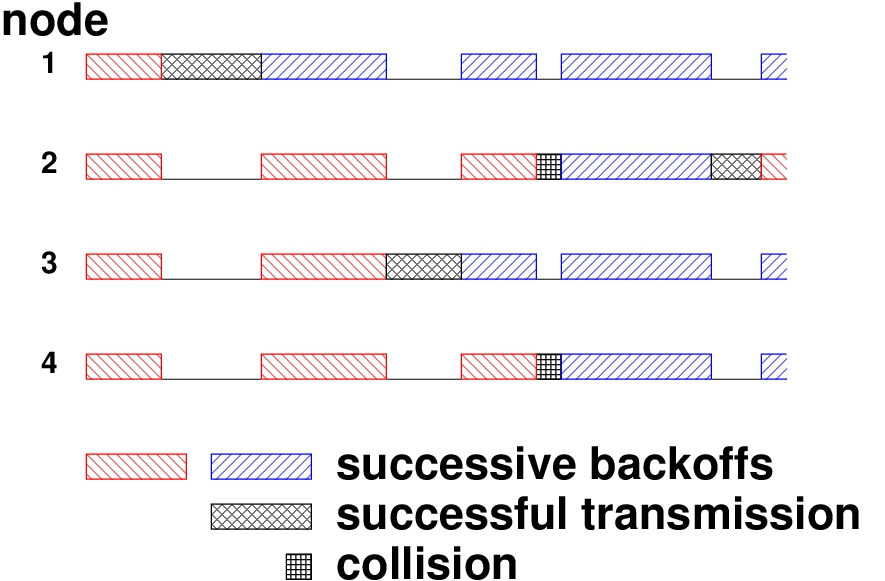}  
  \end{center}
\caption{(Reproduced from \cite{kumar-etal04new-insights}) The evolution of the backoff periods and channel activity for four
  nodes. Backoffs are interrupted by channel activity, i.e., packet
  transmissions and collisions.}
\label{fig:evolution_abstraction}
 \end{figure} 

  \begin{figure}[ht]
  \begin{center}
  \includegraphics[scale=0.45]{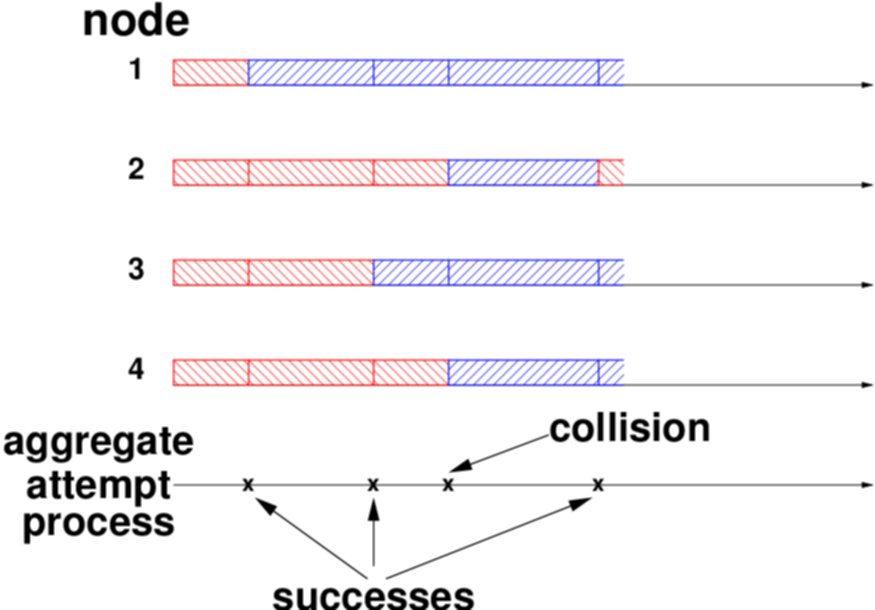}  
  \end{center}
\caption{(Reproduced from \cite{kumar-etal04new-insights}) After removing the channel activity from Figure~\ref{fig:evolution_abstraction}
  only the backoffs remain. At the bottom is shown the aggregate
  attempt process on the channel, with three successes and one
  collision. }
\label{fig:backoff_abstraction}
\end{figure}
We assume basic access without RTS-CTS. Figure~\ref{fig:evolution_abstraction} is a depiction of the the evolution of the
system for 4 nodes; shown are the backoffs, the successful transmissions and
collisions (including overheads). In the IEEE~802.11 standard, the backoff durations are in
multiples of a standardized time interval called a \emph{backoff slot} (e.g.,
$20~\mu$s in IEEE~802.11b). When a node completes its backoff (for example, Node~1 is the first to complete
its backoff in Figure~\ref{fig:evolution_abstraction}), if it senses the channel idle, it starts a packet transmission on the channel. If none of the other nodes finish their backoff simultaneously, they hear the ongoing transmission, and freeze their backoffs. Note that we assume that all the nodes can hear one another's transmissions, i.e., there are \emph{no hidden nodes}. In this case, the packet transmission is successful, and the intended receiver sends a MAC level ACK. Upon receiving the ACK, the node that transmitted the packet waits for an interval called DIFS, and samples a new backoff
interval. All the other nodes resume counting down their old residual backoffs. Note that we assume throughout that nodes always have
packets to transmit; i.e., \emph{all the transmission queues are saturated}.

If two or more nodes complete their backoffs together, then they both start a packet transmission at the same time, leading to a collision (note that the phenomenon of packet capture is not modeled). At the end of the collision duration, each node involved in the collision waits for an interval called EIFS, before sampling fresh backoffs. For example, in Figure~\ref{fig:evolution_abstraction} Nodes 2
and 4 collide after the first two attempts (by Nodes 1 and 3,
respectively) are successful. The other nodes, not involved in the
collision, freeze their backoff during the collision duration (including the EIFS), and resume counting down their old residual backoffs thereafter. If attempts to send the packet at the head-of-the-line (HOL) meet with several successive failures, this packet is discarded. Note that the evolution of the channel activity after an attempt is deterministic. It is either the time taken for a successful transmission or for a collision. 

In the DCF mechanism, the nodes sample their backoff intervals uniformly from a contention window. The initial contention window size is typically small, and after each collision, the subsequent backoff is sampled from a larger contention window. For example, in the IEEE~802.11 standard, the initial contention window is the interval $[1,32]$, and after each collision, the contention window size is doubled, until it reaches a maximum allowed value of 1024. After a successful transmission, the contention window size is reset to the initial value. Throughout this work, we shall assume a \emph{homogeneous} system, i.e., all the nodes have the same backoff parameters (the contention window size, how these are varied in
response to collisions and successes, and the number of retries of a
packet).

\section{Modeling and Analysis of IEEE~802.11 DCF as in \cite{kumar-etal04new-insights} (adapted from \cite{kumar-etal04new-insights})}
\label{sec:modeling-dcf}

Since all nodes freeze their backoffs during channel activity, the
total time spent in backoff up to any time $t$, is the \emph{same}
for every node. With this observation, let us now look at
Figure~\ref{fig:backoff_abstraction} which shows the backoffs of
Figure~\ref{fig:evolution_abstraction} with the channel activity
removed. Thus in this picture ``time'' is just the cumulative backoff
time at each node. In the IEEE~802.11 standard the backoffs are
multiples of the slot time. A success occurs if a single backoff ends
at a slot boundary, and a collision occurs when two or more backoffs
end at a slot boundary. Clearly, the (random) sequence in which the
nodes seek turns to access the channel and whether or not each such
attempt succeeds depends only on the backoff process shown in
Figure~\ref{fig:backoff_abstraction}. It is therefore sufficient to
analyse the backoff process in order to understand the channel
allocation process. We must caution, however, that this simplification of working in backoff time alone will not work if the propagation delays among the nodes are large compared to the duration of a backoff slot; see Part~II for details. 

Thus, in summary, we can delete the channel activity periods, and we
are left with a ``conditional time'' , called \emph{backoff time}. In \cite{kumar-etal04new-insights}, Kumar et al./ analyze this backoff process conditioned on being in
backoff time. Note that since the channel activity durations are deterministic, the original process (in unconditional time) can be exactly reconstructed from the backoff process.

\subsection{A Markov model for system evolution}
\label{subsec:markov-model-dcf}
Let us introduce the following notation for the protocol parameters.
\begin{description}
 \item $K$: The maximum number of reattempts before a packet is discarded
 \item $W_k$: Contention window size for the $k^{th}$ reattempt
 \item $b_k$: Mean backoff duration for the $k^{th}$ reattempt. Note that for uniform backoff distribution, $b_k=\frac{1+W_k}{2}$.
\end{description}
Suppose there are $n$ transmitter-receiver pairs, with saturated queues. As mentioned in the foregoing discussion, we look at the system evolution over backoff time alone; recall Figure~\ref{fig:backoff_abstraction}. We adopt a discrete time model by focusing on the system evolution over backoff slots, $t \in \{0, 1, 2, \cdots\}$. Let $S_i(t)$ be the backoff stage of Node~$i$ in slot $t$, i.e., the number of reattempts so far for the current HOL packet at Node~$i$; $S_i(t) \in \{0, 1, \cdots , K\}$. Let $B_i(t)$ be the residual backoff of Node~$i$ in slot $t$; $B_i(t)\in\{1,\ldots,W_{S_i(t)}\}$. Then it follows from the foregoing discussion of the IEEE~802.11 DCF protocol that $\{(S_i(t),B_i(t))_{i=1}^n, t \geq 0\}$ is a Discrete Time Markov chain (DTMC) embedded at the backoff slot boundaries. However, the size of the state space of this DTMC is $(W_0+W_1+\ldots+W_K)^n$, i.e., growing exponentially with the number of nodes. For the default protocol parameters of IEEE~802.11b, the state space size is over 9 million, even for $n=2$, thus making this DTMC analytically intractable.

An alternative Markov model was also proposed in \cite{kumar-etal04new-insights} \emph{by assuming a geometric backoff distribution} instead of the uniform distribution adopted in the standard. In particular, the assumption is that when a node is in back-off stage $k$, it attempts in the next slot with probability $\frac{1}{b_k}$. With this assumption, let us look at the process that counts the number of nodes in each back-off stage. This will be a $(K+1)$-dimensional process
for any number of nodes. Define the number of nodes in the back-off
stage $k \in \{0, 1, \cdots, K\}$ in slot $t$
to be $Y^{(n)}_k(t)$.  Let $\mathbf{Y}^{(n)} (t)$ denote the vector
random process with components $Y^{(n)}_k(t)$. Then, due to the assumption of Bernoulli attempt processes, $\mathbf{Y}^{(n)} (t)$ is a Markov process taking
values in the set \( {\cal Y}^{(n)} := \{\mathbf{y}: y_k \ 
\mbox{nonnegative integers}; \sum_{k=0}^K y_k = n\} .  \)

It was shown in \cite{kumar-etal04new-insights} that the DTMC
$\mathbf{Y}^{(n)} (t)$ is positive recurrent, and hence has a unique stationary distribution, which, in principle, can be obtained, and the system performance measures computed therefrom. However, the state space size of even this DTMC is $\binom{n+K}{K}$, which quickly becomes unwieldy as $n$ and $K$ increase. 

Since an exact analysis of the system evolution for the DCF mechanism is computationally intractable, approximate analytical techniques were developed to predict the system performance with reasonable accuracy. We describe next, the approximate analysis developed in \cite{kumar-etal04new-insights}, which was a generalization of the seminal work \cite{bianchi00performance} by Bianchi.

\subsection{The approximate analysis in \cite{kumar-etal04new-insights} (adapted from \cite{kumar-etal04new-insights})}
\label{subsec:bianchi-analysis}
We start with the following key approximation.

\noindent
\textbf{The Decoupling Approximation:} Let $\beta$ denote the long run
average back-off rate (\emph{in back-off time}) for each node. By symmetry, all nodes achieve the same value of $\beta$.  Let
there be $n$ contending transmitters, and consider a tagged node. The
decoupling approximation is to assume that the aggregate attempt
process of the other $(n-1)$ nodes is independent of the back-off
process of the tagged node. Then the overall approach is the following: 

\noindent
(i) \emph{Modeling the evolution at a tagged node: }The
``influence'' of the other nodes on a tagged node is modeled via the
decoupling approximation. Attempts by a tagged node over slots
experience the collision probability $\gamma$. For a given collision
probability this yields one equation $\beta = G(\gamma)$ (see
Eqn.~\ref{eqn:G_gamma}). 

\noindent
(ii) \emph{Modeling the system evolution: }The nodes are assumed to attempt in each
slot with a constant (state independent) probability equal to the
average attempt rate, $\beta$. Then, conditional on a tagged node
attempting, the number of attempts by other nodes is binomially
distributed. This yields the other (``coupling'') equation $\gamma =
\Gamma(\beta)$ (see Eqn.~\ref{eqn:Gamma_beta_binomial}).  When these
equations are put together we obtain the desired fixed point equation.

A justification for the decoupling approximation comes from a Mean Field type asymptotic analysis. Please see Section~\ref{sec:mean-field} for details.

\subsubsection{Analysis of the backoff process at a tagged node}
\label{sec:back-off_process}

\begin{figure}[t]
  \begin{center}
    \includegraphics[scale=0.45]{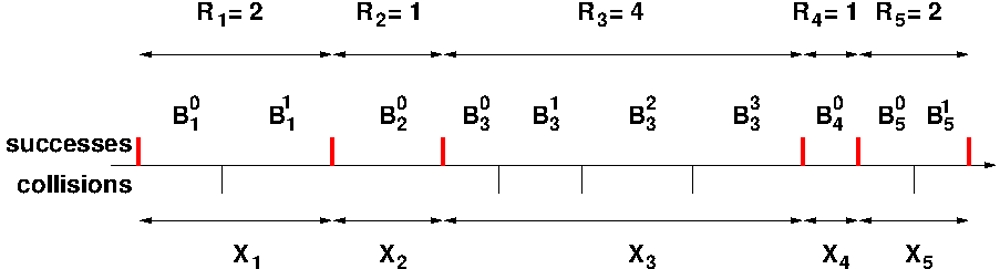}
  \end{center}
\caption{(Reproduced from \cite{kumar-etal04new-insights}) Evolution of the back-offs of a node. Each attempted packet starts
  a renewal ``cycle.''}
\label{fig:renewal_model}
\end{figure}
In Figure~\ref{fig:renewal_model} we show the evolution of the
back-off process for a single node. There are $R_j$ attempts until
success for the $j$th packet (no case of a discarded packet is shown
in this diagram), and the sequence of back-offs for the $j$th packet
is $B_j^{(i)}, 0 \leq i \leq R_j-1$. Thus the total back-off for the
$j$th packet is given by $ X_j = \sum_{i=0}^{R_j-1} B_j^{(i)} $ with
$\EXP{B_j^{(i)}} = b_{i}$. We observe that the sequence $X_j, j \geq
1, $ are renewal life times, since after each success or packet discard, the node returns to backoff stage 0. Hence, viewing the number of attempts
$R_j$ for the $j$th packet as a ``reward'' associated with the renewal
cycle of length $X_j$, we obtain from the renewal reward theorem that
the back-off rate is given by ${\EXP{R}}/{\EXP{X}}$.  Now let $\gamma$
be the collision probability seen by a node, i.e.,
\[\gamma := \prob{\mbox{an attempt by a node fails because of a collision}} \]
By the approximation made earlier, the successive collision
events are independent. It is then easily seen that
\begin{eqnarray*}
  \EXP{R} &=& 1 + \gamma  + \gamma^2 +  \cdots + \gamma^K  \\
  \EXP{X} &=& b_0 + \gamma b_1 + \gamma^2 b_2 + \cdots +  \gamma^k b_k 
             + \cdots +  \gamma^K b_K
\end{eqnarray*}
which yields the following formula for the attempt rate for a given 
collision probability $\gamma$.
\begin{eqnarray}
  \label{eqn:G_gamma}
 G(\gamma)  &:=& \frac{1 + \gamma  + \gamma^2  + \cdots + \gamma^K}{
  b_0 + \gamma b_1 + \gamma^2 b_2 + \cdots + \gamma^k b_k + \cdots +
  \gamma^K b_K}  
\end{eqnarray}
Note that, since the back-off times are in slots, the attempt rate
$G(\gamma)$ is in \emph{attempts per slot}.

\subsubsection{The fixed point equation}
\label{sec:fixed_point_equation}

Since each node attempts with probability $\beta, 0 \leq \beta
\leq 1$, in each slot, conditioning on an attempt of a given node, the
probability of this attempt experiencing a collision is the
probability that any of the other nodes attempts in the same slot. Thus, 
under the decoupling approximation, the probability of collision of an
attempt by a node is given by
\begin{eqnarray}
   \label{eqn:Gamma_beta_binomial}
   \Gamma(\beta) &:=& 1 - (1 - \beta)^{(n-1)}
\end{eqnarray}
which, for a large number of nodes, can be approximated by (see \cite{kumar-etal04new-insights} for details)
\begin{eqnarray}
  \label{eqn:Gamma_beta_poisson}
  \Gamma(\beta) &:=& 1 - e^{-(n-1) \beta} 
\end{eqnarray}
Thus, we have the following fixed
point equation, which is expected to approximate the equilibrium behavior of the system.
\begin{eqnarray}
  \label{eqn:fp_equation}
  \gamma &=& \Gamma( G (\gamma))
\end{eqnarray}

\noindent
\remarks
\begin{enumerate}
 \item It was shown in \cite{kumar-etal04new-insights} that $\Gamma( G (\gamma)): [0,1] \to [0,1]$, has a unique fixed point
    if $b_k, k \geq 0$, is a nondecreasing sequence, which is, in fact, the case for the IEEE~802.11 standard.
 \item The \emph{distribution} of the back-off durations does
not matter in the above analysis. \hfill \Square
\end{enumerate}
\begin{figure}[tpb]
  \begin{center}
     
    \includegraphics[scale=0.35]{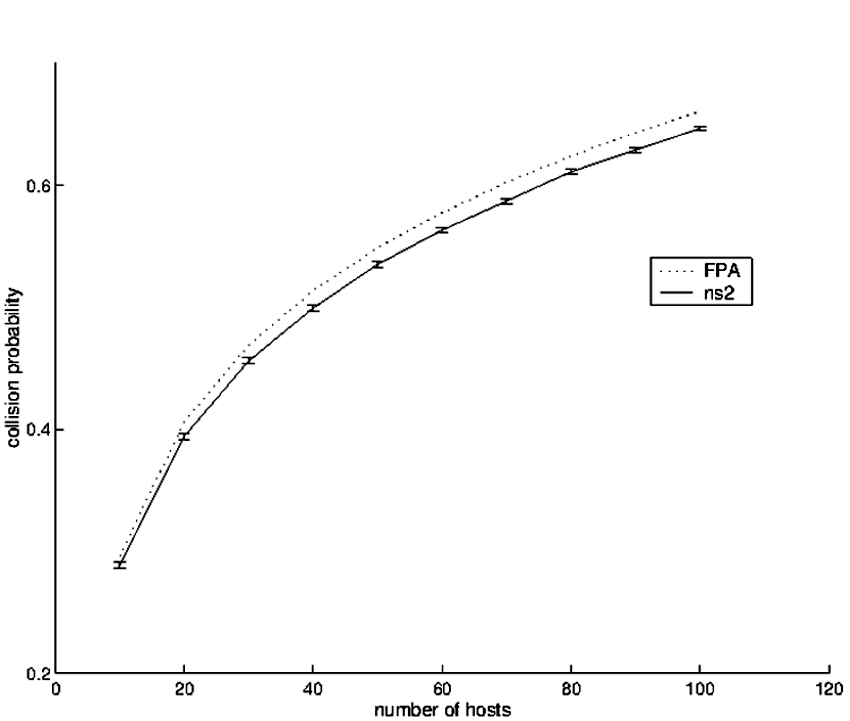}
  \end{center}
\caption{(Reproduced from \cite{kumar-etal04new-insights}) Plot of collision probability versus number of nodes.
  Comparison of collision probability ($\gamma$) obtained from an
  \emph{ns2} simulation (plot labeled \emph{ns2}), and the fixed point
  analysis (plot labeled \emph{FPA}). In the \emph{ns2} simulation the default IEEE~802.11 parameters are
  used}
\label{fig:comparison_collision_ns2_fp}
\end{figure}
Figure~\ref{fig:comparison_collision_ns2_fp} shows the collision probabilities obtained from the fixed point
method and from an \emph{ns2} simulation for a wide range of values of $n$, and for the default parameters of IEEE~802.11. It can be seen that the fixed
point analysis works well for the default IEEE~802.11 parameters even for moderate values of $n$. 

\section{A Mean Field Perspective}
\label{sec:mean-field}
A mean field type asymptotic approximation has been used in the literature in an attempt to understand the scope and limitations of the fixed point analysis proposed in \cite{bianchi00performance,kumar-etal04new-insights}; see, for example, \cite{bordenave-etal05,benaim-leboudec08}. We provide here, a brief overview of the mean field approach; for details of the approach, see, for example, \cite{benaim-leboudec08}. 

\subsection{An overview of the mean field asymptotic approximation}
\label{subsec:scaled-model}
Consider $N$ identical, saturated nodes contending for medium access; the propagation delay among the nodes is negligible, and there are no hidden nodes. The system is slotted. This model is entirely in backoff time (see Section~\ref{sec:modeling-dcf}, and \cite{kumar-etal04new-insights}). Each node's attempt and backoff model is as follows: a node can make up to $K$ reattempts. After each unsuccessful attempt, its \emph{backoff stage} is incremented by one (after the $K^{th}$ reattempt, it is reset to zero); thus the \emph{backoff stage} of a node $\in \mathsf{S} := \{ 0, 1, 2,\cdots, K\}$. In back-off stage $k$, at the beginning of a slot, a node attempts with probability $\frac{p_k}{N}$, independent of everything else. If two or more nodes attempt, there is a collision. 

Let $X^{(N)}_n (i) \in \mathsf{S}, i \geq 0, $ be the state of Node~$n$ at the beginning of Slot~$i$. Then it is easy to observe that $(X^{(N)}_1(k), \cdots, X^{(N)}_N(k)), k \geq 0,$ is an irreducible DTMC on $\mathsf{S}^N$ for each $N \geq 1$. However, the state space of this DTMC grows exponentially in $N$, and is therefore, computationally intractable. Hence, instead, the following method is adopted. It can be argued that (see \cite{benaim-leboudec08}) starting with an exchangeable law at $k=0$, $(X^{(N)}_1(k), \cdots, X^{(N)}_N(k))$ is exchangeable for each $k$. Consider the \emph{empirical measure} process, defined for each $N$, and $k \geq 0$, as follows:
  \begin{align*}
    \mathbf{M}^{(N)}(k) &= \frac{1}{N} \sum_{l=1}^N \mathbf{e}_{X^{(N)}_l(k)}
  \end{align*}
where $\mathbf{e}_i$ is the $i^{th}$ unit vector in $\mathbf{R}^{K+1}$. Thus, $(\mathbf{M}^{(N)}(k))_i$ is the fraction of particles in state $i$ at step $k$, $0 \leq i \leq K$. Clearly, $\mathbf{M}^{(N)}(k)$ is a DTMC on $\mathcal{P}(\mathsf{S})$, the set of probability measures on $\mathsf{S}$. 

\subsubsection{An ordinary differential equation (ODE) limit for the time scaled empirical measure process}
\label{subsubsec:ode}
For $t \geq 0$, the above construction, with the scaled attempt probabilities, is used to define a time scaled version of the empirical measure process as follows: 

\begin{eqnarray}
  \overline{\mathbf{M}}^{(N)} (t) := \mathbf{M}^{(N)} (\lfloor Nt \rfloor)\label{eqn:interpolated-empirical-measure}
\end{eqnarray}

At this point, it is worth recalling the original IEEE~802.11 system with \emph{geomtric backoff distribution} introduced in Section~\ref{subsec:markov-model-dcf}. To see the connection of the current model with the original 802.11 DCF system analyzed earlier, think of the process $\overline{\mathbf{M}}^{(N)}(t)$ intuitively as follows. Each backoff slot is divided into $N$ ``mini-slots'', and in each mini slot, each node in backoff stage $k$, $0\leq k\leq K$, attempts w.p. $\frac{p_k}{N}$, independent of everything else; we set $p_k=\frac{1}{b_k}$. It is as if, in each mini-slot, each node ``chooses'' to attempt with probability $\frac{1}{N}$, and then, having chosen to attempt, actually attempts with probability $p_k$, if its back-off stage is $k$. Thus, the expected number of times that a node chooses to attempt in a slot is 1, and the expected number of attempts that a node actually makes in a slot is $p_k=\frac{1}{b_k}$, the same as in the original system. The process $\mathbf{M}^{(N)}(i), i \geq 0,$ is the empirical measure process for this scaled process, embedded at mini-slots. $\overline{\mathbf{M}}^{(N)} (t)$, defined in Eqn.~\ref{eqn:interpolated-empirical-measure}, is then just the step interpolation of $\mathbf{M}^{(N)}(i), i \geq 0$. It is a continuous time random process, taking values in $\mathcal{P}(\mathsf{S})$.

It can be shown that (see \cite{benaim-leboudec08}), if $\mathbf{M}^{(N)} (0) \stackrel{p}{\to} \mathbf{\mu}$ then, for each $t \geq 0$,
\begin{eqnarray*}
   \overline{\mathbf{M}}^{(N)} (t)  \stackrel{w}{\to} \mathbf{\mu}(t)
\end{eqnarray*}
i.e., the scaled and interpolated empirical measure Markov chain converges weakly to the deterministic function $\mathbf{\mu}(t)$ as $N \to \infty$, where $\mathbf{\mu}(t)$ is the (unique) solution to the following Ordinary Differential Equation (ODE) on $\mathcal{P}(\mathsf{S})$ with initial condition $\mathbf{\mu}(0) = \mathbf{\mu}$
\begin{align}
  \dot{\mu_0} (t) &= -\mu_0p_0 + \mathbf{p}\cdot\mathbf{\mu}e^{-\mathbf{p}\cdot\mathbf{\mu}} + p_K\mu_K (1-e^{-\mathbf{p}\cdot\mathbf{\mu}})\nonumber\\
  \dot{\mu_i} (t) &= -\mu_i p_i + \mu_{i-1} p_{i-1}(1-e^{-\mathbf{p}\cdot\mathbf{\mu}})\quad i=1,\ldots,K \label{eqn:ode}
\end{align}
where $\mathbf{p}=[p_0\:p_1\ldots p_K]$.
 
A formal derivation of the ODE requires considering the expected drift of the process $\mathbf{M}^{(N)}(k)$, and taking limits of an appropriately scaled version of this drift as $N\to\infty$. However, one can intuitively interpret the equations as follows: a node in backoff state 0 leaves state 0 if it makes an attempt; thus the rate of leaving state 0 is $N\mu_0 \frac{p_0}{N}=\mu_0p_0$. A node in state $i\neq 0,K$ enters state 0 if its attempt is successful. Thus the rate of entering state 0 from state $i\neq 0,K$ is $N\mu_i\frac{p_i}{N}(1-\frac{p_i}{N})^{N\mu_i-1}\prod_{k\neq i}(1-\frac{p_k}{N})^{N\mu_k} \to \mu_ip_i e^{-\mathbf{p}\cdot\mathbf{\mu}}$ as $N \to \infty$. If a node in state $K$ makes an attempt, it enters state 0, irrespective of success or collision. Thus, the rate of entering state 0 from state $K$ is $\mu_Kp_K$. Combining all these, we get the R.H.S. of the first equation. Interpretation for the expression for $\dot{\mu_i} (t)$ is similar.

\subsubsection{Convergence to ``chaos:'' A motivation for the ``decoupling'' approximation}
\label{subsubsec:decoupling-approx-from-ode}
Denote by $\mathcal{L}( X^{(N)}_1 (\infty), X^{(N)}_2 (\infty), \cdots, X^{(N)}_k (\infty))$, the joint probability law of any $k$ nodes in the steady state regime. Suppose the ODE \eqref{eqn:ode} has a unique stationary point $\mu^*$ to which \emph{all trajectories of the ODE converge} (also called the globally asymptotically stable equilibrium (g.a.s.e.) of the ODE). Then, it can be shown that (see \cite{benaim-leboudec08}), as $N\to\infty$,
\begin{eqnarray*}
\mathcal{L}( X^{(N)}_1 (\infty), X^{(N)}_2 (\infty), \cdots, X^{(N)}_k (\infty))
  &\stackrel{w}{\to}& (\mu^*)^k
\end{eqnarray*} 
i.e., the stationary joint probability law of the backoff states any $k$ nodes in the steady state regime is approximately $(\mu^*)^k$ for large $N$. In other words, for large $N$, in steady state, the time scaled empirical measure process is approximately an i.i.d. vector across the nodes, with common marginal measure $\mu^*$. This motivates the ``decoupling approximation'', and provides a justification for the independence assumption (of node attempt processes) made in the saturation analysis of IEEE~802.11 (Section~\ref{sec:modeling-dcf}).  

\subsection{A justification of the Bianchi approximation from mean field perspective}
\label{subsec:critique-bianchi}
To find the stationary point of the ODE \eqref{eqn:ode}, we need to solve
\begin{eqnarray*}
  p_0 \mu_0 &=& \beta(\mathbf{\mu}) (1 - \gamma(\mathbf{\mu})) + p_K \mu_K \gamma(\mathbf{\mu})\\
\end{eqnarray*}
and,  for $1 \leq k \leq K$,
\begin{eqnarray*}
    p_k \mu_k &=& p_{k-1} \mu_{k-1} \gamma(\mathbf{\mu})
\end{eqnarray*}
where $\beta(\mathbf{\mu})=\mathbf{p}\cdot\mathbf{\mu}$, the total attempt rate of the nodes in a minislot, and $\gamma = 1-e^{-\beta}$. This, in turn, gives the following fixed-point equation after some algebraic manipulations:
\begin{align*}
 \gamma &= (1 - e^{-\beta})  \\
  \beta  &= \frac{\sum_{k=0}^K \gamma^k}{\sum_{k=0}^K \frac{\gamma^k}{p_k}}
\end{align*}
which is of the same form as the fixed point equation in the Bianchi approximation (Section~\ref{sec:fixed_point_equation}), and is known to have a unique solution. Observe that $\gamma$ still has the interpretation of collision probability. To see this, note that for the $N$-node system, the collision probability of a node is given by $1-\prod_{k=0}^K (1-\frac{p_k}{N})^{N\mu_k} \to 1-e^{-\mathbf{p}\cdot\mathbf{\mu}}=1-e^{-\beta}=\gamma$ as $N\to\infty$.

Thus, Bianchi's method amounts to finding the unique stationary point of the ODE, which is then taken as the steady state operating point of the system. However, we need to exercise some caution.

\begin{enumerate}
 \item From the discussion in Section~\ref{subsubsec:decoupling-approx-from-ode}, the asymptotic independence (which motivates the decoupling approximation) \emph{provably holds} when the stationary point is also the globally asymptotically stable equilibrium of the ODE \cite{benaim-leboudec08}. Uniqueness of the stationary point alone is not enough to ensure that. When the stationary point is not a g.a.s.e. of the ODE (for example, when the ODE has a limit cycle), the stationary point may not represent the equilibrium behavior of the ODE (\cite{benaim-leboudec08}). 
 \item Even when the ODE has a g.a.s.e., the asymptotic independence is, after all, only an asymptotic result that holds for large $N$. Thus, for moderate values of $N$, the independence assumption made in the analysis in Section~\ref{sec:modeling-dcf} might not hold. The accuracy of the Bianchi analysis for default backoff parameters of the IEEE~802.11 standard appears to indicate that the decoupling approximation works well even for small values of $N$. However, in the next section, we provide examples of backoff sequences that result in high correlation in the system evolution for small to moderate values of $N$, thus violating the decoupling approximation. 
\end{enumerate}
From here onwards, we refer to the analysis in Section~\ref{sec:modeling-dcf} as Bianchi analysis or mean field analysis interchangeably.

\section{Systems with Short Term Unfairness and the Bianchi Analysis}
\label{sec:stu-examples}
The DCF mechanism is finding its way to newer applications beyond the WLAN standards, thanks in large part to its simple, distributed implementation, and the advent of inexpensive chipsets. For example, the IEEE~802.15.4 MAC protocol used in IoT applications is a simple variation of the DCF protocol of IEEE~802.11. Also, there have been proposals to use the IEEE~802.11 DCF for rural broadband access, and for Unmanned Aerial Vehicle (UAV) communications. When variations of the same protocol are used for a wide range of different applications, a common engineering practice is to adapt the parameters of the protocol to suit the needs of the particular application at hand. For example, the backoff parameters of the IEEE~802.15.4 are quite different from those of the IEEE~802.11 standard. It is, therefore, convenient to have an analytical technique that can predict the system performance not just for the standard protocol parameters, but for more general backoff parameters as well.

This brings up the following natural question: will the mean field analysis continue to predict the system performance well, if the protocol parameters are changed from those in the IEEE~802.11 standard? In particular, will it work for any \emph{non-decreasing} backoff sequence $\{b_0,\ldots, b_K\}$ (recall from Section~\ref{sec:fixed_point_equation} that the mean field analysis has a unique fixed point for non-decreasing backoff sequences), and any number of nodes, $n$? The following examples demonstrate that this is not the case. 

\subsection{Example~1: IEEE~802.11-like backoff expansion framework (adapted from \cite{ramaiyan-etalYYfp-analysis})}
\label{subsec:example1} 
\begin{figure}[t]
%\footnotesize
\begin{center}
\includegraphics[scale=0.35]{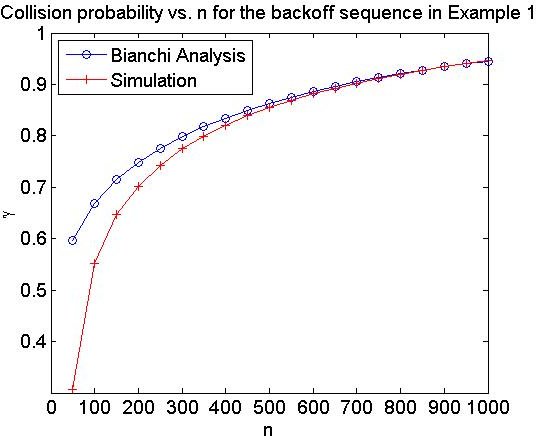}
\caption{Example~1: Collision probability vs. number of nodes; comparison of the values obtained from the Bianchi analysis against those obtained from simulations. We see that the error in the collision probability obtained from the Bianchi analysis is much worse than 10\% when the number of nodes, $n$, is less than 100. Note that in a practical network, the number of nodes could just be in the tens.}
\label{fig:gamma-vs-n-example1}
\vspace{-5mm}
\end{center}
\normalsize
%\vspace{-6mm}
\end{figure}
\begin{figure}[t]
%\footnotesize
\begin{center}
\includegraphics[scale=0.3]{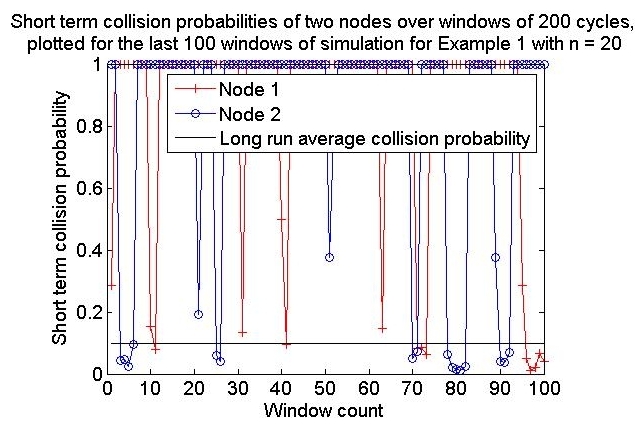}
\hspace{0.1mm}
\includegraphics[scale=0.3]{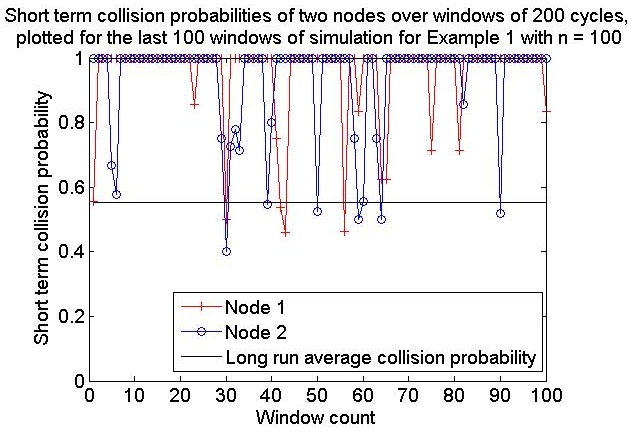}
\vspace{0.1mm}
\includegraphics[scale=0.3]{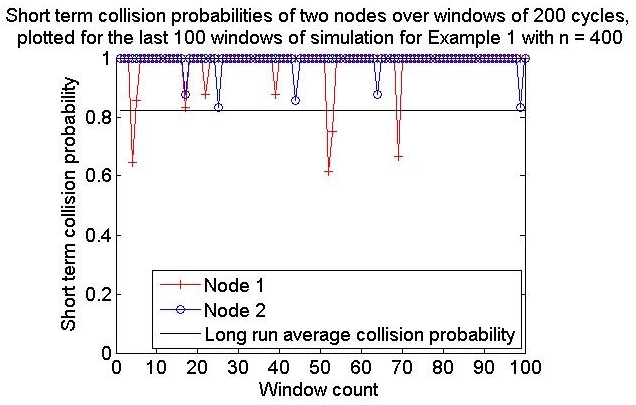}
\hspace{0.1mm}
\includegraphics[scale=0.3]{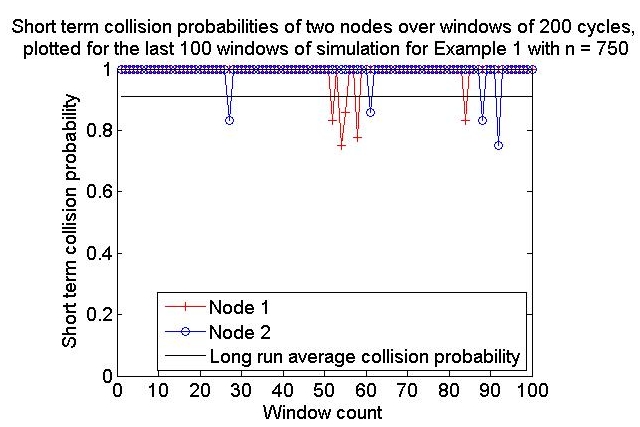}
\caption{Example~1: Simulation results depicting short term unfairness for various $n$. Shown are the short term collision probabilities of two of the transmitters; also plotted are the long run average collision probabilities, averaged over all the nodes and all simulation time. We see that short term unfairness decreases as the number of nodes, $n$, increases.}
\label{fig:short-term-unfairness-example1}
\vspace{-5mm}
\end{center}
\normalsize
%\vspace{-6mm}
\end{figure}
Consider a system where all nodes use the IEEE~802.11 DCF backoff expansion framework for medium access, but with parameters different from those in the standard, namely, $K=7$, $b_0=1$, $b_k = 3^k b_0$ for all $0\leq k\leq K$. This system is of interest because of its close resemblance with the standard IEEE~802.11 DCF backoff mechanism, except for the values of the protocol parameters; specifically the retry limit $K$ has been changed to 7 from 6 in the standard, the initial mean backoff $b_0=1$ instead of 16.5 in the standard, and the backoff multiplier has been changed to 3 from 2 in the standard.

Figure~\ref{fig:gamma-vs-n-example1} demonstrates the performance of the Bianchi analysis in predicting the collision probabilities for this example for various $n$. As we can see from the plot, the error in the collision probability obtained from the Bianchi analysis is much worse than 10\% when the number of nodes, $n$, is less than 100. 

To understand why the mean field analysis does not capture the system performance, let us take a closer look at the system behavior for lower values of $n$. Consider a system with $n=20$ nodes, and backoff parameters as above. It turns out that this system exhibits \emph{short term unfairness}, in the sense that when a node's transmission is successful, it monopolizes the channel for the next several thousands of backoff slots, resulting in starvation and high short term collision probabilities for the other nodes \cite{ramaiyan-etalYYfp-analysis}. 

Panel~1 (panel numbers are row-wise, from left to right) of Figure~\ref{fig:short-term-unfairness-example1} depicts the short term collision probabilities of two of the 20 transmitters. Each point in the plot is the short term collision probability of a node computed over a window of 200 consecutive system activities (i.e., successful transmissions or collisions), and the process was repeated for the last 100 windows in the simulation, thus giving 100 values for each node. The short term collision probability of Node~$i$ in Window~$j$ is computed as $\frac{C_j(i)}{A_j(i)}$, where $C_j(i)$ and $A_j(i)$ are respectively the number of collisions experienced, and the number of attempts made by Node~$i$ in Window~$j$. Also plotted is the long run average collision probability, averaged over all the nodes, and the simulation duration. This is given by $\frac{1}{n}\sum_{i=1}^n\frac{C(i)}{A(i)}$, where $C(i)$ and $A(i)$ are respectively the total number of collisions experienced, and the total number of attempts made by Node~$i$ over the entire simulation duration. It can be observed from the plot that there is high variance in the short term collision probabilities of the two nodes w.r.t the long run average collision probability. In particular, it is often the case that in a window where Node~1 has a low short term collision probability, Node~2 has a very high short term collision probability, and vice-versa, thus indicating that one of the nodes monopolizes the channel in each window, shutting out the other node, thus leading to a high collision probability for the other node during that period.

Intuitively, the short term unfairness in this system can be explained as follows: when a node succeeds, it attempts again in the immediate next slot (since the initial backoff window is only 1 slot), whereas due to the large variability in backoff, the other nodes are busy counting down their large residual backoffs. This causes the successful node to monopolize the channel (attempt in every slot). See \cite{ramaiyan-etalYYfp-analysis} for more details. 

This also explains why the collision probability predicted by the Bianchi analysis is higher than that obtained from simulations. This is because in the presence of short term
unfairness, the last successful node has a much larger probability of accessing the channel in the next slot than the other nodes, thus further boosting its success probability, unlike in a fair system, where all the nodes have comparable probability of accessing the channel, resulting in a higher probability of collision. \emph{The mean
field analysis ignores the correlation in the system evolution in an unfair system.} The high correlation in the system evolution (manifested as short term unfairness) means that the asymptotic independence yielded by the mean field approach in Section~\ref{subsec:scaled-model} (and hence the decoupling approximation made in the Bianchi analysis) does not hold, which explains why the analysis does not work. 

Figure~\ref{fig:short-term-unfairness-example1} also demonstrates the variation in short term unfairness as a function of the number of nodes, $n$ (see Panels~2, 3, and 4). It can be seen that as the number of nodes increases, the variance in the short term collision probabilities w.r.t. the long run average collision probability decreases, implying fairer access to the channel for all the nodes, i.e., the short term unfairness gradually decreases. This is consistent with the fact that the Bianchi analysis (and the decoupling approximation) works well for larger $n$. 

The decrease in short term unfairness with increasing $n$ can be intuitively explained as follows. The successful node goes to backoff stage~0, where it attempts again with probability 1 in the very next slot. The other nodes have large backoffs and hence the probability of any individual node attempting in the same slot as the successful node is small. However, if there are enough of other nodes (i.e., $n$ is sufficiently large) then the probability of the successful node colliding in its next attempt can be large, thereby causing that node as well to quickly join the ranks of the nodes with large backoffs, thus ameliorating the unfairness. 

These observations are further reinforced in the next examples.

\subsection{Example~2: Large backoff variability (adapted from \cite{ramaiyan-etalYYfp-analysis})}
\label{subsec:example2}
\begin{figure}[htpb]
%\footnotesize
\begin{center}
\includegraphics[scale=0.3]{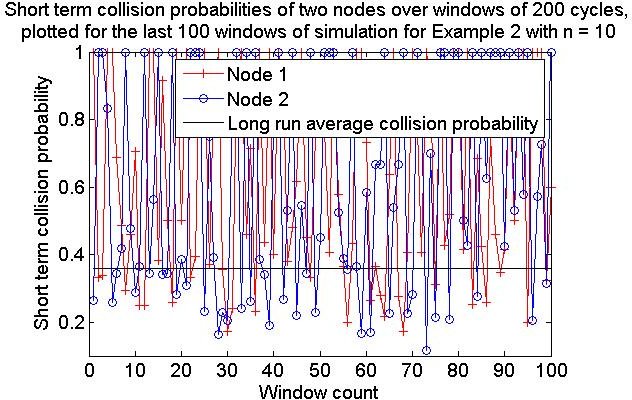}
\hspace{0.1mm}
\includegraphics[scale=0.3]{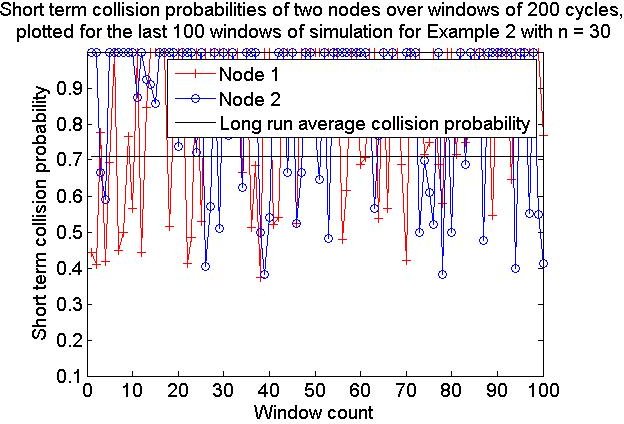}
\vspace{0.1mm}
\includegraphics[scale=0.3]{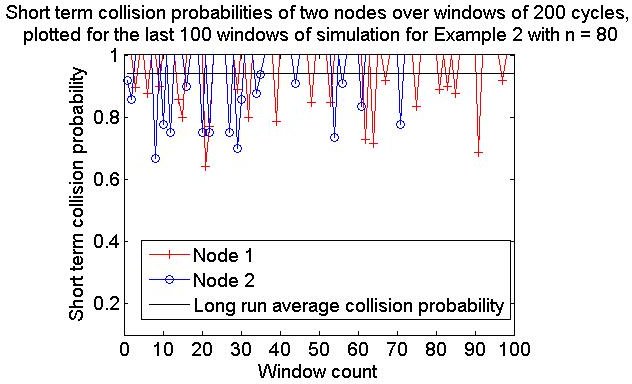}
\hspace{0.1mm}
\includegraphics[scale=0.3]{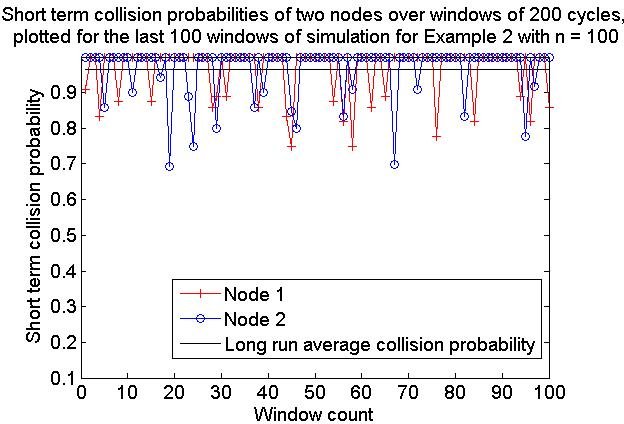}
\caption{Example~2: Simulation results depicting short term unfairness for various $n$. Shown are the short term collision probabilities of two of the transmitters; also plotted are the long run average collision probabilities, averaged over all the nodes and all simulation time. We see that short term unfairness decreases as the number of nodes, $n$, increases.}
\label{fig:short-term-unfairness-example2}
\vspace{-5mm}
\end{center}
\normalsize
%\vspace{-6mm}
\end{figure}
\begin{figure}[htpb]
%\footnotesize
\begin{center}
\includegraphics[scale=0.4]{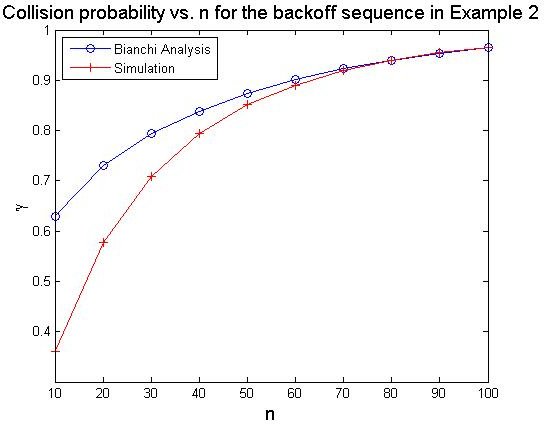}
\caption{Example~2: Collision probability vs. number of nodes; comparison of the values obtained from the Bianchi analysis against those obtained from simulations. We see that for relatively smaller values of $n$, the Bianchi analysis does not predict the performance well.}
\label{fig:gamma-vs-n-example2}
\vspace{-5mm}
\end{center}
\normalsize
%\vspace{-6mm}
\end{figure}
Consider a system where all nodes use the following backoff parameters: $K=7$, $b_0 = b_1 = b_2 = b_3 = 1$, $b_4 = b_5 = b_6 = b_7= 64$. Intuitively, this system will also encounter the same problem as the previous one; one node will attempt in every slot, while the others will be in large backoff.

Panel~1 (panel numbers are row-wise, from left to right) of Figure~\ref{fig:short-term-unfairness-example2} depicts the short term collision probabilities of two of the nodes for a system with $n=10$ transmitters, computed in the same way as in Example~1. As in Example~1, there is high variance in the short term collision probabilities of the two nodes w.r.t the long run average collision probability. In particular, it is often the case that in a window where Node~1 has a low short term collision probability, Node~2 has a very high short term collision probability, and vice-versa, thus indicating that one of the nodes monopolizes the channel in each window, shutting out the other node. Comparison of Panels~1 to 4 in Figure~\ref{fig:short-term-unfairness-example2} also shows that short term unfairness decreases with increasing $n$. 

A comparison of Figure~\ref{fig:short-term-unfairness-example2} with Figure~\ref{fig:short-term-unfairness-example1} reveals that for the backoff sequence in Example~2, fairness kicks in with fewer number of nodes compared to that in Example~1. This can be explained intuitively as follows. The maximum backoff a node can sample in Example~2 is much smaller compared to that in Example~1 (127 in Example~2 vs. 4373 in Example~1). Hence, after a successful transmission in the system, the residual backoffs of the nodes are likely to be much smaller in Example~2 than those in Example~1. Hence, the probability of any individual node attempting in the same slot as the successful node is higher than that in Example~1. Hence, a smaller number of nodes than that in Example~1 would be needed to cause the successful node to collide with a high probability in its next attempt. 

Finally, Figure~\ref{fig:gamma-vs-n-example2} demonstrates the performance of the Bianchi analysis in predicting the collision probabilities for Example~2 for various $n$. As can be expected from the short term unfairness observations earlier, the analysis (and the decoupling approximation) does not work well for $n<30$, and the accuracy gets better as $n$ increases. 

\subsection{Example~3: Small number of nodes, limited retry}
\label{subsec:example3}
\begin{figure}[htpb]
%\footnotesize
\begin{center}
\includegraphics[scale=0.3]{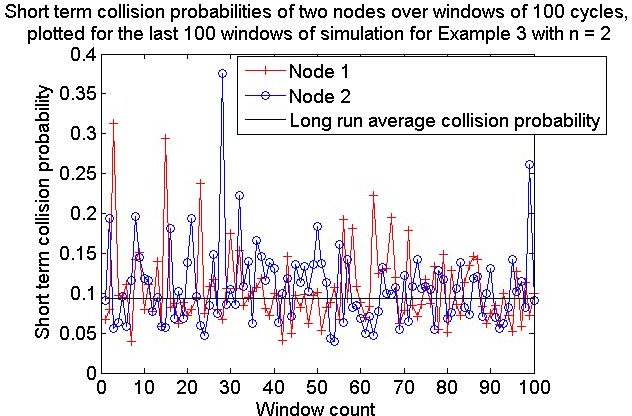}
  \hspace{0.1mm}
  \includegraphics[scale= 0.3]{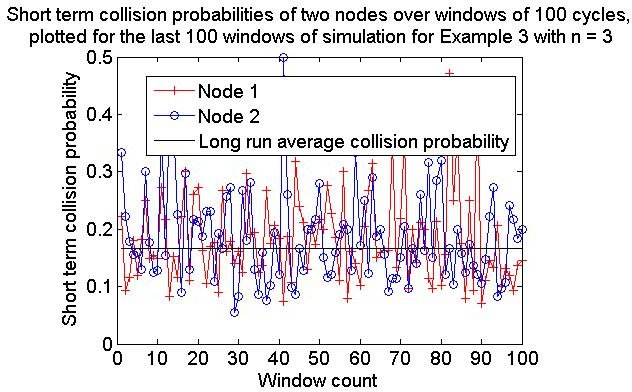}
  \vspace{0.1mm}
  \includegraphics[scale=0.3]{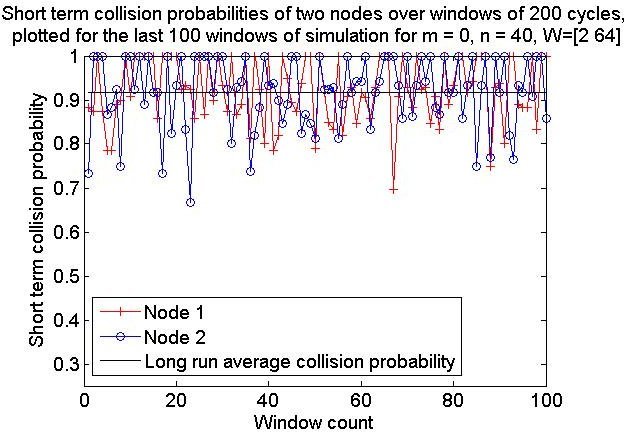}
  \hspace{0.1mm}
  \includegraphics[scale= 0.3]{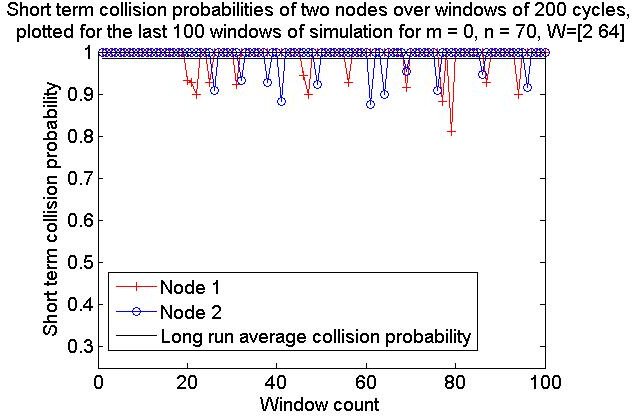}
\caption{Example~3: Simulation results depicting short term unfairness for various $n$. Shown are the short term collision probabilities of two of the transmitters; also plotted are the long run average collision probabilities, averaged over all the nodes and all simulation time. We see that short term unfairness decreases as the number of nodes, $n$, increases.}
\label{fig:short-term-unfairness-example3}
\vspace{-5mm}
\end{center}
\normalsize
%\vspace{-6mm}
\end{figure}
\begin{figure}[htpb]
%\footnotesize
\begin{center}
\includegraphics[scale=0.45]{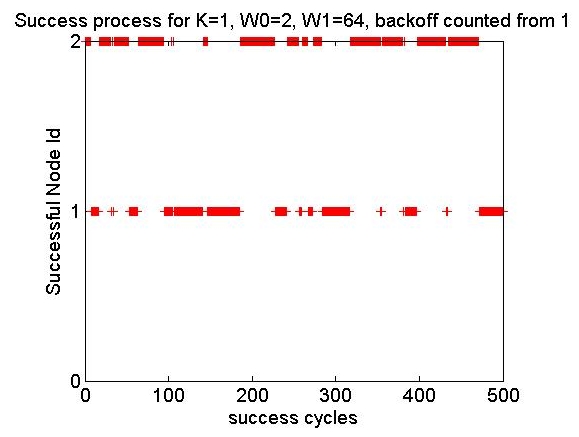}
\caption{Example~3: Simulation results depicting short term unfairness for a system with 2 nodes. Shown is the evolution of the success process of the two nodes over 500 successful transmissions of the system in Example~3. The success process is bursty, indicating short term unfairness.}
\label{fig:success-process-example3}
\vspace{-5mm}
\end{center}
\normalsize
%\vspace{-6mm}
\end{figure} 
\begin{figure}[htpb]
%\footnotesize
\begin{center}
\includegraphics[scale=0.45]{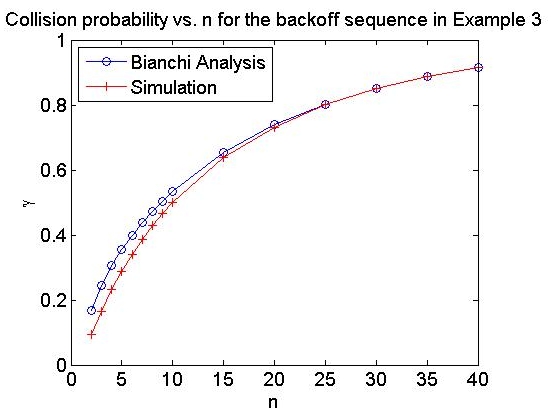}
\caption{Example~3: Collision probability vs. number of nodes; comparison of the values obtained from the Bianchi analysis against those obtained from simulations. We see that for $n<10$, errors in the Bianchi analysis are more than 10\%.}
\label{fig:gamma-vs-n-example3}
\vspace{-5mm}
\end{center}
\normalsize
%\vspace{-6mm}
\end{figure}

In the previous examples, the number of nodes were ten or more, and the retry limit $K$ was moderate. What if the number of nodes, and the retry limit are both small? Does a large backoff variability still cause unfairness? Turns out it does. To see this, we consider a system with the following backoff parameters: $K=1$, $b_0=1.5$, $b_1=32.5$.

Panel~1 (panel numbers are row-wise, from left to right) of Figure~\ref{fig:short-term-unfairness-example3} depicts the short term collision probabilities of the nodes for a system with $n=2$ transmitters, computed in the same way as in Example~1. As in Examples~1 and 2, there is high variance in the short term collision probabilities of the two nodes w.r.t the long run average collision probability. In particular, it is often the case that in a window where Node~1 has a low short term collision probability, Node~2 has a relatively high short term collision probability, and vice-versa, thus indicating that one of the nodes monopolizes the channel in each window, shutting out the other node. However, compared to Examples~1 and 2, the variance in short term collision probabilities is lower in this example, indicating that the extent of unfairness is less compared to Examples~1 and 2. This is because, due to the smaller retry limit and less backoff variability compared to Examples~1 and 2, the node in the higher backoff stage can return to backoff stage 0 faster compared to Examples~1 and 2. Nevertheless, the system does exhibit some short term unfairness. To see this more clearly, we plot in Figure~\ref{fig:success-process-example3}, the Node IDs of the successful nodes for the last 500 successful transmissions in a simulation run of the 2-node system. It can be seen from Figure~\ref{fig:success-process-example3} that the success processes at the two nodes are bursty in nature, indicating that one node captures the channel over prolonged durations, while the other gets zero throughput during that period, i.e., there is short term unfairness. 

Comparison of Panels~1 to 4 in Figure~\ref{fig:short-term-unfairness-example3} also shows that short term unfairness decreases with increasing $n$.

Finally, Figure~\ref{fig:gamma-vs-n-example3} demonstrates the performance of the Bianchi analysis in predicting the collision probabilities for this example for various $n$. As can be expected from the short term unfairness observations earlier, the analysis (and the decoupling approximation) does not work well at relatively lower values of $n$ ($n<10$), and the accuracy gets better as $n$ increases. 

\subsection{Example~4: Limited backoff variability, large retransmission limit}
\label{subsec:example4}
\begin{figure}[htpb]
%\footnotesize
\begin{center}
\includegraphics[scale=0.3]{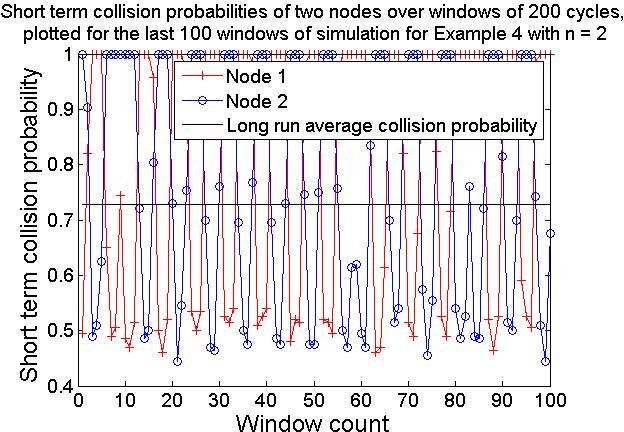}
  \hspace{0.1mm}
  \includegraphics[scale= 0.3]{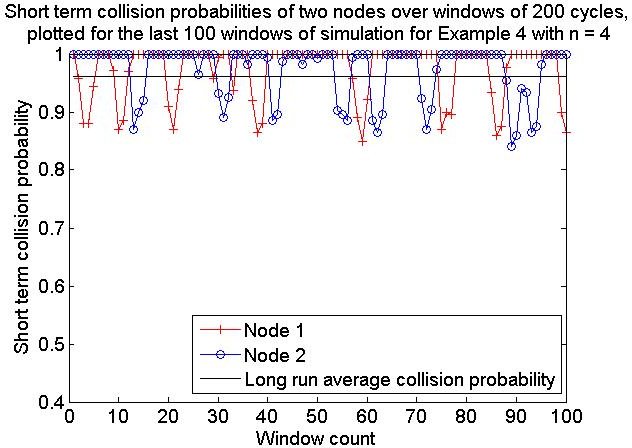}
  \vspace{0.1mm}
  \includegraphics[scale=0.3]{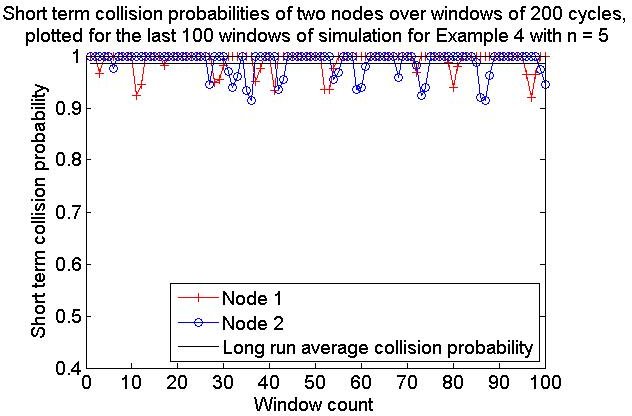}
  \hspace{0.1mm}
  \includegraphics[scale= 0.3]{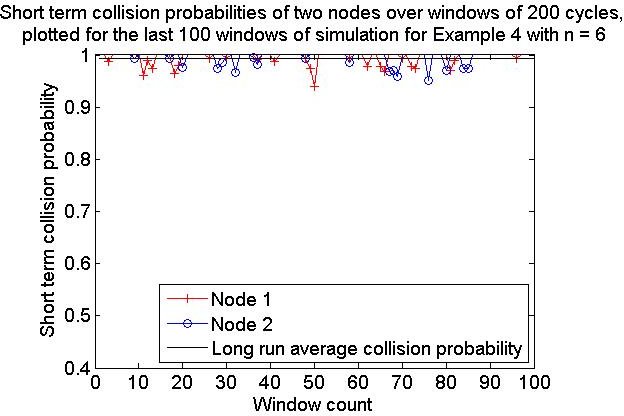}
\caption{Example~4: Simulation results depicting short term unfairness for various $n$. Shown are the short term collision probabilities of two of the transmitters; also plotted are the long run average collision probabilities, averaged over all the nodes and all simulation time. We see that short term unfairness decreases as the number of nodes, $n$, increases.}
\label{fig:short-term-unfairness-example4}
\vspace{-5mm}
\end{center}
\normalsize
%\vspace{-6mm}
\end{figure}
\begin{figure}[htpb]
%\footnotesize
\begin{center}
\includegraphics[scale=0.4]{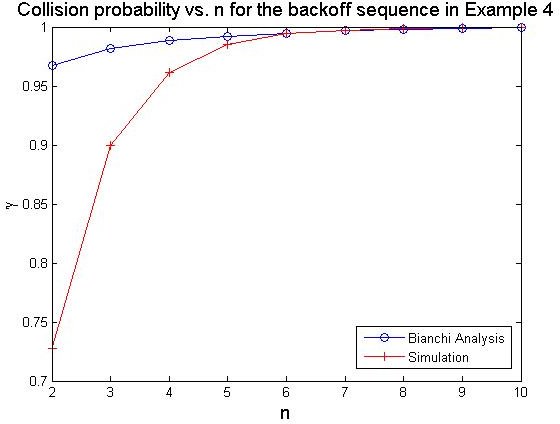}
\caption{Example~4: Collision probability vs. number of nodes; comparison of the values obtained from the Bianchi analysis against those obtained from simulations. We see that for $n\leq 3$, the Bianchi analysis does not predict the performance well.}
\label{fig:gamma-vs-n-example4}
%\vspace{-5mm}
\end{center}
\normalsize
%\vspace{-6mm}
\end{figure}
In all the previous examples, the short term unfairness arose from the large variability in backoff. What if the minimum and maximum backoffs are comparable? The following example demonstrates that there could still be short term unfairness if the retry limit is large enough. Consider a system with backoff parameters $K=400$, $b_0 = \ldots = b_{100} = 1$, $b_{101}=\ldots=b_{400}=2$. 

Panel~1 (panel numbers are row-wise, from left to right) of Figure~\ref{fig:short-term-unfairness-example4} depicts the short term collision probabilities of the nodes for a system with $n=2$ transmitters, computed in the same way as in Example~1. As in the previous examples, there is relatively high variance in the short term collision probabilities of the two nodes w.r.t the long run average collision probability, implying short term unfairness.

The short term unfairness in this example can be explained intuitively as follows. Suppose both the nodes start their backoffs together from the initial backoff stage. Since the backoff window is 1 up to backoff stage 100, they will continue to attempt together and collide until they both reach backoff stage 101 together. Beyond this point, they sample their backoffs from a larger window, and hence one of them, say Node~1 will succeed at some point in time. The backoff stage of Node~1 is reset to 0, while Node~2's backoff stage is somewhere between 101 and 400. For concreteness, let us say Node~2's backoff stage is 101. Now Node~1 will attempt in every backoff slot until Node~1's backoff stage exceeds 100 (i.e., it encounters 100 successive collisions), since $b_0=\ldots=b_{100}=1$; the probability of this event is very small. On an average, Node~2 makes an attempt every 2 backoff slots; thus Node~1's attempt succeeds every alternate slot on an average. Note that all of Node~2's attempts encounter collisions, since Node~1 attempts in every slot. To return to backoff stage 0 (and thus again be on the same page as Node~1), Node~2 has to encounter 300 successive collisions, starting from backoff stage 101, since $K=400$. Since a collision occurs every 2 backoff slots on an average, 300 collisions will require 600 backoff slots. Over these 600 backoff slots, Node~1 will have 300 successful transmissions on an average, while all of Node~2's attempts will collide. Thus, we will see a burst in the success process of Node~1, and zero throughput for Node~2 in the corresponding period.

Comparison of Panels~1 to 4 in Figure~\ref{fig:short-term-unfairness-example4} also shows that short term unfairness decreases with increasing $n$ for this example as well. 

Finally, Figure~\ref{fig:gamma-vs-n-example4} demonstrates the performance of the Bianchi analysis in predicting the collision probabilities for this example for various $n$. As can be expected from the short term unfairness observations earlier, the analysis (and the decoupling approximation) does not work well for $n\leq 3$, and the accuracy gets better as $n$ increases. 

\subsection{Convergence of the ODE trajectories}
\label{subsec:ode-convergence}
For each of the above examples, we also studied the ODE trajectories obtained from the mean-field analysis (Section~\ref{sec:mean-field}) starting from different initial conditions. The trajectories were obtained using the ode45/ode23 tool in MATLAB. We also obtained the unique stationary point, $\mu^*$, of the ODE in each case using the method described in Section~\ref{subsec:critique-bianchi}, and studied the euclidean norm difference, $||\mu(t)-\mu^*||$, of the ODE trajectory and the stationary point as a function of time, for different initial conditions. The results are summarized in Figure~\ref{fig:ode-convergence}. 
\begin{figure}[htpb]
%\footnotesize
\begin{center}
\includegraphics[scale=0.3]{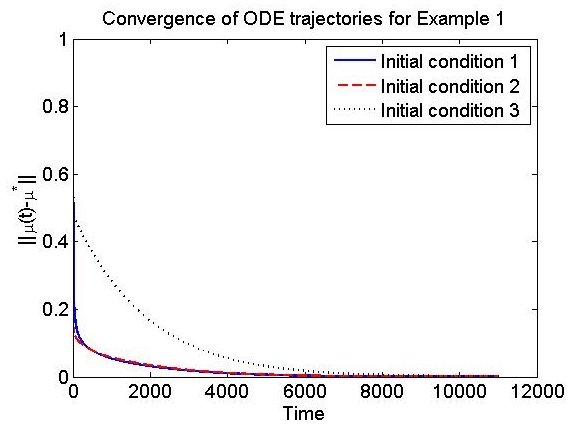}
  \hspace{0.1mm}
  \includegraphics[scale= 0.3]{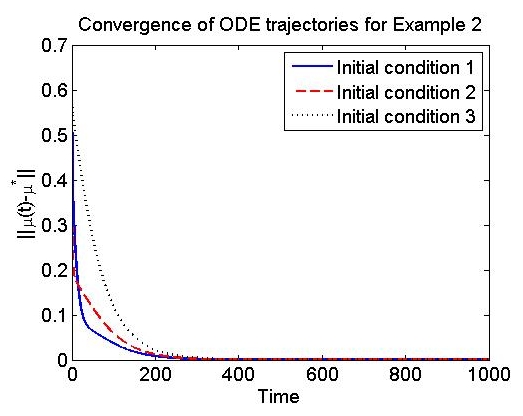}
  \vspace{0.1mm}
  \includegraphics[scale=0.3]{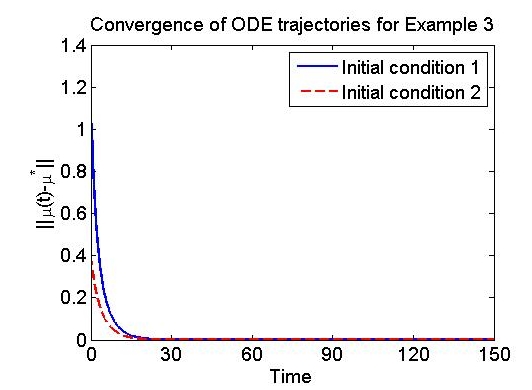}
  \hspace{0.1mm}
  \includegraphics[scale= 0.3]{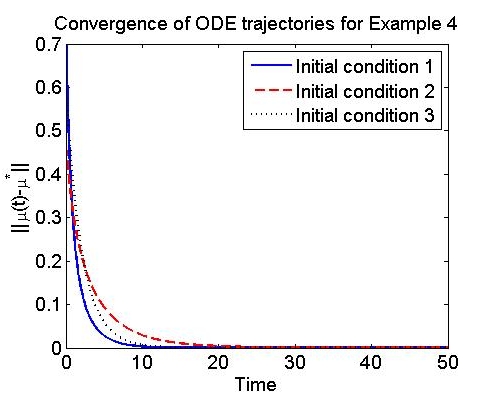}
\caption{Convergence of ODE trajectories to the unique stationary point starting with different initial conditions; seems to suggest that the stationary point is, in fact, a g.a.s.e.}
\label{fig:ode-convergence}
\vspace{-5mm}
\end{center}
\normalsize
%\vspace{-6mm}
\end{figure}
For each example, the ODE trajectories seem to converge to the unique stationary point, starting from different initial conditions. This indicates that the stationary point might, in fact, be a globally asymptotically stable equilibrium of the ODE in each case, and therefore, the decoupling approximation should hold asymptotically. Despite this observation, we have seen that the decoupling approximation does not hold for these examples at smaller values of $n$, a clear evidence that the mean-field asymptotic approach is not adequate to predict the system behavior for practical values of $n$.     

\subsection{Summary}

The above discussion can be summarized as follows.

\begin{enumerate}
 \item There exist several classes of backoff parameters that lead to short term unfairness at small to moderate values of $n$. 
 \item For such systems, due to the high correlation in the system evolution, the decoupling approximation does not hold. In particular, the Bianchi analysis does not predict the system performance well.
 \item However, as the number of nodes in the system is made large (resulting in an increase in collision probability), the short term unfairness gradually disappears, and the Bianchi analysis becomes more accurate in predicting the system performance.  
\end{enumerate}

Since for several classes of backoff parameters, the Bianchi approximation does not work well for systems with a ``permissible'' number of nodes (permissible in the sense that the resulting collision probability is not too high; for example, $\gamma \leq 0.7$), and it is very hard to know beforehand if the analysis will work well for a given number of nodes, we need an alternate analytical technique that can predict the system peformance well even in the presence of high correlation in the system evolution. This will be the focus of our work from here onwards.

Unlike the Bianchi model where a state independent, constant attempt rate is assumed for all the nodes, we will need state dependent attempt rates to capture the bursty nature of the success processes of the nodes (see Figure~\ref{fig:success-process-example3}, and Section~\ref{subsec:system-evolution-mrp}). To this end, we need to maintain some state at the end of each transmission over the medium, and determine appropriate attempt rates following the transmission periods. 

We start with an alternate stochastic model of the system.

\section{A Markov Renewal Model of the System Evolution}
\label{sec:exact_model}

\begin{figure}[ht]
%\footnotesize
\begin{center}
\includegraphics[scale=0.4]{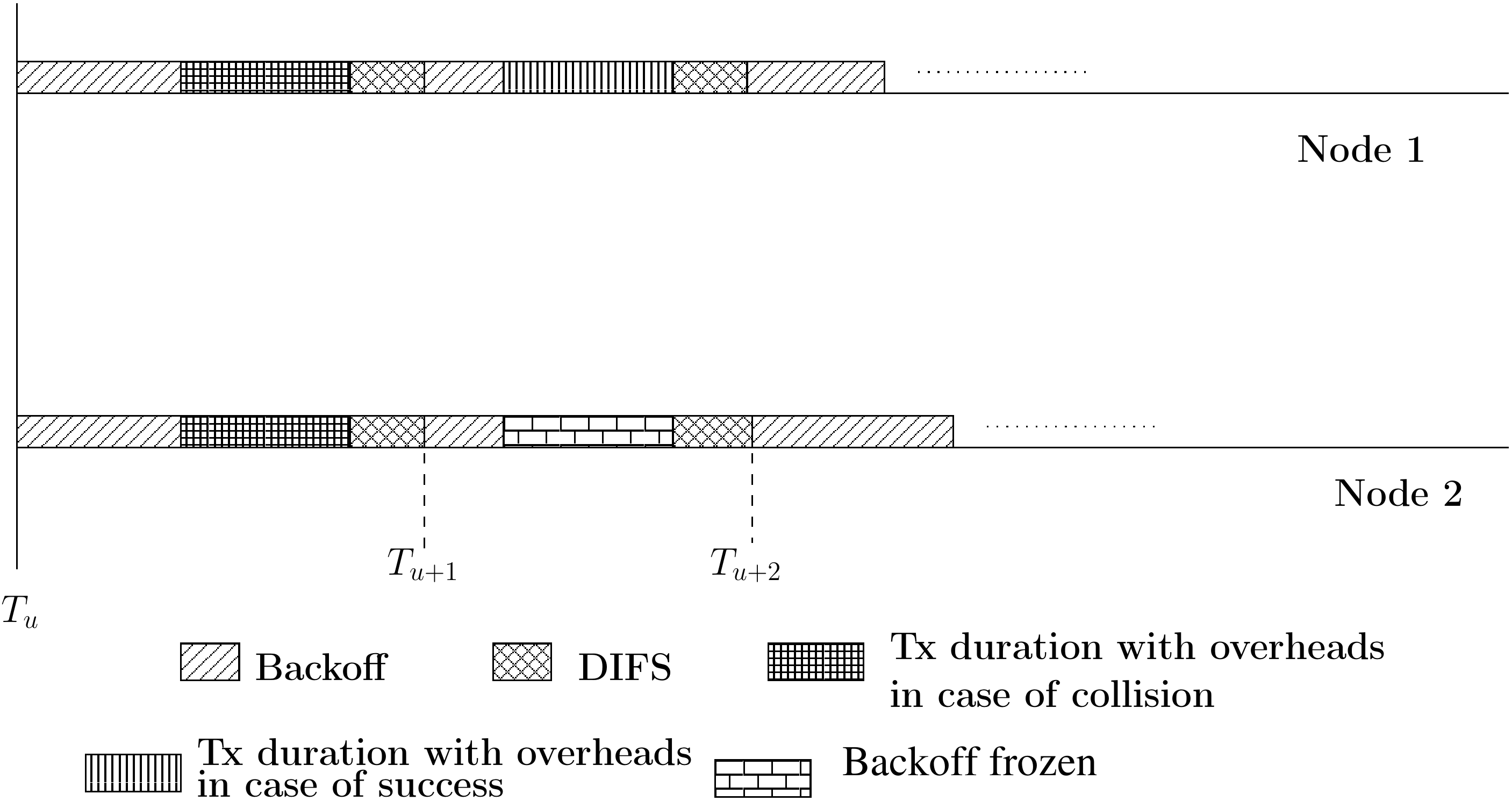}
\caption{\textbf{Transmission Cycles for $n=2$}. Denote by $T_u$, the first instant after the $u^{th}$ activity in the medium when the nodes start counting down their backoffs; we set $T_0=0$. The intervals $[T_u, T_{u+1}]$ and $[T_{u+1},T_{u+2}]$ are, respectively, the $(u+1)^{th}$ and $(u+2)^{th}$ \emph{transmission cycles}. In each transmission cycle, the system encounters one successful packet transmission, or a packet collision.}
\label{fig:tx-cycle-explain}
%\vspace{-5mm}
\end{center}
\normalsize
%\vspace{-6mm}
\end{figure}

Our system consists of $n\geq 2$ saturated transmitting nodes, and their receivers, operating under IEEE~802.11 DCF. Recall from Section~\ref{subsec:markov-model-dcf} that the system evolution can be modeled by a DTMC embedded at the backoff slot boundaries, the states of the DTMC being the backoff stage and the residual backoff count of each node. 

Alternatively, we can model the system evolution as a Markov renewal process. This model is equivalent to the DTMC model as explained in the remark at the end of this section, but unlike the DTMC model, it avoids embedding at deterministic transitions. The model is as follows: let $T_u$ be the first instant after the $u^{th}$ activity in the medium when the nodes start counting down their backoffs. See, for example, Figure~\ref{fig:tx-cycle-explain}, which depicts a sample path of the system evolution for $n=2$. We call the interval $[T_u, T_{u+1}]$ the $(u+1)^{th}$ \emph{transmission cycle}. In each transmission cycle, there is excactly one activity in the medium. 

Let $B_{u,i}, S_{u,i}$, denote respectively the residual backoff count, and backoff stage of Node~$i$, $i=1,2,\ldots,n$ at $T_u$. Recalling the notation for the protocol parameters of IEEE~802.11 DCF, $S_{u,i}\in\{0,1,\ldots,K\}$, $B_{u,i}\in\{1,\ldots,W_{S_{u,i}}\}$. Then, the process $(\{B_{u,i}, S_{u,i}\}_{i=1}^n, T_u)$ is a Markov Renewal Process \cite{kulkarni95modeling-stochastic-systems}, with $\{B_{u,i}, S_{u,i}\}_{i=1}^n$ being the embedded Markov chain, whose transition structure is explained next. 

Note that $(T_u+B_{u,i})$ is the instant when Node~$i$ is scheduled to finish its backoff, and attempt a transmission in the $(u+1)^{th}$ transmission cycle. Let $\underline{B}_u = \min_{1\leq i\leq n} B_{u,i}$, and $I_u = \arg \min_{1\leq i\leq n} B_{u,i}$. In case of a tie, take $I_u$ to be the smallest node ID among the nodes involved in the tie. 

\noindent
\textbf{Observations:}

\begin{enumerate}
 \item $(T_u+B_u)$ and $I_u$ are, respectively, the attempt instant, and Node id of the first node to attempt transmission in the $(u+1)^{th}$ transmission cycle. 
 \item A successful transmission happens \emph{iff} for all $i\neq I_u$, $B_{u,i} > \underline{B}_u$, and a collision happens otherwise. 
\end{enumerate}

With the above information, the transition structure of the embedded Markov chain can be summarized as follows:
\begin{enumerate}
\item Initialize the set of nodes attempting in the $(u+1)^{th}$ transmission cycle as $S_{a,u} = \phi$. For each node~$i$, $1\leq i\leq n$, if $B_{u,i} > \underline{B}_u$, i.e., the node hears the ongoing transmission before finishing its backoff, then Node~$i$'s backoff is frozen in the $(u+1)^{th}$ transmission cycle, and its backoff states are updated as $B_{u+1,i} = B_{u,i}  - \underline{B}_u$, and $S_{u+1,i}=S_{u,i}$.

If, on the other hand, $B_{u,i}=\underline{B}_u$, then Node~$i$ attempts in the $(u+1)^{th}$ transmission cycle, and the set of attempting nodes is updated as $S_{a,u} = S_{a,u}\cup \{i\}$.

\item If $|S_{a,u}|=1$, i.e., exactly one node, namely, Node~$I_u$ attempted in the $(u+1)^{th}$ transmission cycle, then the transmission is successful, and Node~$I_u's$ backoff stage becomes $S_{u+1,I_u}=0$; $B_{u+1,I_u}$ is sampled from a uniform distribution from $\{1, CW_{\min}\}$. The duration of the transmission cycle in this case is \\$B_u + \text{successful transmission duration with overheads}$. See, for example, the transmission cycle $[T_{u+1},T_{u+2}]$ in Figure~\ref{fig:tx-cycle-explain}.

\item If $|S_{a,u}|> 1$, then more than one node attempted in the $(u+1)^{th}$ transmission cycle, resulting in a collision for all the nodes in $S_{a,u}$. For each node~$j\in S_{a,u}$, its backoff stage will be updated as $S_{u+1,j} = (S_{u,j}+1) mod (K+1)$, where $K$ is the maximum allowed number of retransmissions. $B_{u+1,j}$ is sampled uniformly from the contention window corresponding to $S_{u+1,j}$. The duration of the transmission cycle in this case is $B_u + \text{collision duration with overheads}$. See, for example, the transmission cycle $[T_{u},T_{u+1}]$ in Figure~\ref{fig:tx-cycle-explain}. 

Note that this step captures the practical fact that if $K+1$ attempts are reached without success then the HOL packet is discarded, the next packet (in the saturated queue) enters the HOL location, and the backoff state is reset.  
\end{enumerate}

\noindent
\remark
Observe that the Markov renewal model embedded at the epochs $T_u$ is equivalent to the DTMC model (Section~\ref{subsec:markov-model-dcf}) embedded at the backoff slot boundaries in the following sense: for any given system, suppose we simulate the two models starting with the same initial conditions (backoff stages of the nodes), and the same random seed; the same random seed ensures that the backoff sampled by a Node~$i$ after the $k^{th}$ retransmission of its $j^{th}$ packet is the same for both the simulations, for all $i,j,k$. Then, the two models give rise to the same sample path for the system evolution (after reconstructing the original process in unconditional time from the backoff process obtained from the DTMC model). To see this, note that in the DTMC model, in each backoff slot intervening the epochs $T_u$, the nodes do nothing but count down their residual backoffs by 1. The net effect of this countdown process is just an update of the residual backoff count, and the backoff stage of each node at the subsequent $T_u$ epoch. This is incorporated in the Markov renewal model through the update rules for $\{B_{u,i}, S_{u,i}\}_{i=1}^n$.  \hfill \Square

From now on, we shall focus on the Markov renewal model. However, even this model involves an embedded $2n$-dimensional Markov chain, whose state space is, in fact, the same as the DTMC model. The size of the state space is $(W_0+W_1+\cdots+W_K)^n$, which grows \emph{exponentially} with the number of nodes, and is prohibitively large even for $n=2$ for the default protocol parameters of IEEE~802.11b, making an exact analysis of the embedded Markov chain computationally intractable. We, therefore, focus on developing an approximate, parsimonious analysis that still accurately captures the system behavior.

\section{A Parsimonious Simplification of the Markov Renewal Model in Section~\ref{sec:exact_model}}
\label{sec:mrp-state-dependent} 

\subsection{An approximate Markov renewal model for system evolution}
\label{subsec:system-evolution-mrp}
While retaining the embedded Markov process structure at the starts of transmission cycles, we aim to simplify the evolution of the process between these embedding points to reduce the computational complexity. In particular, we aim to avoid the exponential growth of the underlying state space size with the number of nodes. The complexity of the analysis of the detailed process constructed in Section~\ref{sec:exact_model} comes from the complex transition structure, due to the necessity to keep track of the various events, and their timing, between the embedding points. One possible way to simplify the evolution between the embedding instants is to adopt the state independent, Bernoulli attempt process approximation in \cite{bianchi00performance,kumar-etal04new-insights} (see Section~\ref{subsec:bianchi-analysis}). In the context of the Markov renewal model, this amounts to making the following approximation: in each transmission cycle, each node attempts with a constant probability $\beta$ in each slot, conditioned on being in backoff, independent of everything else. Consider the consequence of this approximation on the success processes of the nodes. Observe that under this approximation, the probability that the next successful transmission in the system is due to a particular Node~$j\in\{1,\ldots,n\}$ is $\frac{1}{n}$, \emph{independent of which node made the last successful transmission}. To see this, note that under the constant, state independent attempt rate approximation, the evolution of the process from the last successful transmission onwards does not depend on the node id of the last successful node. 

Let us compare this against observations from simulations. 
\begin{figure}[t]
%\footnotesize
\begin{center}
\includegraphics[scale=0.35]{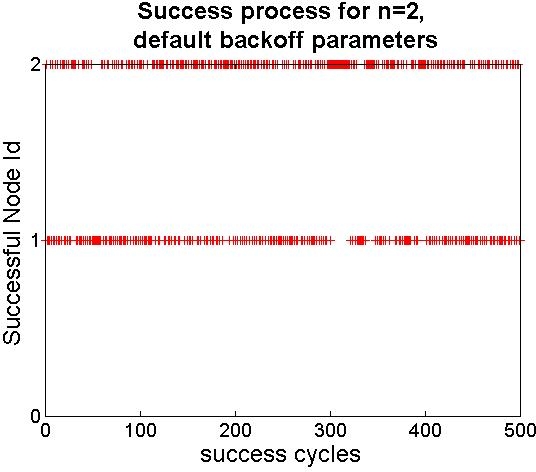}
\hspace{0.1mm}
\includegraphics[scale=0.35]{plots/success_sequence_m0_K1_new_unfairness_example.jpg}
\hspace{0.1mm}
\includegraphics[scale=0.35]{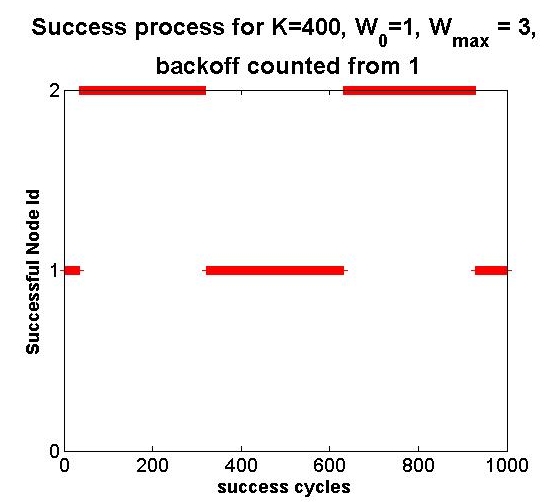}
\caption{Evolution of the success process for a 2-node system over hundreds of successful transmission cycles. (Panel numbers are row-wise from left to right) Panel 1: For default backoff parameters of IEEE~802.11b}. Panel 2: For the backoff sequence in Example~3 (Section~\ref{subsec:example3}). Panel 3: For the backoff sequence in Example~4 (Section~\ref{subsec:example4}).
\label{fig:success-processes}
\vspace{-5mm}
\end{center}
\normalsize
%\vspace{-6mm}
\end{figure} 
Figure~\ref{fig:success-processes} shows the success processes of the nodes for a 2-node system over several hundreds of successful transmission cycles. The plots were obtained in the same manner as Figure~\ref{fig:success-process-example3} in Section~\ref{subsec:example3}. Panel~1 shows the success process for the default backoff parameters of IEEE~802.11b. Panels~2 and 3 show the success processes for the backoff sequences introduced in Examples~3 and 4 respectively in Sections~\ref{subsec:example3} and \ref{subsec:example4}. We observe from Figure~\ref{fig:success-processes} that while the success process for the default parameters of IEEE~802.11b is consistent with the conclusions drawn earlier from the state independent constant attempt rate approximation, those conclusions clearly do not hold for the success processes in Examples~3 and 4, which exhibit significant correlation (short term unfairness). In particular, the burstiness of the success processes indicate that the attempt rates of the nodes are skewed in favor of the node that succeeded last. Thus, the constant, state independent attempt rate approximation will not work in such cases (as we have also seen in Section~\ref{sec:stu-examples}). 

\noindent
\textbf{Accounting for \emph{short term unfairness}:}
Taking cue from this, we adopt the Bernoulli attempt process approximation for the nodes as in \cite{bianchi00performance,kumar-etal04new-insights}, but introduce \emph{state dependent attempt rates, namely, $\bs,\bc$, and $\bd$ to distinguish among three cases: whether a node encountered a success, a collision, or an interruption (of its backoff), respectively, in the previous transmission cycle}. Furthermore, observe that under this approximation, in order to construct the system evolution in a transmission cycle, we need to know the attempt rates of the nodes at the start of the transmission cycle, which, in turn, depend on the number of nodes that attempted in the last cycle, since the nodes that did not attempt (i.e., were interrupted) will attempt at rate $\bd$ in the next cycle, while the nodes that attempted in the last cycle will attempt at rate $\bs$ or $\bc$, depending on whether the last transmission was a success or a collision. Hence, we associate with each epoch $T_u$, a state, $N_u$, \emph{the number of nodes that attempted in the previous cycle}. In the detailed model of Section~\ref{sec:exact_model}, we did not need this state since we kept track of more detailed states, namely, the backoff stage, and the residual backoff of each node, which completely determine the subsequent evolution (including the number of nodes attempting in a transmission cycle). 
 
Our approximations about the node attempt processes are summarized as follows:

\noindent
\textbf{(A1)} Suppose there was a success in the $u^{th}$ transmission cycle. All the nodes start their backoffs from $T_u$. The node that was successful in the previous transmission cycle attempts independently with probability $\beta_s$ in each slot, conditioned on being in backoff. The other nodes, all having been interrupted during their backoffs in the previous cycle, attempt independently with probability $\beta_d$ in each slot, conditioned on being in backoff.\hfill \Square

\noindent 
\textbf{(A2)} Suppose there was a collision involving $N_u$ nodes in the $u^{th}$ transmission cycle. All nodes start their backoffs from $T_u$. $N_u$ of the nodes attempt independently with probability $\beta_c$ in each slot, while the remaining $n-N_u$ nodes attempt independently with probability $\beta_d$ in each slot, all conditioned on being in backoff. \hfill \Square 

\noindent
\remarks

\begin{enumerate}
 \item After a successful transmission in the system, we may expect the residual backoffs of the interrupted nodes to be relatively large compared to the next backoff of the successful node (which samples its backoff from the smallest contention window), especially for backoff sequences that lead to short term unfairness; thus, the attempt rates of the interrupted nodes can be expected to be significantly lower than that of the successful node. This is the rationale behind introducing the attempt rates $\bs$ and $\bd$ to distinguish between the successful node, and the interrupted nodes. 
 \item Following a similar rationale, in case of a collision, we may expect the nodes that were interrupted (did not participate in the collision) to have relatively large residual backoffs compared to the nodes involved in the collision. Also, since after a collision, a node will sample backoff from a larger contention window, its attempt rate after a collision can be expected to be lower than that after a success. Hence we introduce the attempt rate $\bc$ to distinguish the colliding nodes from the interrupted nodes, as well as the successful node.
 \item One can of course, think of more complicated, and perhaps more accurate, approximate models; e.g., we may want to distinguish between the interrupted nodes in case of a success, and the interrupted nodes in case of a collision. That will require two different values of $\bd$, instead of a common value as above. However, this model is significantly harder, and computationally more complex, to analyze than the approximate model we have introduced. Moreover, it turns out that the approximate analysis based on the three-attempt-rates model introduced above predicts the system performance quite accurately even in the presence of high correlation in the system evolution. \hfill \Square  
\end{enumerate}

\noindent
\textbf{A simple Markov renewal process model for the system:}
With these approximations, observe that the process $\{N_u, T_u\}$, is a Markov renewal process (MRP), the state space of the embedded Markov chain being $\{1,\ldots,n\}$. 

\subsection{Analysis of the MRP, given $\beta_c, \beta_d$, and $\beta_s$}
\label{subsec:mrp-analysis-given-beta}

The Markov renewal process model has $n$ as a parameter, and requires the quantities $\beta_c, \beta_d$, and $\beta_s$, which are not known a priori. We shall explain how to compute $\beta_c, \beta_d$, and $\beta_s$ in Section~\ref{subsec:tagged-node-evolution}. Given $\beta_c, \beta_d$, and $\beta_s$, let $P$ be the transition probability matrix of the embedded Markov chain. We now proceed to write down the transition probabilities. We use the shorthand $p(n_a,n^\prime_a)$ to denote the probability $Pr[N_{u+1}=n^\prime_a|N_u = n_a]$. 

\noindent
\emph{Computation of $p(n_a,n^\prime_a)$}:

Define the sets $F(n_a,n^\prime_a) = \{(i,j): 0\leq i\leq n_a, 0\leq j\leq n-n_a, i+j = n^\prime_a\}$ for all $n_a, n^\prime_a\in\{1,\ldots,n\}$. Also define

\begin{align}
 q(n_a,n^\prime_a) &= \sum_{(i,j)\in F(n_a,n^\prime_a)}\binom{n_a}{i}\binom{n-n_a}{j}\bx^i(1-\bx)^{n_a-i}\bd^j(1-\bd)^{n-n_a-j}\label{eqn:q_na_na_prime}
\end{align}
where $\bx=\bs$ if $n_a=1$, and $\bx=\bc$, if $n_a > 1$. 

Observe that given the information that $n_a$ nodes are attempting at rate $\bx$, and remaining $(n-n_a)$ nodes are attempting at rate $\bd$, $q(n_a,n^\prime_a)$ is the probability that $n_a^\prime$ nodes attempt together in a backoff slot, while the remaining $(n-n_a^\prime)$ nodes remain silent. 

Then we can write

\begin{align}
 p(n_a,n^\prime_a) &= (1-\bx)^{n_a}(1-\bd)^{n-n_a}p(n_a,n^\prime_a) + q(n_a,n^\prime_a)\nonumber\\
%\Rightarrow p(n_a,n^\prime_a) &= \frac{q(n_a,n^\prime_a)}{1-(1-\bx)^{n_a}(1-\bd)^{n-n_a}}
\end{align}
Here, the first term corresponds to the event that none of the nodes attempt in the first backoff slot; in this case, due to the assumption of Bernoulli attempt processes, the system encounters a renewal with state $n_a$, and the conditional probability (given that none of the nodes attempted in the first slot) of the next state being $n^\prime_a$ remains $p(n_a,n^\prime_a)$. Thus we have

\begin{align}
 p(n_a,n^\prime_a) &= \frac{q(n_a,n^\prime_a)}{1-(1-\bx)^{n_a}(1-\bd)^{n-n_a}}\label{eqn:p_na_n_prime}
\end{align}
where $\bx=\bs$ if $n_a=1$, and $\bx=\bc$, if $n_a > 1$.
 
From the above transition probability structure, it is easy to observe that for positive attempt rates, the embedded DTMC is finite, irreducible, and hence, \emph{positive recurrent}. Let $\pi$ denote the stationary distribution of this DTMC, which can be obtained as the unique solution to the system of equations $\pi=\pi P$, subject to $\displaystyle{\sum_{n_a\in\{1,2,\ldots,n\}}}\pi(n_a)=1$. 

\subsubsection{Obtaining the collision probability, $\gamma$}
\label{subsubsec:gamma_general_n}
By symmetry, the long run average collision probability for all the nodes is the same, which we denote by $\gamma$. It is defined as 
\begin{equation*}
 \gamma = \lim_{t\to\infty}\frac{C_i(t)}{A_i(t)},\:i=1,2,\ldots,n
\end{equation*}
where, $C_i(t)$ and $A_i(t)$ denote respectively, the number of collisions and the number of attempts by Node~$i$ until time $t$. Denoting $C(t)\define\sum_{i=1}^n C_i(t)$, the total number of collisions in the system until time $t$, and $A(t)\define\sum_{i=1}^n A_i(t)$, the total number of attempts in the system until time $t$, it is also easy to observe (by noting that the long run time-average collision rates, and the long run time-average attempt rates of all the nodes are equal by symmetry) that 
\begin{equation*}
 \gamma = \lim_{t\to\infty}\frac{C(t)}{A(t)}
\end{equation*}

Denote by $\Ccal$ and $\Acal$, respectively, the random variables representing the number of collisions, and the number of attempts in the system in a transmission cycle. Then, using Markov regenerative theory (see, for example, \cite{kulkarni95modeling-stochastic-systems}), we have
\begin{equation}
\gamma = \frac{\sum_{n_a=1}^n\pi(n_a)E\Ccal(n_a)}{\sum_{n_a=1}^n\pi(n_a)E\Acal(n_a)}\: a.s\label{eqn:gamma-expression-general-n} 
\end{equation}
where, $E\Ccal(n_a)$ and $E\Acal(n_a)$ denote respectively, the expected number of collisions, and attempts in the system in a transmission cycle starting with state $n_a$, and can be computed by using renewal arguments similar to those used for obtaining the transition probabilities earlier, and observing that every collision event involving $n^\prime_a$ nodes results in $n^\prime_a$ collisions (and involves $n^\prime_a$ attempts, one from each node), and every success event involves 1 attempt (from the successful node). We have, for all $n_a = 1,2,\ldots, n$,

\begin{align}
E\Ccal(n_a) &= \sum_{n^\prime_a=2}^n p(n_a,n^\prime_a)n^\prime_a\\
E\Acal(n_a) &= \sum_{n^\prime_a=1}^n p(n_a,n^\prime_a)n^\prime_a
\end{align}
where, $\bx=\bs$ if $n_a=1$, and $\bx=\bc$, if $n_a > 1$.

\subsubsection{Obtaining the normalized system throughput, $\Theta$}
\label{subsubsec:throughput_general_n}
The normalized system throughput is defined as
\begin{equation*}
 \Theta = \lim_{t\to\infty}\frac{T(t)}{t}
\end{equation*}
where $T(t)$ is the total successful data transmission duration without overheads until time $t$. 

Denote by $\Tcal$, the random variable representing the duration of successful data transmission excluding overheads in a transmission cycle. Then, by Markov regenerative theory, we have

\begin{equation}
 \Theta = \frac{\sum_{n_a=1}^n\pi(n_a)E\Tcal(n_a)}{\sum_{n_a=1}^n\pi(n_a)EX(n_a)}\: a.s\label{eqn:theta-expression-general-n}
\end{equation}
where, $E\Tcal(n_a)$ and $EX(n_a)$ are, respectively, the mean duration of successful data transmission excluding overheads, and the mean duration of the transmission cycle when the transmission cycle starts in state $n_a$. We can write down the expressions for $E\Tcal(\cdot)$ and $EX(\cdot)$ using renewal arguments similar to those given earlier as follows.

\begin{align}
E\Tcal(n_a) &= \frac{q(n_a,1)T_d}{1-(1-\bx)^{n_a}(1-\bd)^{n-n_a}}\\
EX(n_a) &= \frac{1+q(n_a,1)T_s+\sum_{n^\prime_a=2}^n q(n_a,n^\prime_a)T_c}{1-(1-\bx)^{n_a}(1-\bd)^{n-n_a}}
\end{align}
for all $n_a=1,\ldots,n$. As before, $\bx=\bs$ if $n_a=1$, and $\bx=\bc$, if $n_a > 1$. Also, $T_s$ is the time duration in a successful transmission cycle from the start of the data transmission in the medium until the time the medium is idle again, and the nodes start counting their backoffs (i.e., until the start of the next transmission cycle), and is given by
\begin{align*}
 T_s &= T_d + ACK + 2\times PHY\:\: HDR + 2T_o + SIFS + DIFS
\end{align*}
and $T_c$ is the time duration in a collision transmission cycle from the start of the first data transmission in the medium until the time the nodes start counting their backoffs (i.e., until the start of the next transmission cycle), and is given by
\begin{align*}
 T_c &= T_d + PHY\:\: HDR + T_o + SIFS + DIFS
\end{align*}
In the above expressions, $T_o$ denotes the Rx-to-Tx turnaround time. 

This completes the analysis of the system evolution, given $\bs$, $\bd$, $\bc$. 

\noindent
\remark
Until this point, what has been shown is the procedure to get the performance measures if the attempt rates, $\bs,\bc,\bd$ are given. It is an interesting exercise to relate this to what was done in the well known Bianchi analysis (Section~\ref{subsec:bianchi-analysis}). Indeed, if we set $\bs=\bc=\bd=\beta$, i.e., a state independent, constant attempt rate, we get back from Equation~\ref{eqn:gamma-expression-general-n}, the collision probability as $\gamma = 1-(1-\beta)^{n-1}$, i.e., the same expression as in the Bianchi analysis (Equation~\ref{eqn:Gamma_beta_binomial} in Section~\ref{sec:fixed_point_equation}). 

To see this, observe that under the state independent, constant attempt rate approximation, the transition probabilities $p(n_a,n^\prime_a)$ are independent of $n_a$, so that for all $n_a=1,\ldots,n$, and $n_a^\prime=1,\ldots,n$, $p(n_a,n^\prime_a)=\pi(n_a^\prime)$, the stationary probabilities of the embedded Markov chain, and these are given by $\pi(n_a^\prime)=\frac{\binom{n}{n_a^\prime}\beta^{n_a^\prime}(1-\beta)^{n-n_a^\prime}}{1-(1-\beta)^n}$. Thus, the expectations $E\Ccal(n_a)$ and $E\Acal(n_a)$ are independent of $n_a$, and are given by 
\begin{align}
 E\Ccal &= \sum_{n^\prime_a=2}^n \pi(n_a^\prime)n^\prime_a\\
E\Acal &= \sum_{n^\prime_a=1}^n \pi(n_a^\prime)n^\prime_a
\end{align}
Using these along with the expression for $\pi(n_a^\prime)$ in Equation~\ref{eqn:gamma-expression-general-n} yields 
\begin{align}
 \gamma &= \frac{\sum_{n^\prime_a=2}^n n^\prime_a \binom{n}{n_a^\prime}\beta^{n_a^\prime}(1-\beta)^{n-n_a^\prime}}{\sum_{n^\prime_a=1}^n n^\prime_a \binom{n}{n_a^\prime}\beta^{n_a^\prime}(1-\beta)^{n-n_a^\prime}}
\end{align}
Observing that the denominator is simply the expectation of a Binomial distribution with parameters $n$ and $\beta$, and the numerator lags behind the same expectation by just one term (that corresponding to $n^\prime_a=1$), we have
\begin{align}
 \gamma &= \frac{n\beta - n\beta(1-\beta)^{n-1}}{n\beta}\nonumber\\
&= 1-(1-\beta)^{n-1}
\end{align}
The same as in the Bianchi analysis. Thus, our analysis can indeed be viewed as a generalization of the Bianchi analysis with \emph{state dependent} attempt rates. \hfill \Square
 
It remains to obtain the state dependent attempt rates $\bs$, $\bd$, $\bc$. To do this, we focus on the evolution at a tagged node as described next.

\subsection{Analysis for determining $\beta_c, \beta_d$, and $\beta_s$}
\label{subsec:tagged-node-evolution}
Here we shall set up a system of fixed point equations in $\bc$, $\bd$, and $\bs$ by modeling the evolution at a tagged node. This can, in turn, be solved iteratively to yield the rates. This step is analogous to the fixed point equation ($\beta=G(\Gamma(\beta))$) in the analysis in \cite{bianchi00performance,kumar-etal04new-insights}.
\begin{figure}[ht]
%\footnotesize
\begin{center}
\includegraphics[scale=0.4]{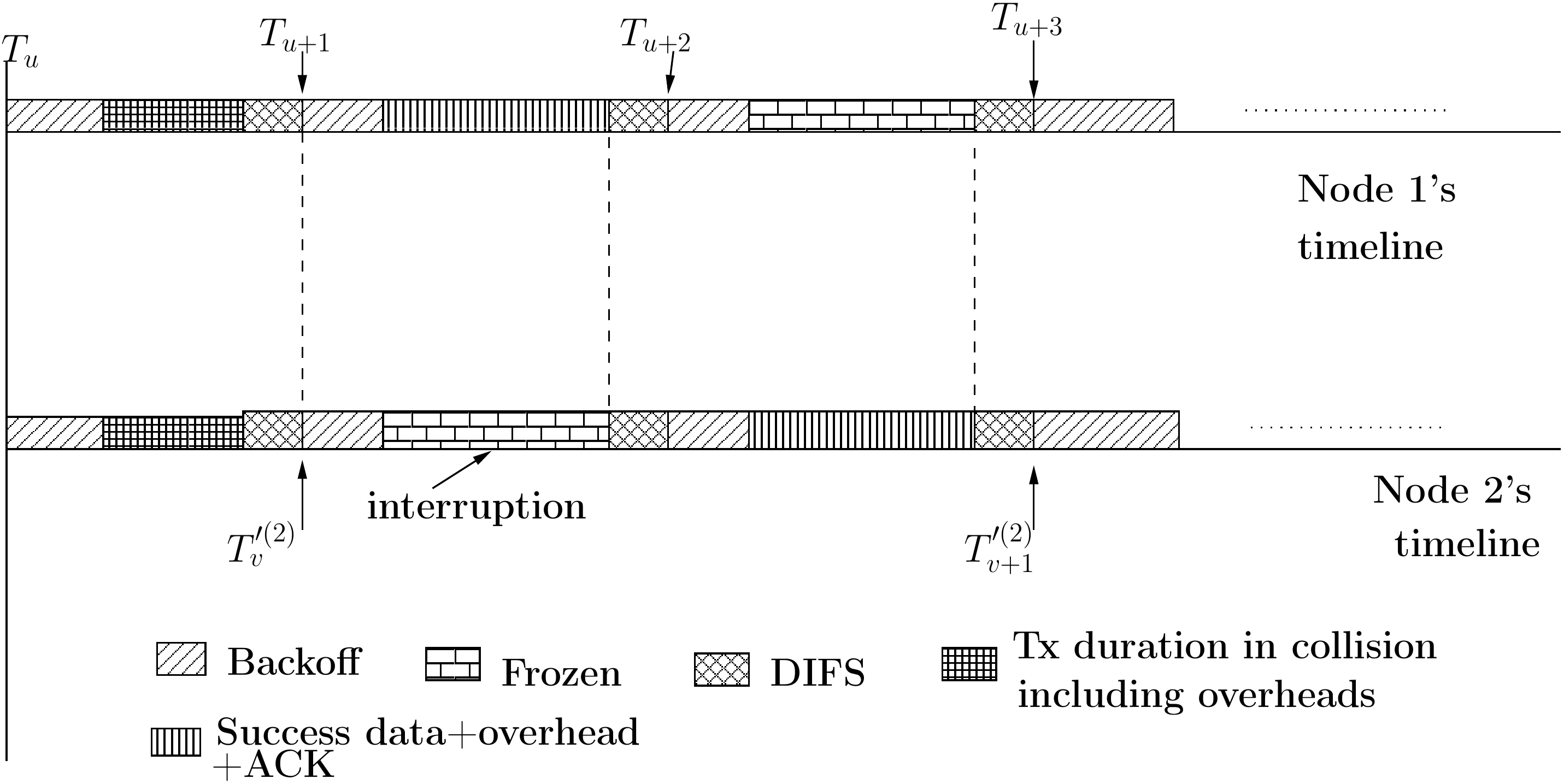}
\caption{\textbf{Backoff Cycles} for a tagged node, Node~2 in this case. The two timelines demonstrate the system evolution in unconditional time over three consecutive transmission cycles, with $T_u$,\ldots, $T_{u+3}$ being the start and end points of the transmission cycles. Denote by $T^{\prime (i)}_v$, \emph{the start of the transmission cycle following the $v^{th}$ transmission by the tagged node, $i$}, Node~2 in this example. The interval $[T^{\prime (2)}_v,T^{\prime (2)}_{v+1}]$ is called a \emph{backoff cycle} of Node~2, since in this interval, Node~2 completes one full backoff. Note that the tagged node can have \emph{exactly one attempt (backoff completion)}, and several intermediate backoff \emph{interruptions} in a backoff cycle. During each system transmission cycle $[T_u, T_{u+1}]$, any node can have at most one backoff segment. The backoff chosen at the start of a tagged node's backoff cycle is thus partitioned into several backoff segments over a random number of system transmission cycles during the tagged node's backoff cycle. Hence, a backoff cycle can encompass several transmission cycles during which the tagged node was interrupted (i.e., did not attempt).}
\label{fig:bo-cycle-explain}
%\vspace{-5mm}
\end{center}
\normalsize
%\vspace{-6mm}
\end{figure}
We consider the evolution of the process at the tagged node, say Node~$i$, and identify embedding instants $T^{\prime (i)}_v$ in this process as explained in Figure~\ref{fig:bo-cycle-explain}, where the transmission cycle break-points $T_u,\ldots$ are shown, along with the epochs $T^{\prime (2)}_v\ldots$ for Node~2 (the tagged node). After each such epoch, the tagged node samples a new backoff, using its current backoff stage $S_v$. We associate with each $T^{\prime (i)}_v$, two states: (i) $S_v\in\{0,1,\ldots,K\}$, Node~$i's$ new backoff stage, (ii) $N_v\in\{1,\ldots,n\}$, number of nodes (including the tagged Node~$i$) that attempted in the just concluded \emph{transmission cycle}.   

Notice from Figure~\ref{fig:bo-cycle-explain} that \emph{transmission cycles} are common to the entire system, whereas \emph{backoff cycles} are defined for each node. Each backoff cycle of a node \emph{comprises one or more transmission cycles of the system}. \emph{The backoff cycle of a tagged node can comprise several successful transmissions and/or collisions by the other nodes}, and ends at the end of a transmission cycle in which the tagged node transmits. 

We make the following additional approximations.

\noindent
 \textbf{(A3)} Node~$i$ samples its successive back-offs from a uniform distribution, as in the standard. When a new backoff cycle starts for Node~$i$ after a successful transmission, the other nodes, conditioned on being in backoff, attempt independently in each slot with probability $\bd$ \emph{until the end of the first transmission cycle within this backoff cycle}. If the new backoff cycle for Node~$i$ starts after a collision involving $N_v$ nodes (including Node~$i$), then $N_v-1$ of the nodes, conditioned on being in backoff, attempt independently in each slot with probability $\bc$, and the remaining $n-N_v$ nodes, conditioned on being in backoff, attempt independently in each slot with probability $\bd$ \emph{until the end of the first transmission cycle within this backoff cycle}.\hfill \Square

\noindent
\textbf{(A4)} If Node~$i$ is interrupted within a backoff cycle due to attempts by $n_a$ other nodes ($1\leq n_a\leq n-1$), thus freezing its backoff (see Figure~\ref{fig:bo-cycle-explain}), then in the next transmission cycle within this backoff cycle, Node~$i$ resumes its residual backoff countdown, all the $n-1-n_a$ nodes (excluding Node~$i$) that did not attempt in the previous transmission cycle attempt independently in each slot with probability $\bd$, conditioned on being in backoff, while the $n_a$ nodes that attempted in the previous transmission cycle attempt with probability $\beta_c$ or $\beta_s$ (depending on whether the previous transmission cycle ended in collision or success, i.e., whether $n_a > 1$ or $n_a=1$) in each slot, conditioned on being in backoff.\hfill \Square

\noindent
\remark
Note that we assumed Bernoulli attempt processes for \emph{all} nodes in obtaining the performance measures in the previous subsection, whereas to obtain the attempt rates, we now \emph{retain the standard uniform backoff process for the tagged node}; this approach is akin to the modeling in Bianchi's work \cite{bianchi00performance}.\hfill \Square
% \noindent
% \textbf{(A3)} Node~$i$ samples its successive back-offs from a uniform distribution, as in the standard. When a new backoff cycle starts for Node~$i$, if $X_v=0_s$ (respectively $X_v\neq 0_s$), the other node, say Node~$j$, conditioned on being in backoff, attempts independently in each slot with probability $\bd$ (respectively, $\bc$) \emph{until the end of the first transmission cycle within this backoff cycle}.\hfill \Square
% 
% \noindent
% \textbf{(A4)} If Node~$i$ is interrupted within a backoff cycle, thus freezing its backoff in the first transmission cycle within the backoff cycle (see Figure~\ref{fig:bo-cycle-explain}), then from the point of interruption until the backoff completion of Node~$i$, Node~$j$ attempts independently in each slot with probability $\bs$, conditioned on being in backoff\footnote{(A3) and (A4) have to be slightly modified for $n>2$. See Section~VIII in \cite{techreport}. Under the modified assumptions, $\{(S_v,X_v,M_v),T^{\prime}_v\}$ is again an MRP, and the analysis can be carried out.}.\hfill \Square 

Under assumptions~(A3)-(A4), observe that the process $\{(S_v,N_v), T^{\prime (i)}_v\}$ is a Markov Renewal process (MRP), with the state space of the embedded Markov chain being $\{0,\ldots,K\}\times\{1,\ldots,n\}$. 

It can be shown that the embedded Markov chain has a unique stationary distribution, denoted by $\psi$. We defer the detailed derivation of this stationary distribution to the Appendix. We discuss next, how we can compute the attempt rates $\bd,\bc$ and $\bs$ given $\psi$. 
 
Recall that $\bs$ and $\bc$ are the mean attempt rates of a node in a transmission cycle after it resumes backoff following a succeessful transmission, and a collision, respectively, while $\bd$ is the mean attempt rate of a node in a transmission cycle after it resumes backoff following an interruption. Thus, observe that in a backoff cycle of a tagged node, the contributions to $\bs$ and $\bc$ come from only the first transmission cycle within the backoff cycle, whereas the remainder (if any) of the backoff cycle contributes towards $\bd$. 

\subsubsection{Computation of $\bd$}
\label{subsubsec:bd-general-n} 

Looking at the backoff evolution of the tagged Node~$i$ (see Figure~\ref{fig:bo-cycle-explain}), we can define $\bd$ more formally as

\begin{equation*} 
\bd = \lim_{t\to\infty}\frac{\sum_{k=1}^{N(t)}\ind_{\{\text{Node~$i$ interrupted in backoff cycle $k$}\}}}{\sum_{k=1}^{N(t)}B_{r,k}}
\end{equation*}
where, $N(t)$ is the number of backoff cycles until time $t$, and $B_{r,k}$ is the \emph{residual backoff} to be counted by Node~$i$ from the point of first interruption until its backoff completion in backoff cycle $k$ provided that it was interrupted; $B_{r,k}=0$ if Node~$i$ was not interrupted in backoff cycle $k$. It suffices to count the residual backoff from first interruption to backoff completion since the node does not sample any fresh backoff in between, and any intermediate interruption will find the node counting parts of the same residual backoff. Thus, the denominator is the total residual backoff counted by Node~$i$ until time $t$ \emph{after being interrupted}. The numerator is the total number of attempts made by Node~$i$ until time $t$ upon completion of its residual backoff countdown after interruptions. Note that by our definition of backoff cycles, each backoff cycle must end with an attempt by Node~$i$; the indicator function simply tracks whether the attempt followed an interruption or not. 

Denote by $\Bcal_r$, the random variable representing the residual backoff counted by Node~$i$ from the point of first interruption until its backoff completion in a backoff cycle. Then, by Markov Regenerative theory, 

\begin{equation}
 \bd = \frac{\sum_{(s,n_a)}\psi(s,n_a)P_I(s,n_a)}{\sum_{(s,n_a)}\psi(s,n_a)E\Bcal_r(s,n_a)}\: a.s\label{eqn:bd-expression-general-n}
\end{equation}
where, $P_{I}(s,n_a)$ is the probability that Node~$i$ is interrupted when the backoff completion cycle starts in state $(s,n_a)$, and $E\Bcal_r(s,n_a)$ is the mean residual backoff counted by Node~$i$ from its first interruption until its backoff completion in a backoff cycle that started with state $(s,n_a)$; they can be computed as follows.

\noindent

\emph{Computation of $P_{I}(\cdot,\cdot)$:}
%\label{subsubsec:compute-prob-interruption}

\noindent
When the backoff cycle starts in state $(s,n_a)$, we know from (A3) that during the first transmission cycle within this backoff cycle, $(n_a-1)$ nodes will attempt w.p. $\bc$ in each slot, and the remaining $(n-n_a)$ nodes (that did not attempt in the previous cycle) will attempt w.p. $\bd$ in each slot. Suppose Node~$i$ samples a backoff of $l$ slots uniformly from $[1,W_s]$. Then, Node~$i$ will be interrupted if at least one of the other nodes attempts within the first $(l-1)$ slots. This happens with probability $1-((1-\bc)^{n_a-1}(1-\bd)^{n-n_a})^{l-1}$. Thus, we have

\begin{align}
 P_I(s,n_a) &=\frac{1}{W_s}\sum_{l=1}^{W_s}\bigg[1-((1-\bc)^{n_a-1}(1-\bd)^{n-n_a})^{l-1}\bigg] \label{eqn:P-I-s-na}
\end{align}
for all $s\in\{0,\ldots,K\}$, $n_a\in\{1,\ldots,n\}$. 

\noindent
\emph{Computation of $E\Bcal_r(s,n_a)$:}

Consider a backoff cycle starting with state $(s,n_a)$. Suppose Node~$i$ samples (uniformly from $\{1,1,\ldots, W_s\}$) a backoff of $l$ slots. As was explained earlier, to interrupt Node~$i$, at least one other node must make an attempt by slot $l-1$. Suppose one or more of the other nodes make an attempt at slot $w$, $1\leq w\leq l-1$; this happens with probability $((1-\bc)^{n_a-1}(1-\bd)^{n-n_a})^{w-1}(1-(1-\bc)^{n_a-1}(1-\bd)^{n-n_a})$. In this case, the residual backoff of Node~$i$ is $(l-w)$. Thus, we have, for any $s\in\{0,\ldots,K\}$, and $n_a\in\{1,\ldots,n\}$,
\begin{align}
 E\Bcal_r(s,n_a) &= \frac{1}{W_s}\sum_{l=1}^{W_s}\sum_{w=1}^{l-1}(l-w)\times((1-\bc)^{n_a-1}(1-\bd)^{n-n_a})^{w-1}\nonumber\\
&\times(1-(1-\bc)^{n_a-1}(1-\bd)^{n-n_a})\label{eqn:EB-r-s-na}
\end{align}

\subsubsection{Computation of $\bs$}
\label{subsubsec:bs-general-n}
Looking at the backoff evolution of the tagged Node~$i$, we can define $\bs$ more formally as
\begin{equation*}
 \bs = \lim_{t\to\infty}\frac{\sum_{k=1}^{N_s(t)}\ind_{\{\text{Node~$i$ was not interrupted in backoff cycle $k$}\}}}{\sum_{k=1}^{N_s(t)}B_{s,k}}
\end{equation*}
 where, $N_s(t)$ is the number of backoff cycles until time $t$ that start with the state $(0,1)$ (implying that Node~$i$ was successful in the previous transmission cycle), and $B_{s,k}$ is the backoff counted by Node~$i$ \emph{in the transmission cycle that started along with backoff cycle $k$}; in other words, $B_{s,k}$ is the backoff counted by Node~$i$ until it gets interrupted, or completes its backoff, whichever is earlier. Thus, the denominator is the total backoff counted by Node~$i$ until time $t$, in those transmission cycles that followed a successful transmission by Node~$i$. Similarly, the numerator is the total number of attempts by Node~$i$ until time $t$ in those transmission cycles that followed a successful transmission by Node~$i$. 

Denote by $\Bcal_s$, the random variable representing the backoff counted by Node~$i$ in the first transmission cycle within a backoff cycle starting in state $(0,1)$. Then, by Markov regenerative theory, it follows that

\begin{equation}
 \bs = \frac{1-P_I(0,1)}{E\Bcal_s(0,1)} \: a.s.\label{eqn:bs-expression-general-n}
\end{equation}
where, $E\Bcal_s(0,1)$ is the mean time spent in backoff by Node~$i$ until it gets interrupted, or completes its backoff in the backoff cycle starting in state $(0,1)$, and can be computed as follows.

Suppose Node~$i$ samples (uniformly from $\{1,\ldots, W_0\}$) a backoff of $l$ slots. To interrupt Node~$i$, at least one of the other nodes must attempt within slot $(l-1)$. Now there are two possibilities:
\begin{enumerate}
 \item None of the other nodes attempt up to slot $(l-1)$. Then Node~$i$ does not get interrupted, and its backoff count is $l$. This happens with probability $(1-\bd)^{(n-1)(l-1)}$.
 \item One or more of the other nodes attempt at slot $w$, $1\leq w\leq l-1$. Then, Node~$i$ is interrupted, and its backoff counted until interruption is $w$. This happens with probability $(1-\bd)^{(n-1)(w-1)}(1-(1-\bd)^{n-1})$. 
\end{enumerate}

Combining all of these together,
\begin{align}
 E\Bcal_s(0,1) &= \frac{1}{W_0}\sum_{l=1}^{W_0}\bigg[(1-\bd)^{(n-1)(l-1)}l\nonumber\\
&+\sum_{w=1}^{l-1}w(1-\bd)^{(n-1)(w-1)}(1-(1-\bd)^{n-1})\bigg]\label{eqn:EB-s-general-n}
\end{align}

\subsubsection{Computation of $\bc$}
\label{subsubsec:bc-general-n}

Looking at the backoff evolution of the tagged Node~$i$, we can define $\bc$ more formally as
\begin{equation*}
 \bc = \lim_{t\to\infty}\frac{\sum_{k=1}^{N_c(t)}\ind_{\{\text{Node~$i$ was not interrupted in backoff cycle $k$}\}}}{\sum_{k=1}^{N_c(t)}B_{c,k}}
\end{equation*}
 where, $N_c(t)$ is the number of backoff cycles until time $t$ that start with states other than $(0,1)$ (implying that Node~$i$ encountered a collision in the previous transmission cycle), and $B_{c,k}$ is defined as the backoff counted by Node~$i$ \emph{in the transmission cycle that started along with backoff cycle $k$}; in other words, $B_{c,k}$ is the backoff counted by Node~$i$ until it gets interrupted, or completes its backoff, whichever is earlier. Thus, the denominator is the total backoff counted by Node~$i$ until time $t$, in those transmission cycles that followed a collision by Node~$i$. Similarly, the numerator is the total number of attempts by Node~$i$ until time $t$ in those transmission cycles that followed a collision by Node~$i$. 

Denote by $\Bcal_c$, the random variable representing the backoff counted by Node~$i$ in the first transmission cycle following a collision involving Node~$i$. Then, by Markov regenerative theory, it follows that
\begin{equation}
 \bc = \frac{\sum_{(s,n_a)\neq (0,1)}\psi(s,n_a)(1-P_I(s,n_a))}{\sum_{(s,n_a)\neq (0,1)}\psi(s,n_a)E\Bcal_c(s,n_a)}\: a.s\label{eqn:bc-expression-general-n}
\end{equation}
where, $E\Bcal_c(s,n_a)$ is the mean time spent in backoff by Node~$i$ until it gets interrupted, or completes its backoff in the backoff cycle starting in state $(s,n_a)$, and can be computed as follows.

Suppose Node~$i$ samples (uniformly from $\{1,\ldots, W_s\}$) a backoff of $l$ slots. As explained earlier, to interrupt Node~$i$, at least one of the other nodes must make an attempt by slot $l-1$. Now, there are two possibilities:
\begin{enumerate}
 \item None of the other nodes attempt up to slot $(l-1)$. Node~$i$ does not get interrupted, and its backoff count is $l$. This happens with probability $((1-\bc)^{n_a-1}(1-\bd)^{n-n_a})^{l-1}$. 
 \item One or more of the other nodes attempt at slot $w$, $1\leq w\leq (l-1)$. Then, Node~$i$ is interrupted, and its backoff count until interruption is $w$. This happens with probability $((1-\bc)^{n_a-1}(1-\bd)^{n-n_a})^{w-1}(1-(1-\bc)^{n_a-1}(1-\bd)^{n-n_a})$.
\end{enumerate}

Combining these together, we have, for any $n_a\in\{2,\ldots,n\}$, and any $s\in\{0,\ldots,K\}$,
\begin{align}
 E\Bcal_c(s,n_a) &= \frac{1}{W_s}\sum_{l=1}^{W_s}\bigg[l((1-\bc)^{n_a-1}(1-\bd)^{n-n_a})^{l-1}\nonumber\\
&+ \sum_{w=1}^{l-1}w((1-\bc)^{n_a-1}(1-\bd)^{n-n_a})^{w-1}\nonumber\\
&\times(1-(1-\bc)^{n_a-1}(1-\bd)^{n-n_a})\bigg]\label{eqn:EB-c-s-na}
\end{align}

Equations~\ref{eqn:bd-expression-general-n}-\ref{eqn:EB-c-s-na} along with the expressions for the stationary probabilities $\psi(s,n_a)$ (derived in the Appendix) form a system of vector fixed point equations in $(\bd,\bc)$ (observe from Eqns.~\ref{eqn:P-I-s-na} and \ref{eqn:EB-s-general-n} that $\bs$ is a deterministic function of $\bd$ alone), which can be solved using an iterative procedure until convergence to obtain the attempt rates $\bd$, $\bs$, and $\bc$.

\subsubsection{Computation of the average attempt rate, $\beta$, over all backoff time}
\label{subsubsec:beta-general-n}
The backoff cycle analysis can be used to obtain the long run average attempt rate, $\beta$, averaged over all backoff time (irrespective of system state). This is the quantity that was used in the fixed point approximation proposed in \cite{bianchi00performance,kumar-etal04new-insights}; see Section~\ref{subsec:bianchi-analysis}. 

To obtain $\beta$, note that each backoff cycle contains exactly one attempt by the tagged node, and the backoff counted by the tagged node in the entire backoff cycle contributes towards $\beta$. In a backoff cycle starting in state $(s,n_a)$, the mean backoff counted by the tagged node is clearly $(W_s+1)/2$. Thus, using Markov regenerative analysis, we have
\begin{equation}
 \beta = \frac{1}{\sum_{(s,n_a)}\psi(s,n_a)\frac{W_s+1}{2}}\label{eqn:beta-wifi-general-n}
\end{equation}

\subsection{Discussion on the existence and uniqueness of the fixed point}
\label{subsec:existence}
\begin{theorem}
\label{thm:existence-general-n}
There exists a fixed point for the system of equations~\ref{eqn:bd-expression-general-n}-\ref{eqn:EB-c-s-na} in the set $\mathbf{C}=[1/W_K,1]\times[1/W_K,1]$. 
\end{theorem}

\begin{proof}
Observe that all the functions involved in the system of equations are continuous in $(\bd,\bc)$ when $(\bd,\bc)\in \mathbf{C}$. We need to show that the system of equations maps the set $\mathbf{C}$ into itself. 

Suppose we start an iteration of the fixed point equations with $(\bd^{(0)},\bc^{(0)})\in\mathbf{C}$. Consider a simulation of the process evolution at a tagged node \emph{obeying the approximations (A3)-(A4)}, under attempt rates $\bd^{(0)}$, $\bc^{(0)}$, and the corresponding $\bs^{(0)}$. Note that the system of equations~\ref{eqn:bd-expression-general-n}-\ref{eqn:EB-c-s-na} is an exact representation of this simulated system under attempt rates $\bd^{(0)}$, $\bc^{(0)}$, and $\bs^{(0)}$. For $(\bd^{(0)},\bc^{(0)})\in\mathbf{C}$, it can be observed that the transition probabilities given by Eqns.~\ref{eqn:trans-prob-gen-n-first-eqn}-\ref{eqn:transition-prob-gen-n-last-eqn} are positive, so that the embedded Markov chain in the evolution of the tagged node is finite, irreducible, and hence positive recurrent. Furthermore, it can be observed from Eqn.~\ref{eqn:P-I-s-na} that $0<P_I(\cdot,\cdot)<1$ for all states when $(\bd^{(0)},\bc^{(0)})\in\mathbf{C}$. It follows that the tagged node gets interrupted infinitely often in the simulated system. Consider the quantity 

\begin{equation*}
 \lim_{t\to\infty}\frac{N_I(t)}{\sum_{k=1}^{N_I(t)}B_{r,k}}
\end{equation*}
in the simulated system, where $N_I(t)$ is the number of backoff cycles until time $t$ where the tagged node was interrupted, and $B_{r,k}$ is the \emph{residual backoff} of Node~$i$ in backoff cycle $k$ when it is first interrupted. Since the maximum backoff sampled by the tagged node in any cycle is $W_K$, it follows that $1\leq B_{r,k}\leq W_K$. Hence, the above quantity is lower bounded by $1/W_K$, and upper bounded by 1. But by Markov regenerative theory, this quantity is \emph{almost surely} equal to the R.H.S of Eqn.~\ref{eqn:bd-expression-general-n}, which is nothing but our estimate for the next iterate $\bd^{(1)}$. Thus, we have $1/W_K \leq \bd^{(1)}\leq 1$. 

Similarly, we can argue that $1/W_K\leq \bc^{(1)}\leq 1$. Thus, it follows that the system of equations~\ref{eqn:bd-expression-general-n}-\ref{eqn:EB-c-s-na} map $(\bd^{(0)},\bc^{(0)})\in\mathbf{C}$ to $(\bd^{(1)},\bc^{(1)})\in\mathbf{C}$, as desired. 

Thus, the system of fixed point equations is a continuous mapping from the closed, bounded, convex set $\mathbf{C}$ to $\mathbf{C}$. Hence, it follows from Brouwer's Fixed Point theorem that a fixed point exists for the system in $\mathbf{C}$. 
\end{proof}

We do not have proof of uniqueness of the fixed point. However, in our numerical experiments, the iterations always converged to the same solutions (within a tolerance of $10^{-8}$) even when starting with different initial values. 

\section{Model Validation Through Simulations}
\label{sec:numerical-general-n}

To validate our analytical model, we performed extensive simulations with four different backoff sequences, each chosen so as to ensure that the resulting system exhibits short term unfairness (due to large variability in backoff) for low to moderate number of nodes, and the standard fixed point analysis does not work. Henceforth, we shall call these backoff sequences as test sequences. These test sequences are summarized in Table~\ref{tbl:test-sequences}.
\begin{table*}[ht]
  \centering
\caption{Details of the Test backoff sequences}
\label{tbl:test-sequences}
%\scriptsize
  \begin{tabular}{|c|c|}\hline
    Test sequence & Parameters\\    
   \hline
   1 & $K=7$, $b_0=1$,$b_k=3^kb_0$\\   
   \hline
   2 & $K=7$, $b_0=\cdots=b_3=1.5$, $b_4=\cdots=b_7=64$\\   
  \hline
   3 & $K=1$, $b_0=1.5$, $b_1=32.5$\\ 
  \hline
   4 & $K=6$, $b_0=\cdots=b_3=1.5$, $b_4=\cdots=b_6=32.5$\\
  \hline
\end{tabular}
\normalsize
%\vspace{-3mm}
\end{table*} 
Note that test sequences 1 and 3 are the same as the backoff sequences introduced in Examples~1 and 3 respectively in Sections~\ref{subsec:example1} and \ref{subsec:example3}. Test sequence~2 is almost identical to the backoff sequence in Example~2, Section~\ref{subsec:example2}, except that we have made the initial mean backoff 1.5 instead of 1. This change was made for the following reason: a mean initial backoff of 1 implies $\bs=1$, and in this case, one can easily check that our approximate analysis always predicts $\bs$ exactly (see also, the plots in Figure~\ref{fig:approx-model-validation-seq1} for Test sequence~1, where we retained $b_0=1$). In order to verify the accuracy of the analysis for non trivial values of $\bs$, we chose the mean intial backoff to be 1.5 instead of 1. This will also help to demonstrate that to cause short term unfairness, a deterministic initial backoff is not necessary. Finally, test sequence~4 is a new backoff sequence, which is a modified version of test sequence~3 (and Example~3 in Section~\ref{subsec:example3}) with a higher retry limit. 

In Table~\ref{tbl:test-sequences},we have not introduced a test sequence corresponding to Example~4 in Section~\ref{subsec:example4}, since there the short term unfairness was insignificant beyond $n=3$. We have, however, verified that the performance measures predicted by our analysis match well with simulations even for this case; e.g., the collision probability predicted by the analysis for $n=2$ for the backoff sequence in Section~\ref{subsec:example4} is 0.7754, whereas that obtained from simulations is 0.7325, an error of about 5\%.

For each test sequence in Table~\ref{tbl:test-sequences}, we performed simulations for a range of values of the number of nodes, $n$. In particular, for all the test sequences, we considered systems with $n=2$ to $n=10$. In addition, for test sequence~1 (which exhibits particularly severe short term unfairness; compare Figures~\ref{fig:short-term-unfairness-example1} and \ref{fig:gamma-vs-n-example1} with Figures~\ref{fig:short-term-unfairness-example2}-\ref{fig:gamma-vs-n-example4}), we also considered systems with $n=20,40,60,80$, and 100. This gives us a total of 41 test cases.  

For each test case, we used the method of simulating the detailed Markov renewal model, described in Section~\ref{sec:exact_model}, since it is much faster compared to detailed ``off-the-shelf'' event-driven simulation tools such as Qualnet. As remarked at the end of Section~\ref{sec:exact_model}, this model is equivalent to the DTMC model introduced in \cite{kumar-etal04new-insights} (see also Section~\ref{subsec:markov-model-dcf}), which is known to give excellent accuracy (see, for example, \cite{ramaiyan-etalYYfp-analysis}). This also provides us more flexibility in examining the finer details of the system evolution (e.g., it is considerably harder to obtain the conditional attempt rates such as $\bd$ from a Qualnet simulation).

The long run average collision probability, $\gamma$, is obtained from the simulations using the method outlined in Section~\ref{subsec:example1}. To obtain the attempt rates $\bd,\bs,\bc$ from the simulations, we followed the formal definitions of $\bd,\bs,\bc$ introduced in Section~\ref{subsec:tagged-node-evolution} in the context of computing these rates. More precisely, for each node~$i$, we estimated its mean attempt rate following an interruption, $\bd^{(i)}$ as 
\begin{equation} 
\bd^{(i)} = \frac{\sum_{k=1}^{N_i}\ind_{\{\text{Node~$i$ interrupted in backoff cycle $k$}\}}}{\sum_{k=1}^{N_i}B_{r,k}} \label{eqn:bd-from-sim}
\end{equation}
where, $N_i$ is the number of backoff cycles of Node~$i$ during the entire simulation, and $B_{r,k}$ is the \emph{residual backoff} counted by Node~$i$ from the point of first interruption until its backoff completion in backoff cycle $k$ provided that it was interrupted; $B_{r,k}=0$ if Node~$i$ was not interrupted in backoff cycle $k$. As explained in Section~\ref{subsubsec:bd-general-n}, the denominator is the total residual backoff counted by Node~$i$ \emph{after being interrupted}. The numerator is the total number of attempts made by Node~$i$ upon completion of its residual backoff countdown after interruptions. 

If the simulation duration is long enough, then, due to symmetry, $\bd^{(i)}\approx \bd^{(j)}$ for all $i\neq j\in\{1,\ldots,n\}$. This was observed in all our simulations. Finally, we estimated $\bd$ as $\bd=\frac{1}{n}\sum_{i=1}^n \bd^{(i)}$. Similar methods yield $\bs$, and $\bc$ from the simulations. 

The results for test sequences~1 to 4 are summarized in Figures~\ref{fig:approx-model-validation-seq1}-\ref{fig:approx-model-validation-seq4}, where we have compared the collision probabilities, throughputs, and attempt rates obtained from the approximate Markov renewal analysis (henceforth, also called the MRP analysis) against those obtained from simulations. In case of collision probability, we also compared the corresponding values obtained from the Bianchi/Mean field analysis (\cite{bianchi00performance,kumar-etal04new-insights}; see also Section~\ref{subsec:bianchi-analysis}).

\begin{figure*}[htp]
%\footnotesize
\begin{center}
\includegraphics[scale=0.35]{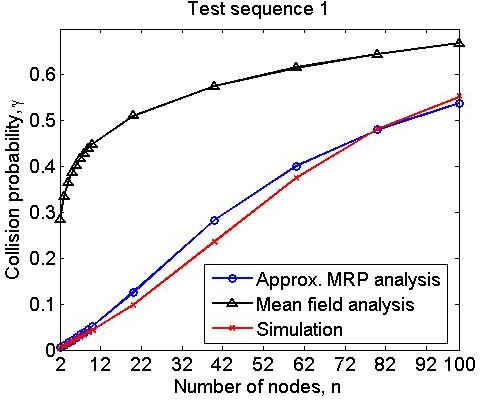}
\hspace{1mm}
\includegraphics[scale=0.35]{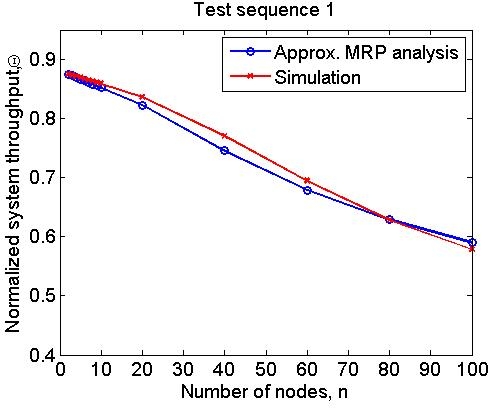}
\vspace{2mm}
\includegraphics[scale=0.33]{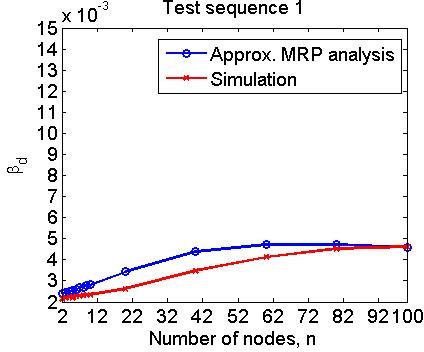}
\hspace{1mm}
\includegraphics[scale=0.32]{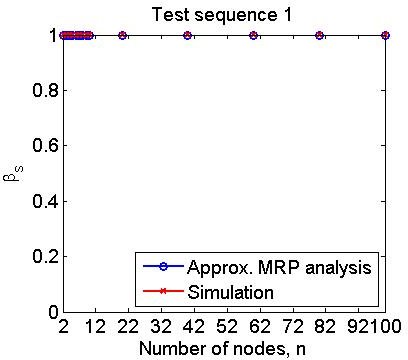}
\hspace{1mm}
\includegraphics[scale=0.33]{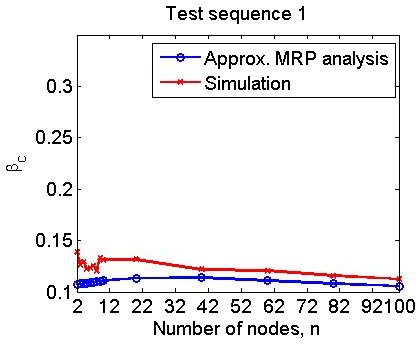}
\caption{Test sequence~1 ($K=7$, $b_0=1$, $b_k=3^kb_0$): Comparison of collision probabilities, throughputs, and attempt rates obtained from the approximate analytical model against simulations for various $n$.}
\label{fig:approx-model-validation-seq1}
\vspace{-5mm}
\end{center}
\normalsize
%\vspace{-6mm}
\end{figure*}

\begin{figure*}[htp]
%\footnotesize
\begin{center}
\includegraphics[scale=0.32]{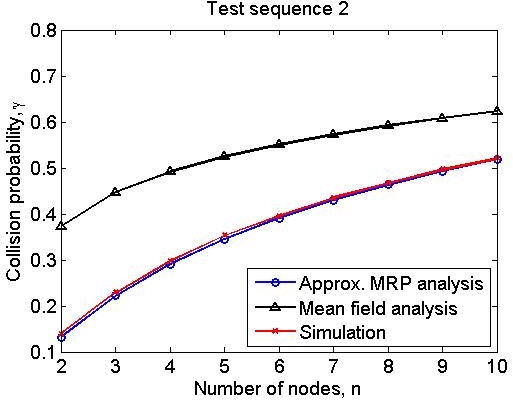}
\hspace{1mm}
\includegraphics[scale=0.32]{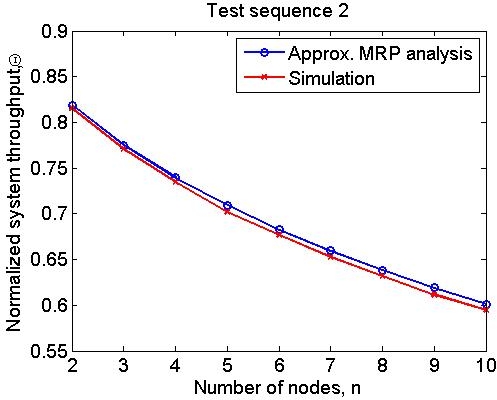}
\vspace{2mm}
\includegraphics[scale=0.28]{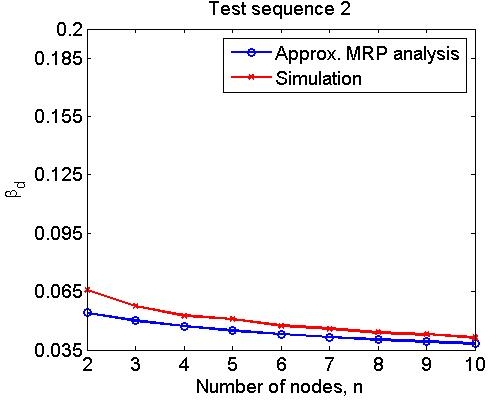}
\hspace{0.1mm}
\includegraphics[scale=0.3]{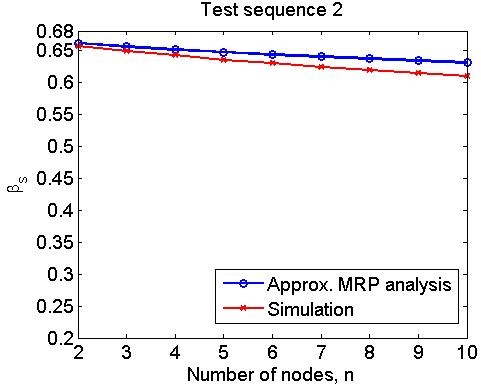}
\hspace{0.1mm}
\includegraphics[scale=0.33]{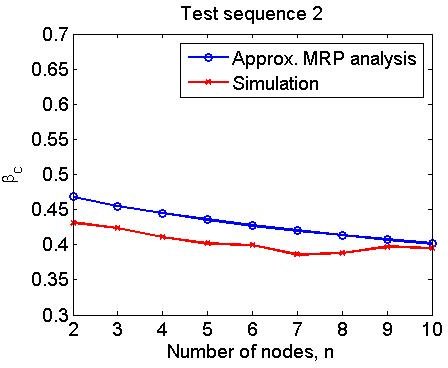}
\caption{Test sequence~2 ($K=7$, $b_0=\cdots=b_3=1.5$, $b_4=\cdots=b_7=64$): Comparison of collision probabilities, throughputs, and attempt rates obtained from the approximate analytical model against simulations for various $n$.}
\label{fig:approx-model-validation-seq2}
\vspace{-5mm}
\end{center}
\normalsize
%\vspace{-6mm}
\end{figure*}

\begin{figure*}[htp]
%\footnotesize
\begin{center}
\includegraphics[scale=0.34]{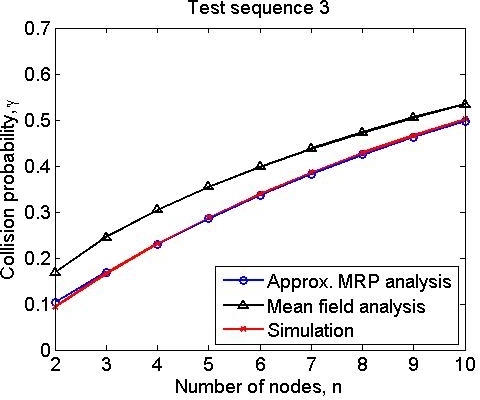}
\hspace{1mm}
\includegraphics[scale=0.36]{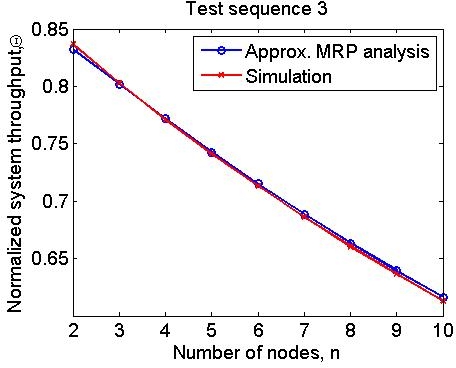}
\vspace{2mm}
\includegraphics[scale=0.31]{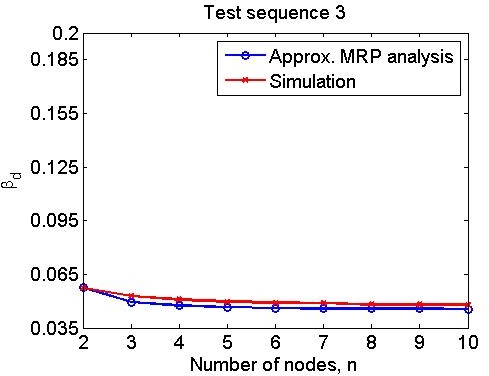}
\hspace{0.1mm}
\includegraphics[scale=0.3]{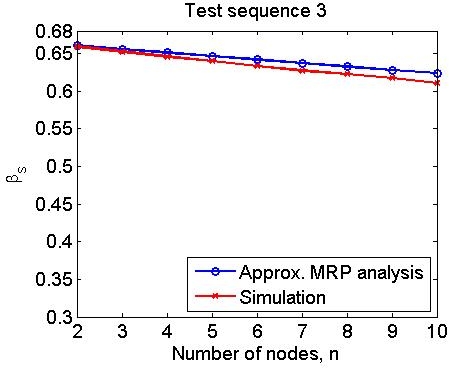}
\hspace{0.1mm}
\includegraphics[scale=0.32]{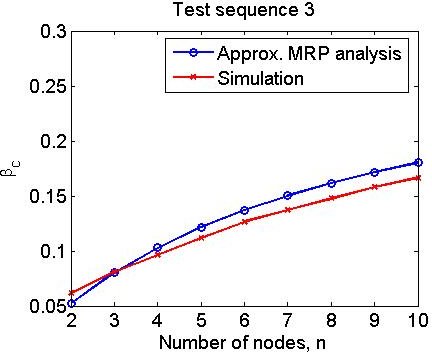}
\caption{Test sequence~3 ($K=1$, $b_0=1.5$, $b_1=32.5$): Comparison of collision probabilities, throughputs, and attempt rates obtained from the approximate analytical model against simulations for various $n$.}
\label{fig:approx-model-validation-seq3}
\vspace{-5mm}
\end{center}
\normalsize
%\vspace{-6mm}
\end{figure*}

\begin{figure*}[htp]
%\footnotesize
\begin{center}
\includegraphics[scale=0.33]{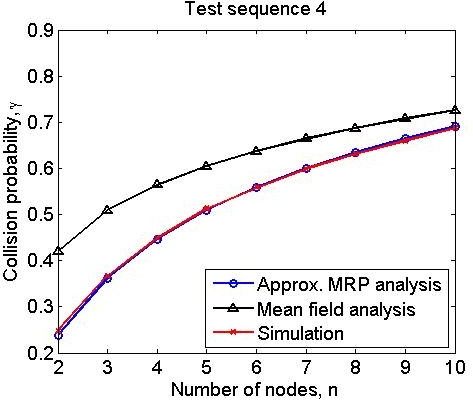}
\hspace{1mm}
\includegraphics[scale=0.355]{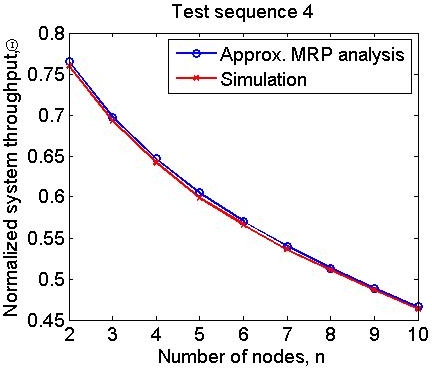}
\vspace{2mm}
\includegraphics[scale=0.3]{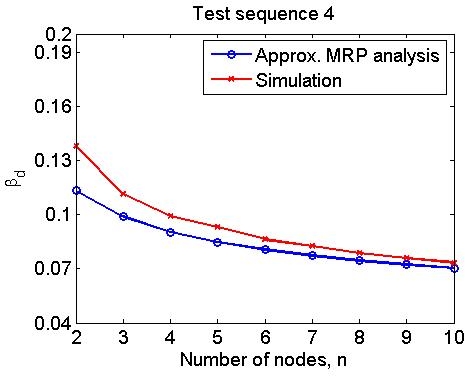}
\hspace{0.1mm}
\includegraphics[scale=0.31]{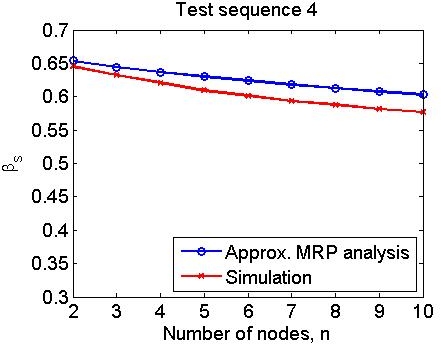}
\hspace{0.1mm}
\includegraphics[scale=0.29]{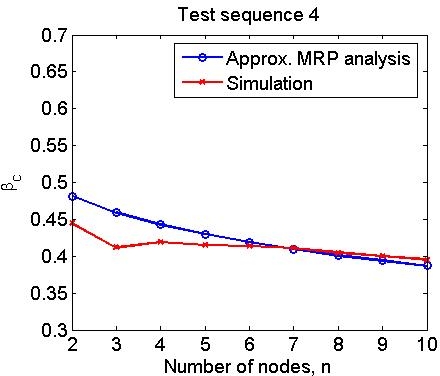}
\caption{Test sequence~4 ($K=6$, $b_0=\cdots=b_3=1.5$, $b_4=\cdots=b_6=32.5$): Comparison of collision probabilities, throughputs, and attempt rates obtained from the approximate analytical model against simulations for various $n$.}
\label{fig:approx-model-validation-seq4}
\vspace{-5mm}
\end{center}
\normalsize
%\vspace{-6mm}
\end{figure*}

% \begin{figure*}[t]
% %\footnotesize
% \begin{center}
% \includegraphics[scale=0.3]{plots/throughput_vs_prop_delay_sigma20_K6_cwmax10.jpg}
% \hspace{0.1mm}
% \includegraphics[scale=0.3]{plots/throughput_vs_prop_delay_sigma60_K6_cwmax10.jpg}
% \hspace{0.1mm}
% \includegraphics[scale=0.3]{plots/throughput_vs_slot_durn_Delta60_K6_cwmax10.jpg}
% \hspace{0.1mm}
% \includegraphics[scale=0.3]{plots/throughput_vs_slot_durn_Delta120_K6_cwmax10.jpg}
% \hspace{0.1mm}
% \includegraphics[scale=0.3]{plots/throughput_vs_slot_durn_Delta160_K6_cwmax10.jpg}
% \caption{Comparison of throughputs obtained from the approximate analytical model against simulations for different combinations of propagation delay, and slot duration.}
% \label{fig:approx-model-theta-validation-general-n}
% \vspace{-5mm}
% \end{center}
% \normalsize
% %\vspace{-6mm}
% \end{figure*} 

From these plots, we can make the following observations:

\noindent
\textbf{Observations:}

\begin{enumerate}
 \item For Test sequence~1, the MRP analysis predicts the collision probability within an error of at most 12\% compared to simulations. Recall from the discussion earlier in this section that this is the sequence that exhibits the most severe short term unfairness among all the examples. In all the other test cases, the MRP analysis predicts the collision probability with excellent accuracy (within 1-2\% error). The mean field analysis, on the other hand, is quite inaccurate in all the test cases.
 \item The MRP analysis also predicts the throughput within an error of at most 2-3\%.
 \item The errors in the MRP analysis compared to simulations are at most 10-14\% in predicting the attempt rates, $\bd,\bs,$ and $\bc$. For all test sequences, the qualitative trends in the attempt rates as a function of $n$ are captured by the MRP analysis.
 \item For all test sequences, the collision probability, $\gamma$, increases with the number of nodes, $n$, as expected.
 \item For all test sequences, the normalized system throughput, $\Theta$, decreases with increasing $n$. An intuition behind this is as follows: as $n$ increases, collision probability increases, causing the nodes to sample backoffs from stochastically larger contention windows. Since the nodes sample stochastically larger backoffs, their attempt rates decrease, increasing the idle time in the system. Moreover, the increase in the number of collisions also tends to waste more slots. Hence throughput decreases with increasing $n$.
 \item In all test cases, $\bs\gg \bd$, i.e., the attempt rate is skewed in favor of the successful node, \emph{a reflection of the short term unfairness property}.
 \item In the presence of short term unfairness, the collision probability predicted by the mean field analysis is always larger compared to that obtained from simulations. This is because in the presence of short term unfairness, the last successful node has a much larger probability of accessing the channel in the next slot than the other nodes, thus further boosting its success probability, unlike in a fair system, where all the nodes have comparable probability of accessing the channel, resulting in a higher probability of collision. \emph{The mean field analysis ignores the correlation in the system evolution in an unfair system.}
 \item Consider the following question: what is the probability that a tagged node gets interrupted in a backoff cycle (i.e., at least one other node samples a backoff, or has a residual backoff, smaller than the tagged node)? Recall that if it gets interrupted, we have a contribution towards $\bd$, and if it does not get interrupted, we have a contribution towards $\bs$ or $\bc$ (see Eqns~\ref{eqn:bd-expression-general-n},\ref{eqn:bs-expression-general-n}, and \ref{eqn:bc-expression-general-n}). 

As the number of nodes increases, this probability is influenced in two ways: 

\noindent
(i) if we hold the contention windows of the nodes fixed, then intuitively, as $n$ increases, this probability of interruption should increase. 

\noindent
(ii) However, as $n$ increases, the collision probability increases, causing the contention windows, and hence the sampled backoffs of the nodes to be stochastically larger; this causes the probability of interruption to decrease. Thus, there is a trade-off. 

Further note that since after a success, the tagged node always samples from the lowest contention window, irrespective of the number of nodes in the system, a decrease in the probability of interruption will increase $\bs$, while an increase in the probability of interruption will decrease $\bs$. Thus, in general, one would expect to see a non-monotonic variation in $\bs$ with the number of nodes, when the initial backoff is stochastic, i.e., $b_0\neq 1$. However, for test sequences~2, 3, and 4, we see $\bs$ decreasing monotonically as $n$ increases. This is because for these test sequences, as the backoff stages become larger, the contention windows do not change (e.g., in test sequence~2, $b_4=\cdots=b_7=64$), thus effectively causing influence~(i) above to be in force. 

Similar argument indicates that it is hard to intuitively predict the trend in $\bc$ as a function of $n$, and in general, it may be non-monotonic. 

Let us now focus on $\bd$. For test sequences~2, 3, and 4, we see that $\bd$ decreases with increasing $n$. We shall give an intuition for this for test sequence~3. Similar intuition works for test sequences~2 and 4. Note that the contention window size of a tagged node is either $W_0=2$, or $W_1=64$. Each interruption of the tagged node in backoff stage 0 contributes 1 to the numerator of Eqn.~\ref{eqn:bd-from-sim}, and 1 to the denominator of Eqn.~\ref{eqn:bd-from-sim}. However, each interruption in backoff stage~1 contributes 1 to the numerator, and a value typically much larger than 1 to the denominator. As the number of nodes increases, collision probability increases, pushing the tagged node to backoff stage~1 faster, thus causing more contributions to $\bd$ from backoff stage~1 than from backoff stage~0. Furthermore, the probability of interruption (and hence the number of interruptions) also increases, as explained earlier; majority of these interruptions happen in backoff stage~1 as just argued. Thus, increase in the denominator far exceeds that in the numerator, causing $\bd$ to decrease with increasing $n$. 

For test sequence~1, however, $\bd$ initially increases slightly, and then flattens off with increasing $n$. One possible explanation for this is as follows: for test sequence~1, the backoff sequence (contention window size) builds up in a more gradual manner than test sequences~2-4; in particular, one can imagine that as $n$ increases, initially, influence~(ii) explained above is in force, causing nodes to sample backoffs from \emph{stochastically larger contention windows}, and the \emph{probability of interruption to decrease}. These two effects together cause $\bd$ to increase slightly. However, when $n$ becomes sufficiently large, further increase in the contention window size has negligible effect, and influence~(i) explained above comes into play, causing the probability of interruption to rise again. This causes the $\bd$ curve to flatten off. 

\item On a Linux based machine with 8 GB RAM, the running time of the MRP analysis is several seconds, while that of the stochastic simulation is of the order of several minutes; it takes hours to run the Qualnet simulation, especially when the short term unfairness is severe. 
\end{enumerate}

%%%%%%%%%%%%%%%%%%%%%%%%%%      Part~II: Large Propagation Delays             %%%%%%%%%%%%%%%%%%%%%%%%%%%%%%%%%%%%%%%%%%%%%%%%%%%%%%%%%%%%%%%%%%%%%%%%%%%%%%%%%%%%%%%%%%%%%%%%%%%%

\newpage
\begin{center}
 \textbf{\Large{Part~II: Large Propagation Delays}}
\end{center}
\normalsize

We had mentioned earlier in Section~\ref{sec:intro} that the IEEE~802.11 DCF mechanism is finding its way into new applications such as rural broadband access, and long distance UAV communications, where the propagation delay among the participating nodes are not negligible compared to the duration of a backoff slot, unlike the conventional WiFi. We shall demonstrate later in this chapter that the phenomenon of short term unfairness which was observed earlier in DCF based systems under non-standard backoff sequences, is also observed in the large propagation delay setting, \emph{even under the default protocol parameters of IEEE~802.11 standard}. 

As mentioned earlier, Simo-Reigadas et al.\ \cite{simo-reigadas10wild} aimed to develop an extension of the Bianchi model to predict the performance of IEEE~802.11 DCF with large propagation delays. However, they adopt a state independent, constant attempt rate approximation as in the Bianchi analysis \cite{bianchi00performance,kumar-etal04new-insights}, which ignores the short term unfairness property, and as a consequence, the \emph{collision/success probabilities computed using the analysis are inaccurate compared to simulation results} (see also Section~\ref{subsubsec:simo-reigadas-performance}). 

In this part, we aim to extend the analysis developed for general backoff sequences in Part~I to the case of systems with large propagation delays. We focus on the case where the transmitters are equidistant from one another, and also each receiver is equidistant from all the transmitters.

\section{IEEE~802.11 DCF Systems with Large Propagation Delays}
\label{sec:dcf-description}
\begin{figure*}[t]
%\footnotesize
\begin{center}
\includegraphics[scale=0.5]{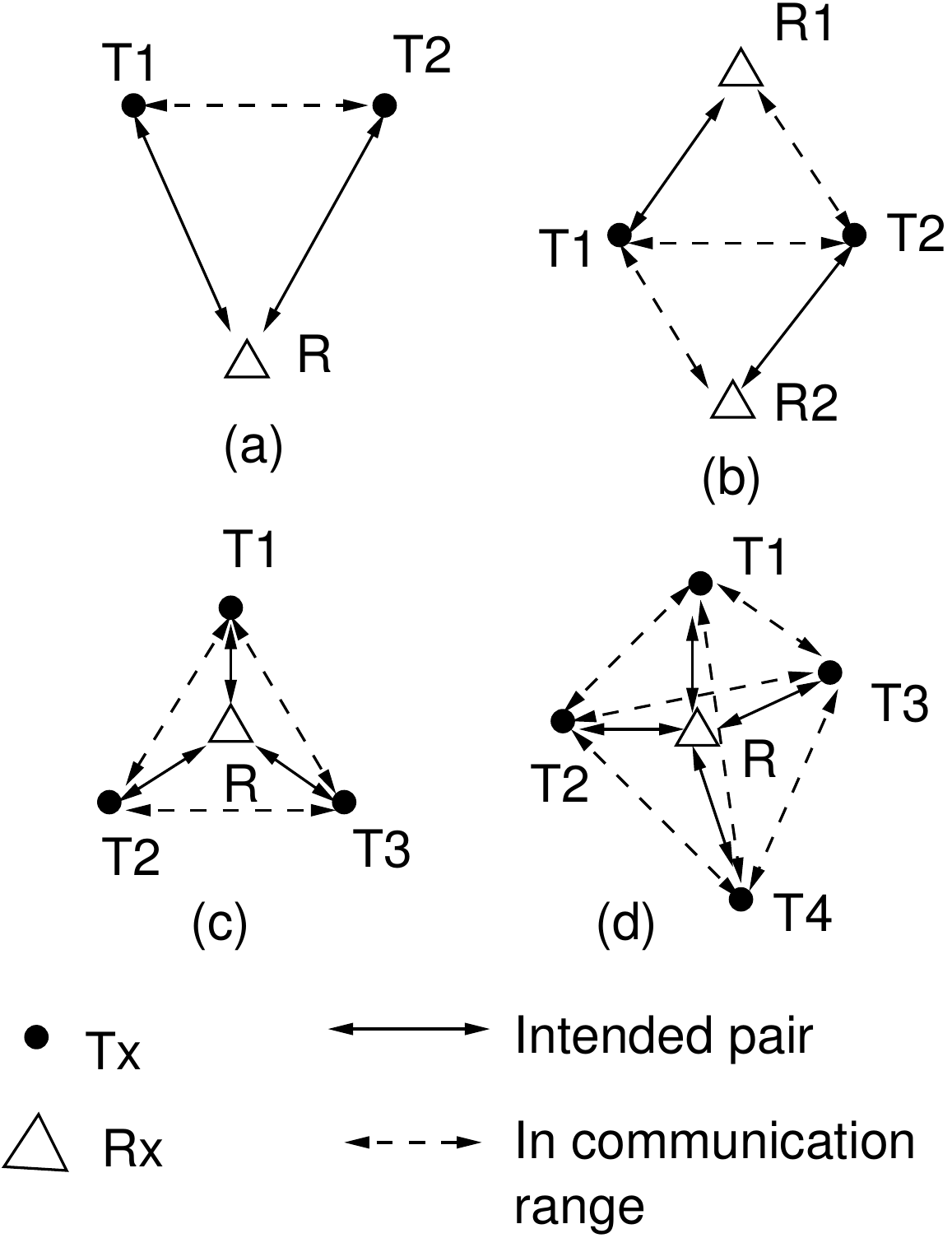}

\caption{Example systems with possibly large propagation delays where all transmitters are equidistant from one another, and each receiver is equidistant from all the transmitters.}
\label{fig:sample-scenarios}
%\vspace{-5mm}
\end{center}
\normalsize
%\vspace{-6mm}
\end{figure*}
We assume basic access without RTS-CTS. Our system consists of $n\geq 2$ saturated transmitting nodes, and their receivers, operating under IEEE~802.11 DCF. Let the propagation delay between each pair of transmitters be $\Delta$, that between each receiver and all the transmitters be $\Delta_r$, and the duration of each backoff slot be $\sigma$. Let $m \define \lfloor\frac{\Delta}{\sigma}\rfloor$, i.e., $m$ is the propagation delay among the transmitters in integer multiples of slots. Also let $m_r \define \frac{\Delta_r}{\sigma}$. See Figure~\ref{fig:sample-scenarios} for an illustration of such systems. When the propagation delays are negligible, $m=m_r = 0$. 

A node's transmission will be heard by the other nodes after a propagation delay of $m$ slots. We consider the setting where the packet duration, $T$, is much larger compared to the propagation delay, $m$.\footnote {This assumption is satisified in most scenarios of interest. For example, if the PHY layer rate is 2 Mbps, the packet duration for a 1500 bytes packet is 6000 $\mu$secs, whereas the propagation delay over a distance of 120 Kms is only 400 $\mu$secs.} In other respects our setting is the same as \cite{bianchi00performance}, i.e., there are \emph{no hidden terminals}, and \emph{no channel errors}. Thus, if two or more nodes finish their backoffs within $m$ slots of one another, their transmissions collide, and all the packets involved are lost. Note that we do not model packet capture. 

Upon a successful transmission, the transmitting node receives an ACK from its intended receiver. Due to the round-trip propagation delay between the transmitter and its receiver, the overall transmission overhead in a successful transmission is increased by $2m_r$ compared to the case without propagation delay. Thus, the ACK Timeout parameter in the protocol has to be suitably adjusted for non-negligible propagation delays. 

\subsection{A key property of the system: misaligned sensing of channel idleness}
\label{subsec:misalignment}

\begin{figure*}[t]
%\footnotesize
\begin{center}
\includegraphics[scale=0.25]{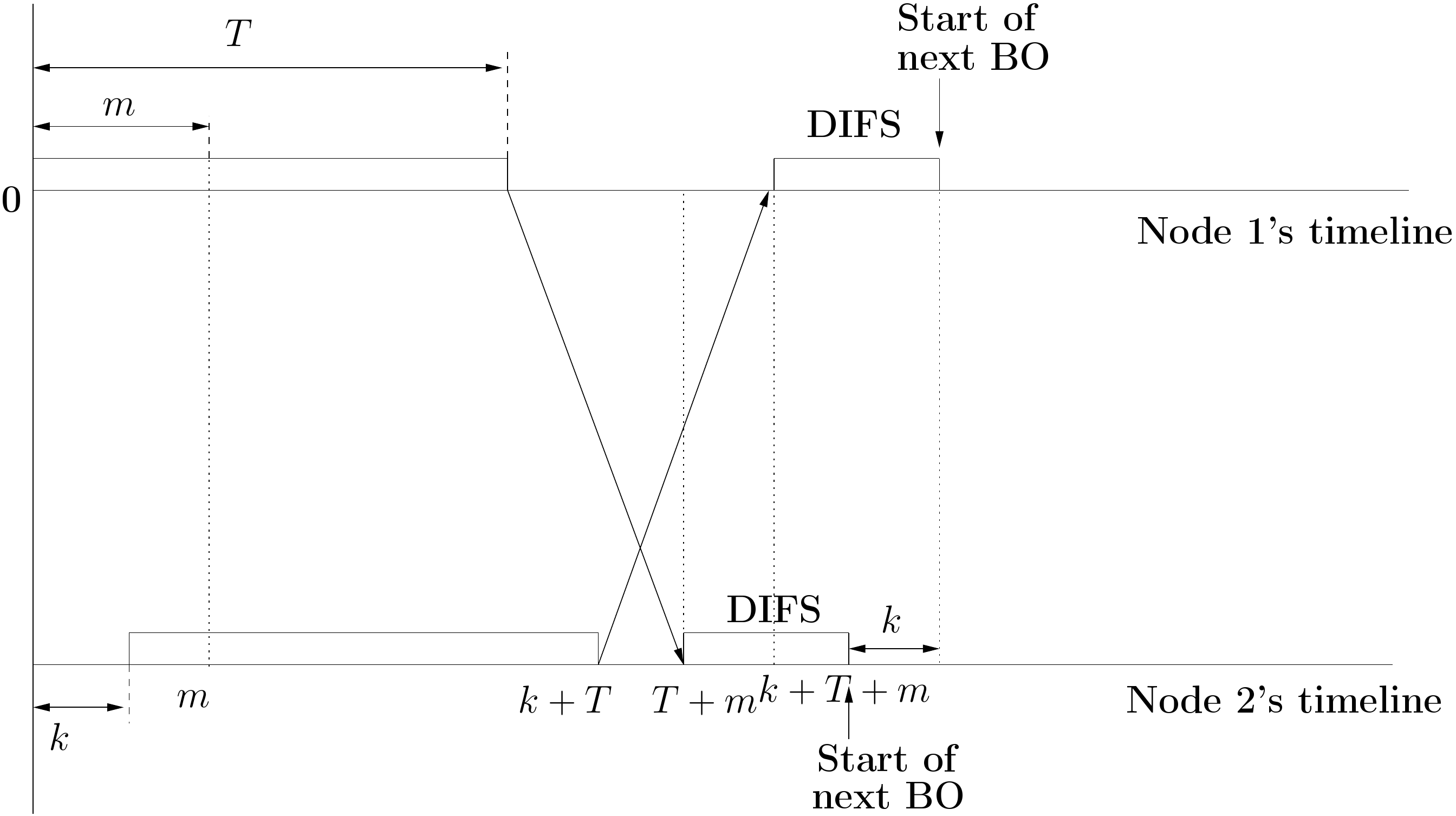}
\hspace{5mm}
\includegraphics[scale=0.25]{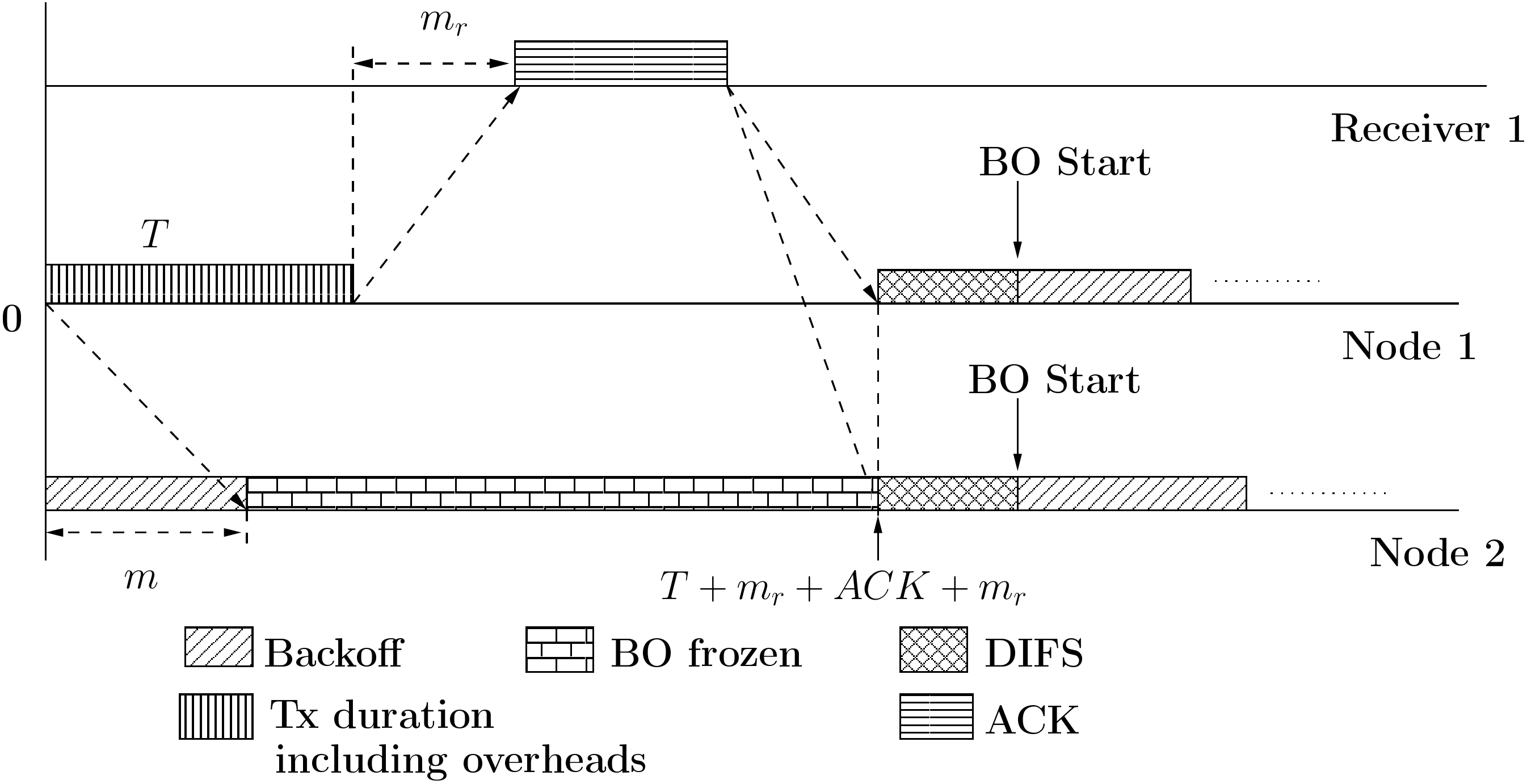}
\caption{\textbf{Left Panel: A collision, leading to misalignment}. Node~1 starts a transmission at time 0. Node~2 finishes backoff $k$ slots after Node~1, where $k<m$, and starts its transmission, only to begin to sense Node~1's transmission at time $m$, thus resulting in a collision. Node~2 will sense the channel idle at time $T+m$, and count down its DIFS, after which, it will start a fresh backoff. However, Node~1 will sense the channel idle only at time $T+k+m > T+m$. Thus, Node~1 will start counting its DIFS $k$ slots after Node~2, and hence it will also start its backoff countdown $k$ slots after Node~2. Thus, \emph{the starting points of the backoff counters are misaligned by $k$ slots}. \textbf{Right Panel: A success}. Node~1 starts a transmission at time 0. Node~2 hears this transmission after a propagation delay of $m$ slots, and freezes its backoff. Receiver~1 receives Node~1's transmission after a propagation delay of $m_r$ slots, i.e., at time $T+m_r$, and starts sending an ACK. Since the propagation delays from Receiver~1 to both the nodes are equal, both Nodes~1 and 2 hear the ACK from Receiver~1 at the same time, and hence, start their DIFS together, following which, the start their next backoffs together. Thus, no misalignment in the next backoff initiation happens in this case.}
\label{fig:misalignment-explain}
%\vspace{-5mm}
\end{center}
\normalsize
%\vspace{-6mm}
\end{figure*}
In a system with negligible (ideally, zero) propagation delay, all nodes sense the start and end of channel activity simultaneously, a DIFS period follows, and then the starts of the back-off periods at all the nodes are \emph{always aligned} (see, for example, Figure~\ref{fig:evolution_abstraction} in Section~\ref{subsec:dcf-description}). In the present case, consider the situation depicted in the left panel of Figure~\ref{fig:misalignment-explain} for a system with two transmitter-receiver pairs. As explained in Figure~\ref{fig:misalignment-explain}, when Node~2 finishes its backoff within $k < m$ slots of Node~1, they encounter a collision, and \emph{the starting points of their next backoff counters are misaligned by $k$ slots}. The misalignment, $k$, can take values in $\{0,1,\ldots, m\}$.

\remarks

\noindent
1. The possible misalignment of the backoff counters happens \emph{only when there is a collision}. In case of a success, as explained in the right panel of Figure~\ref{fig:misalignment-explain}, they start their next backoff together.

\noindent
2. Figure~\ref{fig:misalignment-explain} can be drawn for more than two nodes being involved in a collision. Consider a multiple node collision, and denote by Nodes~1 and 2 respectively, \emph{the node that attempted next to last, and the node that attempted last}. Then it is seen from the left panel of Figure~\ref{fig:misalignment-explain} that Node~2 will start its backoff earlier than the other nodes, all of whom start their backoffs together. The misalignment is precisely the difference between the attempt instants of Nodes~1 and 2. \emph{The general principle is that the node that initiates transmission earlier is the one that will have a delayed backoff in the next cycle, because it will hear the end of the other transmission later.}

\noindent
3. Most importantly, this misalignment of the backoff counters \emph{makes it difficult to apply the analytical approach in \cite{bianchi00performance,kumar-etal04new-insights} in this case}, since there the authors were able to model the process evolution by focusing \emph{only} on back-off times (see also Section~\ref{subsubsec:simo-reigadas-performance}).

\noindent
4. Such misalignment of backoff counters was also observed (even with zero propagation delay) and studied in the context of IEEE~802.11e EDCA; see \cite{tinnirello-bianchi10edca, ramaiyan-etalYYfp-analysis} and references therein. However, \emph{a crucial difference compared to our setting is that the misalignment there is deterministic for given protocol parameters, whereas in the current setting, the misalignment is random}; this prevents the use of the techniques proposed in the EDCA context to address the current problem. \hfill\Square

\section{Short Term Unfairness in Systems with Large Propagation Delays}
\label{subsec:long-distance-wifi} 
We have already seen short term unfairness in IEEE~802.11 DCF based systems where the backoff parameters were modified in some manner from the standard (Section~\ref{sec:stu-examples}). A natural question to ask is, whether the system is always well-behaved (fair) under the standard backoff parameters. It turns out that even this is not the case. In particular, in applications where the propagation delays among the nodes are not negligible compared to a backoff slot duration, especially for propagation delays more than 3 backoff slots, the system exhibits short term unfairness even under the default backoff parameters of IEEE~802.11. 

\begin{figure}[t]
%\footnotesize
\begin{center}
\includegraphics[scale=0.28]{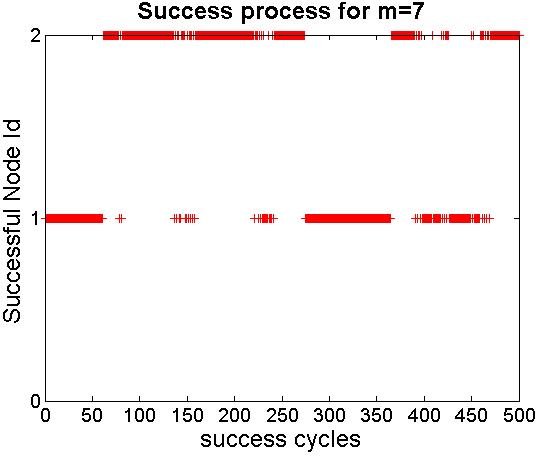}
\hspace{0.1mm}
\includegraphics[scale=0.28]{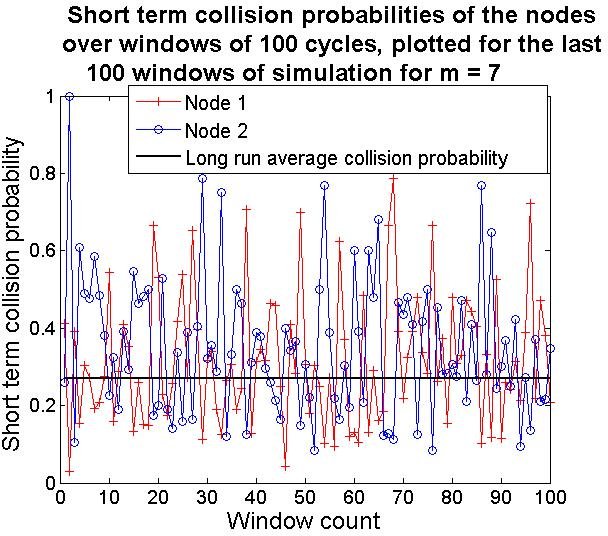}
\vspace{0.1mm}
\includegraphics[scale=0.28]{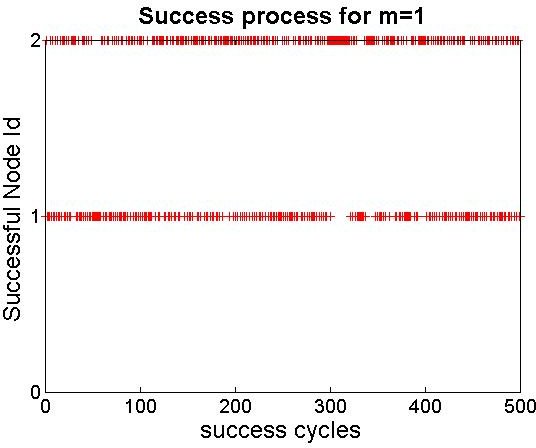}
\hspace{0.1mm}
\includegraphics[scale=0.28]{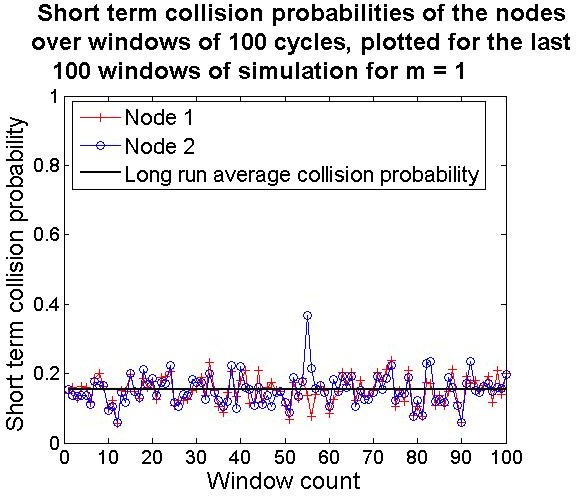}
\caption{Simulation results depicting short term unfairness at higher propagation delays for a system with 2 transmitting nodes. (Panels are row-wise, from left to right) Panels 1 and 2: Propagation delay between node pairs is $m=m_r=7$ slots. Panel 1: Evolution of the success process of the two nodes over 500 successful transmissions of the system. Panel 2: Short term collision probabilities of the two transmitters; also plotted is the long run average collision probability, averaged over nodes and simulation duration. Panels 3 and 4: Same plots as in Panels 1 and 2, but for propagation delay $m=m_r=1$ slot.}
\label{fig:short-term-unfairness-m7}
\vspace{-5mm}
\end{center}
\normalsize
%\vspace{-6mm}
\end{figure}
Panels~1 and 2 in Figure~\ref{fig:short-term-unfairness-m7} depict snapshots of a simulation run with 2 saturated transmitter-receiver pairs operating with the standard protocol parameters of IEEE~802.11b, with a propagation delay of $m=m_r=7$ slots (recall the notation from Section~\ref{sec:dcf-description}). The snapshots were obtained in the same manner as in Section~\ref{sec:stu-examples}. In Panel~1 of Figure~\ref{fig:short-term-unfairness-m7}, we depict the last 500 successful transmissions in the system, and the Node ID of the successful node in each of these transmissions. It is clearly seen from the plot that the success processes for the two nodes are bursty in nature.

To ascertain that this is not a sporadic phenomenon, but typical behavior of the system, we show in Panel~2 of Figure~\ref{fig:short-term-unfairness-m7} the short term collision probabilities of the two nodes; each point in the plot is the short term collision probability of a node computed over a window of 100 consecutive system transmissions, and the process was repeated for the last 100 windows in the simulation, thus giving 100 values for each node. Also plotted is the long run average collision probability, averaged over all the nodes, and the simulation duration. It can be observed from the plots that there is high variance in the short term collision probabilities of the two nodes w.r.t the long run average collision probability. In particular, it is often the case that in a window where Node~1 has a low short term collision probability, Node~2 has a very high short term collision probability, and vice-versa, thus indicating that one of the nodes monopolizes the channel in each window, shutting out the other node, thus leading to a high collision probability for the other node during that period. 

In order to demonstrate that this property is observed only at higher propagation delays, we show in Panels~3 and 4 in Figure~\ref{fig:short-term-unfairness-m7}, snapshots of a simulation run for the same system as before, but with a propagation delay of $m=m_r=1$ slot. It is observed from Panel~3 of Figure~\ref{fig:short-term-unfairness-m7} that the success processes of the two nodes are no more bursty in nature; in particular, no node is starved for a prolonged duration. From Panel~4 of Figure~\ref{fig:short-term-unfairness-m7}, we see that the variance in the short term collision probabilities of the two nodes w.r.t. the long run average is much less compared to that observed for $m=m_r=7$.

\noindent
\textbf{Discussion:}

The phenomenon of short-term unfairness at higher propagation delays stems from the fact that collision probability becomes very large at higher propagation delays, and so backoff distributions become stochastically very large as well. Consider a topology with $n=2$. Let Node~1 sample a backoff of $B_1$ slots, and Node~2 sample a backoff of $B_2$ slots. Observe that even at moderately large $m$, the collision probabilities of the nodes are high (almost 30\% beyond $m=3$; see $\gamma$ plot in Figure~\ref{fig:approx-model-gamma-validation}). This suggests that the contention windows (i.e., the range from which the nodes sample their backoffs) of the nodes will be stochastically larger for those $m$, which in turn implies that $|B_2-B_1|$ will be stochastically larger. Since the standard contention windows under IEEE~802.11b are much larger compared to the values of $m$ we are interested in, this increase in the contention window will dominate over any corresponding increase in $m$. Thus, for $m\geq 3$, the likelihood of $|B_2-B_1|$ growing to a value much larger than $m$ is high. Suppose, for simplicity, at the end of a channel activity, the backoff counters of both the nodes are aligned, and the nodes sample backoffs of $B_1$ and $B_2$ slots respectively. Suppose the next activity is a successful transmission by Node~1. This implies that $B_1 + m < B_2$. Recall from Section~\ref{subsec:misalignment} that at the end of the successful transmission, both the nodes will be aligned, and Node~2's residual backoff will be $B_2 - (B_1 + m)$, which, by the observation above in this paragraph, is likely to be still large. Since Node~1 will sample its next backoff from the initial (smallest) window, this also suggests that Node~1's next backoff is likely to be still much smaller than that of Node~2 (since Node~2's residual backoff is stochastically large, as just argued.), and Node~1 is therefore likely to attempt much earlier than Node~2, and succeed. This will continue to happen until Node~2's residual backoff becomes comparable to Node~1's initial contention window. Then, further, if there is a collision, Node 2's backoff can again become very large. This is reflected in the fact that for moderately large $m$, after a successful transmission, the attempt rate of Node~1 (the successful node) is higher than that of Node~2 (see $\bd$ and $\bs$ plots in Figure~\ref{fig:approx-model-gamma-validation} in Section~\ref{sec:numerical}). At higher $m$, this difference in attempt rates is so large that it causes the successful node to succeed in a burst, thus introducing significant correlation in the success process.

\subsection{Performance of an extension of Bianchi analysis for large propagation delays}
\label{subsubsec:simo-reigadas-performance}
\cite{simo-reigadas10wild} aimed to develop an approximate analytical model for single-hop, long distance WiFi systems by extending Bianchi's model. For the case of a homogeneous system, their model reduces to the following: each node, conditioned on being in backoff, attempts independently with a probability $\beta$ in each slot, irrespective of the system state. When a node transmits, the conditional probability that its transmission encounters a collision, is $\gamma$, independent of the system state. They obtain $\beta$ in terms of $\gamma$ using the well-known polynomial ratio formula (Eqn.~\ref{eqn:G_gamma} in Section~\ref{sec:back-off_process}). To obtain the collision probability, $\gamma$, they observe (inaccurately) that the vulnerable window of a tagged node has size $2m$, since any node attempting within $m$ slots before or after the tagged node's attempt will cause a collision. They then compute the probabilities of any node attempting in that vulnerable window by assuming (inaccurately) that the node was in backoff at the start of the vulnerable window, and using the Markov chain model proposed in \cite{bianchi05modified} that describes the evolution of the node \emph{in backoff time}. Thus, they arrive at a fixed point equation in $\gamma$.

Their model \emph{does not consider the fact that after a collision, the starts of the backoff counters of the nodes could be misaligned} (see Section~\ref{subsec:misalignment}), and hence when a tagged node attempts again, its vulnerable window need not be $2m$, since the other nodes may not even have started their backoff countdowns. Moreover, by assuming a constant attempt probability $\beta$ irrespective of the system state, they \emph{ignore the short term unfairness property}, which has the effect of skewing the attempt probability in favor of a successful node (as explained earlier). Figure~\ref{fig:simo-reigadas-model-comparison} compares the collision probabilities obtained from the Simo-Reigadas et al.\ model against those obtained from simulations for $n=2$, default backoff parameters of IEEE~802.11b, and a range of propagation delays. As can be seen, the values predicted by their model do not match well with the simulation results. 

\begin{figure}[t]
%\footnotesize
\begin{center}
\includegraphics[scale=0.35]{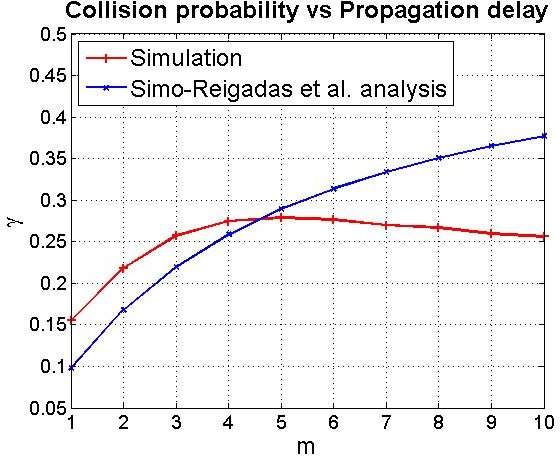}
% \hspace{0.1mm}
% \includegraphics[scale=0.3]{plots/gamma_vs_m_basic_analysis_vs_matlabsim.jpg}
%\includegraphics[scale=0.3]{plots/graph_instance6_spt_design_paths_labeled.pdf}
%\includegraphics[scale=0.3]{plots/graph_instance6_sptirp_design_paths_labeled.pdf}
\caption{Limitation of existing analysis technique in predicting the collision probabilities under large propagation delays. The model in \cite{simo-reigadas10wild} vs. simulations.}
\label{fig:simo-reigadas-model-comparison}
\vspace{-5mm}
\end{center}
\normalsize
%\vspace{-6mm}
\end{figure}

\noindent
\textbf{Discussion and the way forward:}

\noindent
Our aim is to develop an accurate analytical technique to predict the performance of IEEE~802.11 systems with large propagation delays. To that end, we adopt an approach similar to that in Part~I. We start with a detailed Markov renewal process model for the system evolution. This model is, in fact, an extension of the detailed Markov renewal model developed in Section~\ref{sec:exact_model} to the case of large propagation delays. We use this model as a prototype for the system. We then introduce a parsimonious simplification of this Markov renewal model, which, akin to the simplified model in Part~I, uses state dependent attempt rates to capture the bursty nature of the success processes due to short term unfairness (see Figure~\ref{fig:short-term-unfairness-m7}).

\section{A Markov Renewal Model of the System}
\label{sec:exact_model_del}
In this section, we present a Markov renewal process model for the system evolution under possibly large propagation delays. As will be clear from the description below, this model is essentially an extension of the detailed MRP model for systems with negligible propagation delays developed in Section~\ref{sec:exact_model}. We shall demonstrate via comparison with Qualnet simulations \cite{qualnet} (see Figure~\ref{fig:exact-model-gamma-validation}) that this model is indeed a faithful prototype for the system. 

\begin{figure}[ht]
%\footnotesize
\begin{center}
\includegraphics[scale=0.4]{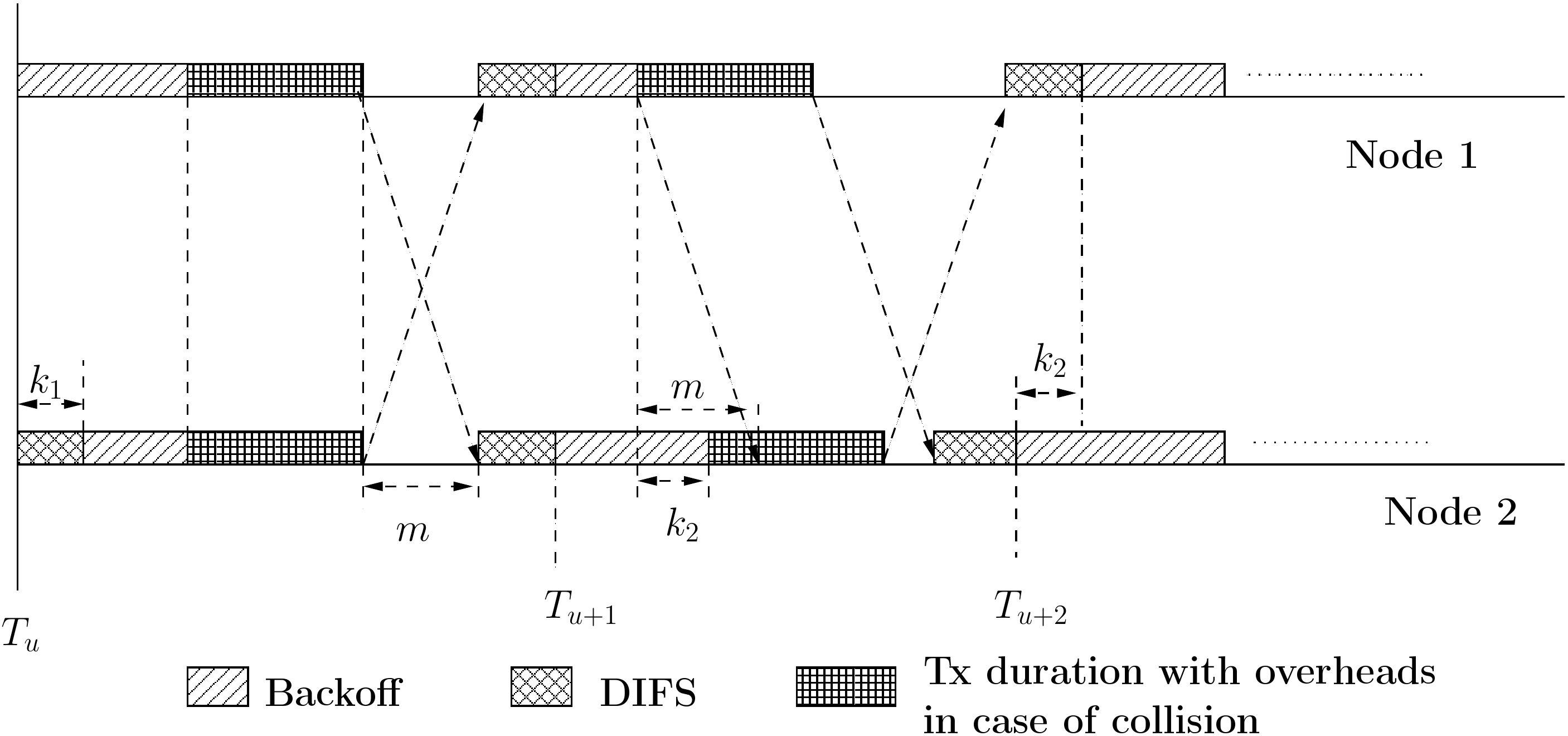}
\caption{\textbf{Transmission Cycles} for $n=2$. The evolution of the timelines can be explained as follows. Node~1 happens to be the first to start its backoff after an activity in the medium. Node~2 starts its backoff after a misalignment of $k_1$ slots. Both the nodes happen to finish their backoffs together, and start a transmission at the same time, leading to a collision. In this case, the ends of their transmissions are aligned, and hence both the nodes sense the channel idle (after a propagation delay of $m$ slots), and start their DIFS at the same time, following which they start fresh backoffs, with the starts of the backoff counters aligned. This time, Node~1 finishes its backoff first, and starts a transmission. Node~2 finishes its backoff $k_2$ slots after Node~1, where $k_2 < m$, thus leading to a collision, and subsequent misalignment of the starts of their next backoffs by $k_2$ slots, in the same manner as explained in the left panel of Figure~\ref{fig:misalignment-explain}, with Node~2 leading Node~1 by $k_2$ slots. Denote by $T_u$, the first instant after the $u^{th}$ activity in the medium when some node starts counting down its backoff. The intervals $[T_u, T_{u+1}]$ and $[T_{u+1},T_{u+2}]$ are, respectively, the $(u+1)^{th}$ and $(u+2)^{th}$ \emph{transmission cycles}. Note that at the start of a transmission cycle, some (but not all) of the nodes may still be counting down their \emph{misalignment slots} before entering backoff. For example, at $T_{u+2}$, Node~1 is counting down a misalignment of $k_2$ slots.}
\label{fig:tx-cycle-explain-del}
%\vspace{-5mm}
\end{center}
\normalsize
%\vspace{-6mm}
\end{figure}

An \emph{``activity''} in the medium is defined as the duration from the instant when a transmission starts in the medium, to the instant when some node is ready to start its next DIFS. For example, in the Left panel of Figure~\ref{fig:misalignment-explain}, there is an activity in the medium during the interval $[0,T+m]$, and in the Right panel of Figure~\ref{fig:misalignment-explain}, there is an activity in the medium during the interval $[0,T+m_r+ACK+m_r]$.
  
Let $T_u$ be the first instant after the $u^{th}$ activity in the medium when some node starts counting down its backoff. See, for example, Figure~\ref{fig:tx-cycle-explain-del}, which depicts a sample path of the system evolution for $n=2$. We call the interval $[T_u, T_{u+1}]$ the $(u+1)^{th}$ \emph{transmission cycle}. In each transmission cycle, there is excactly one activity in the medium. 

Let $B_{u,i}, S_{u,i}, Z_{u,i}$, denote respectively the residual backoff count, backoff stage, and misalignment (w.r.t $T_u$) of the start of backoff counter of Node~$i$, $i=1,2,\ldots,n$ at $T_u$. Recalling the notation for the protocol parameters of IEEE~802.11 DCF, $S_{u,i}\in\{0,1,\ldots,K\}$, $B_{u,i}\in\{1,\ldots,W_{S_{u,i}}\}$, and $Z_{u,i}\in\{0,1,\ldots,m\}$. Then, the process $(\{B_{u,i}, S_{u,i}, Z_{u,i}\}_{i=1}^n, T_u)$ is a Markov Renewal Process \cite{kulkarni95modeling-stochastic-systems}, with $\{B_{u,i}, S_{u,i}, Z_{u,i}\}_{i=1}^n$ being the embedded Markov chain, whose transition structure is explained next. 

Note that $(T_u+B_{u,i} + Z_{u,i})$ is the instant when Node~$i$ is scheduled to finish its backoff, and attempt a transmission in the $(u+1)^{th}$ transmission cycle. Let $\underline{B}_u = \min_{1\leq i\leq n} (B_{u,i} + Z_{u,i})$, and $I_u = \arg \min_{1\leq i\leq n} (B_{u,i} + Z_{u,i})$. 

\noindent
\textbf{Observations:}

\begin{enumerate}
 \item $(T_u+B_u)$ and $I_u$ are, respectively, the attempt instant, and Node id of the first node to attempt transmission in the $(u+1)^{th}$ transmission cycle. 
 \item A successful transmission happens \emph{iff} for all $i\neq I_u$, $B_{u,i} + Z_{u,i} > \underline{B}_u + m$, and a collision happens otherwise. We need to consider only the integer part of the propagation delay between the transmitters in slots, i.e., $m$, since the probabilities of the events corresponding to success and collision are unaffected by the fractional part of the propagation delay; to see this, note that $B_{u,i}$ and $Z_{u,i}$ always take values in integer multiples of slots.
\end{enumerate}

With the above information, the transition structure of the embedded Markov chain can be summarized as follows:
\begin{enumerate}
\item Initialize the set of nodes attempting in the $(u+1)^{th}$ transmission cycle as $S_{a,u} = \phi$. For each node~$i$, $1\leq i\leq n$, if $B_{u,i} + Z_{u,i} > \underline{B}_u + m$, i.e., the node hears the ongoing transmission before finishing its backoff, then Node~$i$ is frozen in the $(u+1)^{th}$ transmission cycle, and its backoff states are updated as $B_{u+1,i} = B_{u,i} + Z_{u,i} - (\underline{B}_u + m)$, and $S_{u+1,i}=S_{u,i}$.

If, on the other hand, $B_{u,i} + Z_{u,i} \leq \underline{B}_u + m$, then Node~$i$ attempts in the $(u+1)^{th}$ transmission cycle, and the set of attempting nodes is updated as $S_{a,u} = S_{a,u}\cup \{i\}$.

\item If $|S_{a,u}|=1$, i.e., exactly one node, namely, Node~$I_u$ attempted in the $(u+1)^{th}$ transmission cycle, then the transmission is successful, and Node~$I_u's$ backoff stage becomes $S_{u+1,I_u}=0$; $B_{u+1,I_u}$ is sampled from a uniform distribution from $\{1, CW_{\min}\}$. In this case, $Z_{u+1,i} = 0$ for all $i=1,\ldots,n$ (recall Remark~1 in Section~\ref{subsec:misalignment}). The duration of the transmission cycle in this case is $B_u + \text{transmission duration with overheads} + ACK + 2m_r$. See the right panel of Figure~\ref{fig:misalignment-explain}.

\item If $|S_{a,u}|> 1$, then more than one node attempted in the $(u+1)^{th}$ transmission cycle, resulting in a collision for all the nodes in $S_{a,u}$. For each node~$j\in S_{a,u}$, its backoff stage will be updated as $S_{u+1,j} = (S_{u,j}+1) mod (K+1)$, where $K$ is the maximum allowed number of retransmissions. $B_{u+1,j}$ is sampled uniformly from the contention window corresponding to $S_{u+1,j}$. The duration of the transmission cycle in this case is $B_u + \text{transmission duration with overheads} + \frac{\Delta}{\sigma}$. See, for example, the transmission cycle $[T_{u+1},T_{u+2}]$ in Figure~\ref{fig:tx-cycle-explain-del}.  

To compute $Z_{u+1,i}, 1\leq i\leq n$, suppose $I_{u,1}$ and $I_{u,2}$ be the indices of the two nodes that attempted last, i.e., $I_{u,2} = \arg\max_{j\in S_{a,u}}(B_{u,j}+Z_{u,j})$, and $I_{u,1}=\arg\max_{j\in S_{a,u}\backslash I_{u,2}}(B_{u,j}+Z_{u,j})$. Then, by Remark~2 in Section~\ref{subsec:misalignment}, it follows that $Z_{u+1,I_{u,2}}=0$, and for all $i\neq I_{u,2}$, $Z_{u+1,i} = B_{u,I_{u,2}}+Z_{u,I_{u,2}} - (B_{u,I_{u,1}}+Z_{u,I_{u,1}})$. Since for all $i\neq I{u,2}$, $Z_{u+1,i}$'s are equal, we denote this common value as $Z_{u+1,+}$.

Note that $Z_{u+1,+}=B_{u,I_{u,2}}+Z_{u,I_{u,2}} - (B_{u,I_{u,1}}+Z_{u,I_{u,1}}) \leq m$, since otherwise Node~$I_{u,2}$ would have heard the transmission from Node~$I_{u,1}$, and refrained from attempting. Also, by our definition of $I_{u,2}$, $B_{u,I_{u,2}}+Z_{u,I_{u,2}} - (B_{u,I_{u,1}}+Z_{u,I_{u,1}}) \geq 0$. It follows that for $m=0$, $Z_{u+1,+} = 0$.
\end{enumerate}

We have simulated this detailed model for the case of $n=2$, default backoff parameters of IEEE~802.11b, and a wide range of propagation delays (with $m=m_r$) to obtain the long run average collision probability, $\gamma$, and compared these analysis results against simulation results obtained from Qualnet\footnote{after correcting an error in the default Qualnet implementation wherein an extra delay of $m_r$ gets added to the NAV of the frozen node in addition to the correct value of $2m_r$.}. 

To see how the long run average collision probability can be obtained from the model-based simulator, note that the model-based simulator runs in steps of transmission cycles. To obtain the average collision probability of Node~$i$, we count the number of transmission cycles in which Node~$i$ made an attempt, denoted by $A_i$, and the number of transmission cycles in which Node~$i$'s attempt encountered a collision, denoted by $C_i$. Then, the average collision probability of Node~$i$ is estimated as $\gamma_i=\frac{C_i}{A_i}$. Note that if the simulation is run long enough, then since the nodes are symmetric, $\gamma_i\approx \gamma_j$, for all $1\leq i,j\leq n$. This was observed in all our simulations. Finally, we estimate the long run average collision probability, $\gamma$, as $\gamma=\frac{1}{n}\sum_{i=1}^n \gamma_i$. 

The results are shown in Figure~\ref{fig:exact-model-gamma-validation}; it can be seen that the proposed model captures the system behavior very accurately. Also, there is a distinct advantage of using a Monte Carlo simulation of this detailed model over using Qualnet (or any other event-driven) simulation for predicting the system performance. Qualnet simulation runs over backoff slots, and works by simulating all the details of the protocol at every node; on the other hand, the model-based simulator runs over transmission cycles, and eliminates all unnecessary details of the protocol. Hence, the model-based simulator can run much faster while achieving comparable accuracy. 

\begin{figure}[t]
%\footnotesize
\begin{center}
\includegraphics[scale=0.5]{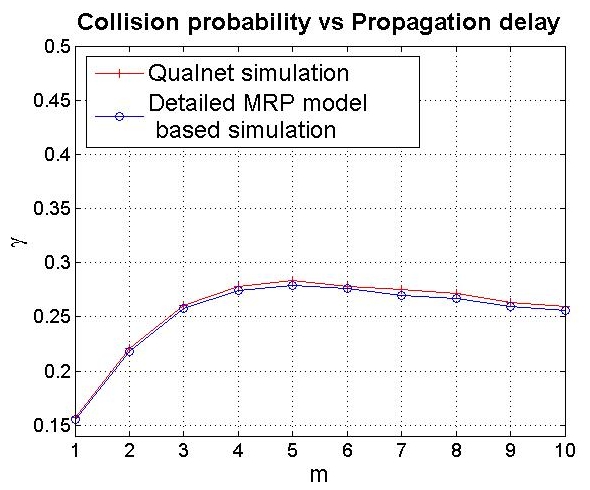}

% \hspace{0.1mm}
% \includegraphics[scale=0.24]{plots/gamma_vs_m_basic_analysis_vs_matlabsim.jpg}
%\includegraphics[scale=0.3]{plots/graph_instance6_spt_design_paths_labeled.pdf}
%\includegraphics[scale=0.3]{plots/graph_instance6_sptirp_design_paths_labeled.pdf}
\caption{Collision probability ($\gamma$) vs. propagation delay ($m$). Comparison of collision probabilities obtained via a Monte-Carlo simulation of the detailed MRP model against those obtained from Qualnet simulations \cite{qualnet}.}
\label{fig:exact-model-gamma-validation}
%\vspace{-5mm}
\end{center}
\normalsize
%\vspace{-6mm}
\end{figure}
However, the proposed model involves an embedded $3n$-dimensional Markov chain, whose state space has size $(nm+1)(W_0+W_1+\cdots+W_K)^n$, where $K$ is the retransmission limit for the protocol, and $W_j$ is the contention window size for backoff stage $j$. For the default protocol parameters of IEEE~802.11b, the size of the state space is prohibitively large even for $m=1$, and $n=2$, making an exact analysis of the embedded Markov chain computationally intractable. We, therefore, focus on developing an approximate, parsimonious analysis, as was done in Section~\ref{sec:mrp-state-dependent}. 

\section{A Parsimonious Simplification of the Markov Renewal Model in Section~\ref{sec:exact_model_del}}
\label{sec:mrp-state-dependent-del} 

\subsection{An approximate Markov renewal model for system evolution}
\label{subsec:system-evolution-mrp-del}

As in Section~\ref{sec:mrp-state-dependent}, we retain the embedded Markov process structure at the starts of transmission cycles, $T_u$, but simplify the evolution of the process between these embedding points by introducing a Bernoulli attempt process approximation for the nodes with \emph{state dependent attempt rates}, namely, $\bs,\bc$, and $\bd$, where $\bs,\bc,$ and $\bd$ have the same interpretation as before. The motivation for the state dependent attempt rates comes from the observation of short term unfairness in Figure~\ref{fig:short-term-unfairness-m7}, where the bursty success processes (Panel~1) indicate that the attempt rates are skewed in favor of the last successful node. 

Furthermore, following the same arguments as in Section~\ref{sec:mrp-state-dependent}, we associate with each epoch $T_u$, a state, $N_u$, \emph{the number of nodes that attempted in the previous cycle}.

\noindent
\textbf{Accounting for \emph{possible misalignment} in case of large propagation delay:}
We saw in Section~\ref{subsubsec:simo-reigadas-performance} that if we do not account for the \emph{possible misalignment} of backoff counters of the nodes after a collision (Section~\ref{subsec:misalignment}), the resulting analysis is not accurate. To account for this, we associate with each $T_u$, another state, namely, the \emph{misalignment}, $Z_u$, of the backoff counters of the nodes at $T_u$. Note that $Z_u=0$ if there was a success in the last transmission cycle, and $Z_u=Z_{u,+}\in\{0,1,\ldots,m\}$ otherwise. For example, in Figure~\ref{fig:tx-cycle-explain-del}, the \emph{misalignments} at $T_u$, $T_{u+1}$ and $T_{u+2}$ are respectively $k_1$, $0$, and $k_2$ slots. 

Further note that to use the state dependent attempt rates, we need to know whether a transmission cycle ended in a success, or a collision. Observe that while $Z_u>0$ clearly indicates a collision in the previous transmission cycle, $Z_u=0$ could indicate either a collision or a success in the previous transmission cycle. To distinguish between these two cases, we introduce two new states, namely $0_s$, and $0_c$, indicating that there is no misalignment of the backoff counters at $T_u$, and that the previous transmission cycle ended in a success, or a collision respectively. Thus, in our new model, $Z_u\in\{0_s,0_c,1,\ldots, m\}$. Finally, note that $N_u=1$ if $Z_u=0_s$, and $N_u\geq 2$ otherwise. 

Our approximations can be summarized as follows:

\noindent
\textbf{(A1)} If $Z_u = 0_s$, all the nodes start their backoffs from $T_u$. The node that was successful in the previous transmission cycle attempts independently with probability $\beta_s$ in each slot, conditioned on being in backoff. The other nodes attempt independently with probability $\beta_d$ in each slot, conditioned on being in backoff.\hfill \Square

\noindent 
\textbf{(A2)} If $Z_u = 0_c$, all nodes start their backoffs from $T_u$. $N_u$ of the nodes attempt independently with probability $\beta_c$ in each slot, while the remaining $n-N_u$ nodes attempt independently with probability $\beta_d$ in each slot, all conditioned on being in backoff. If $Z_u = k>0$, $N_u$ of the nodes attempt independently with probability $\beta_c$ in each slot, conditioned on being in backoff, one starting from $T_u$, and the others, starting from $T_u+k$ (Remark~2, Section~\ref{subsec:misalignment}); the remaining $n-N_u$ nodes attempt independently with probability $\beta_d$ in each slot, conditioned on being in backoff, starting from $T_u+k$.\hfill \Square 

\noindent
\textbf{A simple Markov renewal process model for the system:}
With these approximations, observe that the process $\{(Z_u,N_u), T_u\}$, is a Markov renewal process (MRP), the state space of the embedded Markov chain being $\{0_s,0_c,1,\ldots,m\}\times\{1,\ldots,n\}$. Also, observe that for $n=2$ and arbitrary $m$, it suffices to consider only the state $Z_u$, thus reducing the state space. We develop the details for this case. The underlying principles apply to the more general setting as well, but the equations become more involved. 

\subsection{Analysis of the MRP, given $\beta_c, \beta_d$, and $\beta_s$}
\label{subsec:mrp-analysis-given-beta-del}

As just mentioned, for $n=2$ and arbitrary $m$, $\{Z_u, T_u\}$ is a Markov renewal process (MRP), the state space of the embedded Markov chain being $\{0_s,0_c,1,\ldots,m\}$. This Markov renewal process model has $m$ as a parameter, and requires the quantities $\beta_c, \beta_d$, and $\beta_s$, which are not known a priori. We shall explain how to compute $\beta_c, \beta_d$, and $\beta_s$ in Section~\ref{subsec:tagged-node-evolution-del}. Given $\beta_c, \beta_d$, and $\beta_s$, let $P$ be the transition probability matrix of the embedded Markov chain. We now proceed to write down the transition probabilities. We use the shorthand $p(i,j)$ to denote the probability $Pr[Z_{u+1}=j|Z_u = i]$. 

\noindent
\emph{Computation of transition probabilities from $0_s$}:

If $Z_u=0_s$, three possible events can lead to the state $Z_{u+1}=0_s$.
\begin{enumerate}
 \item The node that was successful in the previous cycle attempts in the first slot, and the other node does not attempt in the first slot, and the next $m$ slots, thus ensuring that the former is successful again in the current cycle. This happens with probability $\beta_s(1-\beta_d)^{(m+1)}$.
\item The node that was frozen in the previous cycle attempts in the first slot, and the other node does not attempt in the first slot, and the next $m$ slots, thus ensuring that the former is successful in the current cycle. This happens with probability $\beta_d(1-\beta_s)^{(m+1)}$. 
\item None of the nodes attempt in the first slot; this happens with probability $(1-\beta_d)(1-\beta_s)$. In this case, due to the assumption of Bernoulli attempt processes, the system encounters a renewal with state $0_s$, and the conditional probability (given that none of the nodes attempted in the first slot) of the next state being $0_s$ remains $p(0_s,0_s)$. 
\end{enumerate}
Putting all of these together, we have
\begin{equation}
 p(0_s,0_s) = \frac{\beta_s(1-\beta_d)^{(m+1)}+\beta_d(1-\beta_s)^{(m+1)}}{1-(1-\beta_d)(1-\beta_s)}
\end{equation}

Using similar arguments, we have
\begin{eqnarray}
 p(0_s,0_c) = \frac{\beta_s\beta_d}{1-(1-\beta_d)(1-\beta_s)}\\
 p(0_s,k) = \frac{\beta_s(1-\beta_d)^k \beta_d + \beta_d(1-\beta_s)^k \beta_s}{1-(1-\beta_d)(1-\beta_s)}\:\forall k=1,\ldots,m
\end{eqnarray}

\noindent
\emph{Computation of transition probabilities from $0_c$}:

Note that when $Z_u=0_c$, Approximation~(A2) is in force. Then, we can use similar renewal argument as before to conclude that
\begin{eqnarray}
 p(0_c,0_s) = \frac{2\beta_c(1-\beta_c)^{(m+1)}}{1-(1-\beta_c)^2}\\
 p(0_c,0_c) = \frac{\beta_c^2}{1-(1-\beta_c)^2}\\
 p(0_c,k) = \frac{2\beta_c(1-\beta_c)^k \beta_c}{1-(1-\beta_c)^2}
\end{eqnarray}

\noindent
\emph{Computation of transition probabilities from state $k\in\{1,\ldots,m\}$}:

Note that when $Z_u = k > 0$, Approximation~(A2) is in force, and exactly one node (let us denote it as Node~1) begins its backoff process from slot 1, while the other node (denote it by Node~2) begins its backoff process from slot $k+1$. Conditioned on being in backoff, each node attempts independently with probability $\beta_c$ in each slot. There are two sets of events that can lead to the state $Z_{u+1}=0_s$:
\begin{enumerate}
 \item Node~1 does not make an attempt in the first $k$ slots; this happens with probability $(1-\beta_c)^k$. In this case, due to the memoryless property of the Bernoulli attempt processes of the nodes, the system undergoes a renewal at the end of slot $k$ with state $0_c$, and the conditional probability (given that Node~1 did not attempt in the first $k$ slots) that $Z_{u+1}=0_s$ is $p(0_c,0_s)$.
\item Node~1 attempts at slot $j$, $1\leq j\leq k$; this happens with probability $(1-\beta_c)^{(j-1)}\beta_c$. Then Node~1 will be successful (thus leading to $Z_{u+1}=0_s$) if and only if Node~2 does not attempt anywhere between slots $(k+1)$ and $(j+m)$, both inclusive; the probability of this event is $1-p^{(k)}_j$, where we define $p^{(k)}_j\define 1-(1-\beta_c)^{(j+m-k)}$, as the probability that Node~2 attempts somewhere between slots $(k+1)$ and $(j+m)$.  
\end{enumerate}
Putting these together, we have
\begin{equation}
 p(k,0_s) = (1-\beta_c)^k p(0_c,0_s) + \sum_{j=1}^k (1-\beta_c)^{(j-1)}\beta_c (1-p^{(k)}_j)
\end{equation}

Observing that if Node~1 attempts within the first $k$ slots, then $Z_{u+1}$, i.e., the next state, cannot be $0_c$, and using a renewal argument as above, we have
\begin{equation}
 p(k,0_c) = (1-\beta_c)^k p(0_c,0_c)
\end{equation}

Finally, there are two sets of events that can lead to the state $Z_{u+1}=k^{\prime}\in\{1,\ldots,m\}$:
\begin{enumerate}
 \item Node~1 does not make an attempt in the first $k$ slots; this happens with probability $(1-\beta_c)^k$. In this case, due to the memoryless property of the Bernoulli attempt processes of the nodes, the system undergoes a renewal at the end of slot $k$ with state $0_c$, and the conditional probability (given that Node~1 did not attempt in the first $k$ slots) that $Z_{u+1}=k^{\prime}$ is $p(0_c,k^{\prime})$.
\item Node~1 attempts at slot $j$, $1\leq j\leq k$; this happens with probability $(1-\beta_c)^{(j-1)}\beta_c$. Then the next state can be $k^{\prime}$ if and only if Node~2 attempts at slot $j+k^{\prime}$. Recalling that Node~2 begins its backoff process from slot $(k+1)$, this happens with probability $(1-\beta_c)^{(j+k^{\prime}-(k+1))}\beta_c$, provided $j\geq k-k^{\prime}+1$. 
\end{enumerate}
Combining these, we have
\begin{align}
 p(k,k^{\prime}) &= (1-\beta_c)^k p(0_c,k^{\prime})+ \sum_{j=\max\{1,k-k^{\prime}+1\}}^k (1-\beta_c)^{(j-1)}\beta_c (1-\beta_c)^{(j+k^{\prime}-(k+1))}\beta_c
\end{align}

From the above transition probability structure, it is easy to observe that for positive attempt rates, the embedded DTMC is finite, irreducible, and hence, \emph{positive recurrent}. Let $\pi$ denote the stationary distribution of this DTMC, which can be obtained as the unique solution to the system of equations $\pi=\pi P$, subject to $\displaystyle{\sum_{k\in\{0_s,0_c,1,\ldots,m\}}}\pi(k)=1$. 

\subsubsection{Obtaining the collision probability, $\gamma$, for $n=2$, and arbitrary $m$}
\label{subsubsec:gamma}
By symmetry, the long run average collision probability for both the nodes is the same, which we denote by $\gamma$. It is defined as 
\begin{equation*}
 \gamma = \lim_{t\to\infty}\frac{C_i(t)}{A_i(t)},\:i=1,2
\end{equation*}
where, $C_i(t)$ and $A_i(t)$ denote respectively, the number of collisions and the number of attempts by Node~$i$ until time $t$. Denoting $C(t)\define\sum_{i=1}^2 C_i(t)$, the total number of collisions in the system until time $t$, and $A(t)\define\sum_{i=1}^2 A_i(t)$, the total number of attempts in the system until time $t$, it is also easy to observe (by noting that the long run time-average collision rates, and the long run time-average attempt rates of both the nodes are equal by symmetry) that 
\begin{equation*}
 \gamma = \lim_{t\to\infty}\frac{C(t)}{A(t)}
\end{equation*}

Denote by $\Ccal$ and $\Acal$, respectively, the random variables representing the number of collisions, and the number of attempts in the system in a transmission cycle. Then, using Markov regenerative theory, we have
\begin{equation}
\gamma = \frac{\sum_{k\in\{0_s,0_c,\ldots,m\}}\pi(k)E\Ccal(k)}{\sum_{k\in\{0_s,0_c,\ldots,m\}}\pi(k)E\Acal(k)}\: a.s\label{eqn:gamma-expression} 
\end{equation}
where, $E\Ccal(k)$ and $E\Acal(k)$ denote respectively, the expected number of collisions, and attempts in the system in a transmission cycle starting with state $k$, and can be computed by using renewal arguments similar to those used for obtaining the transition probabilities earlier, and observing that every collision event results in 2 collisions (and involves 2 attempts, one from each node), and every success event involves 1 attempt (from the successful node). We write down the expressions for $E\Ccal(\cdot)$ and $E\Acal(\cdot)$ below:

\begin{align}
 E\Ccal(0_s) &= \frac{\bs(1-\bd)q_d\cdot 2 + \bd(1-\bs)q_s\cdot 2 + 2\bs\bd}{1-(1-\bs)(1-\bd)}\\
E\Acal(0_s) &= \frac{\bs(1-\bd)(1+q_d) + \bd(1-\bs)(1+q_s) + 2\bs\bd}{1-(1-\bs)(1-\bd)}\\
E\Ccal(0_c) &= \frac{2\bc(1-\bc)q_c\cdot 2 + 2\bc^2}{1-(1-\bc)^2}\\
E\Acal(0_c) &= \frac{2\bc(1-\bc)(1+q_c) + 2\bc^2}{1-(1-\bc)^2}\\
E\Ccal(k) &= (1-\bc)^k E\Ccal(0_c) + \sum_{j=1}^k (1-\bc)^{(j-1)}\bc p_j^{(k)}\cdot 2\:\forall k=1,\ldots,m\\
E\Acal(k) &= (1-\bc)^k E\Acal(0_c)\nonumber\\
& + \sum_{j=1}^k (1-\bc)^{(j-1)}\bc (1+p_j^{(k)})\:\forall k=1,\ldots,m
\end{align}
where, we define $q_d\define 1-(1-\bd)^m$, $q_s\define 1-(1-\bs)^m$, and $q_c\define 1-(1-\bc)^m$. This completes the computation of the average collision probability, $\gamma$, given the conditional attempt rates $\bd$, $\bs$, $\bc$. 

\subsubsection{Obtaining the normalized system throughput, $\Theta$, for $n=2$, and arbitrary $m$}
\label{subsubsec:throughput}
The normalized system throughput is defined as
\begin{equation*}
 \Theta = \lim_{t\to\infty}\frac{T(t)}{t}
\end{equation*}
where $T(t)$ is the total successful data transmission duration without overheads until time $t$. 

Denote by $\Tcal$, the random variable representing the duration of successful data transmission excluding overheads in a transmission cycle. Then, by Markov regenerative theory, we have
\begin{equation}
 \Theta = \frac{\sum_{k\in\{0_s,0_c,\ldots,m\}}\pi(k)E\Tcal(k)}{\sum_{k\in\{0_s,0_c,\ldots,m\}}\pi(k)EX(k)}\: a.s\label{eqn:theta-expression}
\end{equation}
where, $E\Tcal(k)$ and $EX(k)$ are, respectively, the mean duration of successful data transmission excluding overheads, and the mean duration of the transmission cycle when the transmission cycle starts in state $k$. Letting $T_{d}$, $T_o$, $\Delta$, and $\sigma$ denote respectively the data packet duration, Rx-to-tx turnaround time, propagation delay, and slot duration, we can write down the expressions for $E\Tcal(\cdot)$ and $EX(\cdot)$ using renewal arguments similar to those given earlier as follows.

\begin{align}
 EX(0_s) &= \frac{1}{1-(1-\bs)(1-\bd)}[1+(\bs\bd + \bs(1-\bd)q_d\nonumber\\
& + \bd(1-\bs)q_s)T_c + (\bs(1-\bd)(1-q_d)\nonumber\\
&+\bd(1-\bs)(1-q_s))T_s]\\
E\Tcal(0_s) &= \frac{(\bs(1-\bd)(1-q_d)+\bd(1-\bs)(1-q_s))T_d}{1-(1-\bs)(1-\bd)}\\
EX(0_c) &= \frac{1}{1-(1-\bc)^2}[1+(\bc^2+2\bc(1-\bc)q_c)T_c\nonumber\\
&+2\bc(1-\bc)(1-q_c)T_s]\\
E\Tcal(0_c) &= \frac{2\bc(1-\bc)(1-q_c)T_d}{1-(1-\bc)^2}\\
EX(k) &= (1-\bc)^k(k+EX(0_c)) + \sum_{j=1}^k(1-\bc)^{j-1}\bc[j\nonumber\\
&+(1-p^{(k)}_j)T_s + p^{(k)}_j T_c]\:\forall k\in\{1,\ldots,m\}\\
E\Tcal(k) &= (1-\bc)^k E\Tcal(0_c) + \sum_{j=1}^k(1-\bc)^{j-1}\bc(1-p^{(k)}_j)T_d
\end{align}
where, $T_s$ is the time duration in a successful transmission cycle from the start of the data transmission in the medium until the time the medium is idle again, and some node starts counting its backoff (i.e., until the start of the next transmission cycle), and is given by
\begin{align*}
 T_s &= T_d + ACK + 2\times PHY\:\: HDR + 2T_o + SIFS + DIFS + 2\Delta_r
\end{align*}
and $T_c$ is the time duration in a collision transmission cycle from the start of the first data transmission in the medium until the time some node starts counting its backoff (i.e., until the start of the next transmission cycle), and is given by
\begin{align*}
 T_c &= T_d + PHY\:\: HDR + T_o + SIFS + DIFS + \Delta
\end{align*}

This completes the analysis of the system evolution, given $\bs$, $\bd$, $\bc$. 
 
It remains to obtain the state dependent attempt rates $\bs$, $\bd$, $\bc$. To do this, we focus on the evolution at a tagged node as described next.

\subsection{Analysis for determining $\beta_c, \beta_d$, and $\beta_s$}
\label{subsec:tagged-node-evolution-del}
Here we shall set up a system of fixed point equations in $\bc$, $\bd$, and $\bs$ by modeling the evolution at a tagged node; the method is similar to what was done in Section~\ref{subsec:tagged-node-evolution}. This can, in turn, be solved iteratively to yield the rates. 
\begin{figure}[ht]
%\footnotesize
\begin{center}
\includegraphics[scale=0.4]{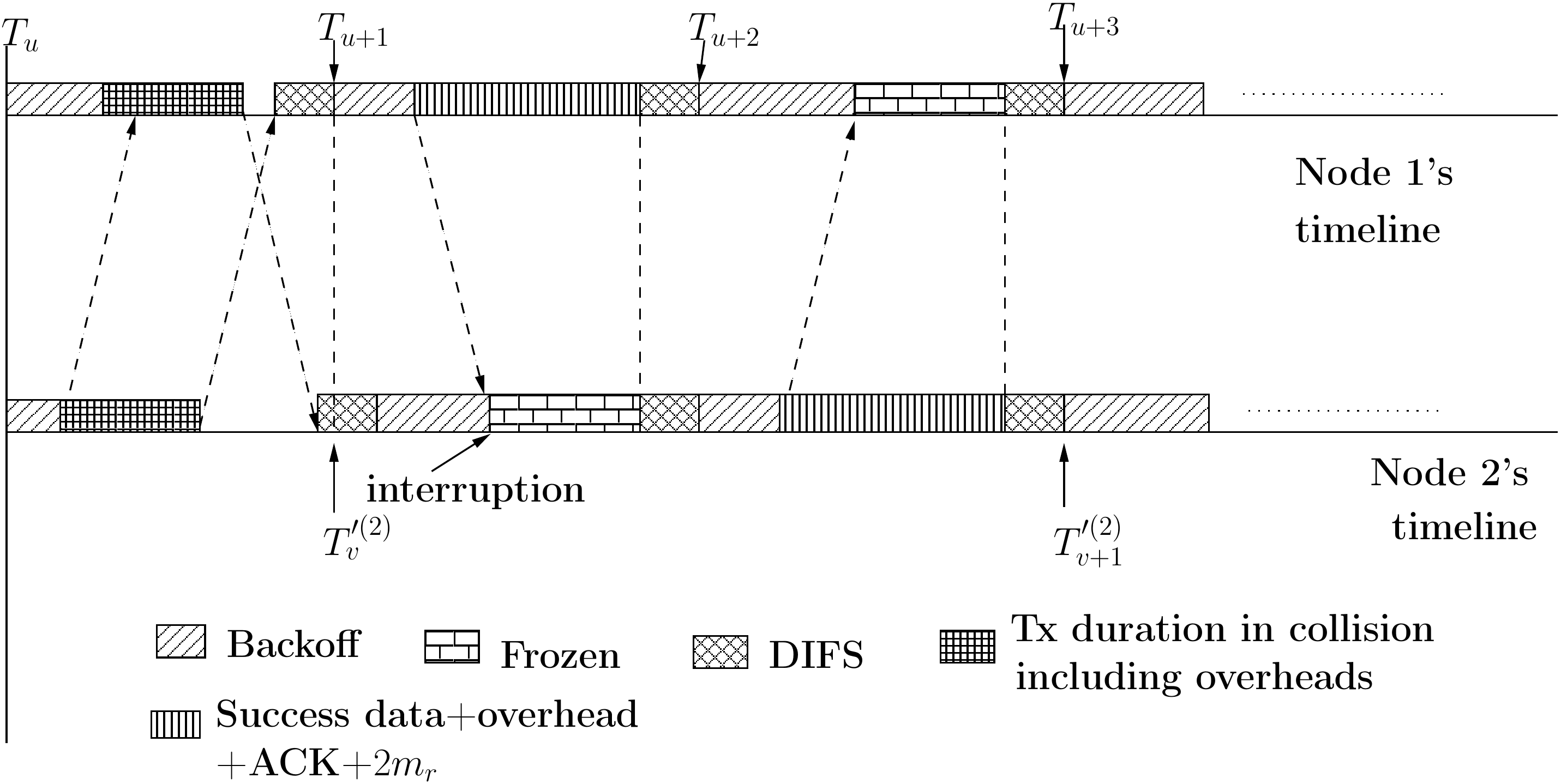}
\caption{\textbf{Backoff Cycles} for a tagged node, Node~2 in this case. The two timelines demonstrate the system evolution in unconditional time over three consecutive transmission cycles, with $T_u$,\ldots, $T_{u+3}$ being the start and end points of the transmission cycles. The explanation of the evolution of these timelines is similar to those in Figures~\ref{fig:misalignment-explain} and \ref{fig:tx-cycle-explain-del}. Denote by $T^{\prime (i)}_v$, \emph{the start of the transmission cycle following the $v^{th}$ transmission by the tagged node, $i$}, Node~2 in this example. The interval $[T^{\prime (2)}_v,T^{\prime (2)}_{v+1}]$ is called a \emph{backoff cycle} of Node~2, since in this interval, Node~2 completes one full backoff. Note that the tagged node can have \emph{exactly one attempt (backoff completion)}, and several intermediate backoff \emph{interruptions} in a backoff cycle. During each system transmission cycle $[T_u, T_{u+1}]$, any node can have at most one backoff segment. Thus the backoff chosen at the start of a tagged node's backoff cycle is partitioned into several backoff segments over a random number of system transmission cycles during the tagged node's backoff cycle. Thus, a backoff cycle can encompass several transmission cycles during which the tagged node was interrupted (i.e., did not attempt).}
\label{fig:bo-cycle-explain-del}
%\vspace{-5mm}
\end{center}
\normalsize
%\vspace{-6mm}
\end{figure}
We consider the evolution of the process at the tagged node, say Node~$i$, and identify embedding instants $T^{\prime (i)}_v$ in this process as explained in Figure~\ref{fig:bo-cycle-explain-del}, where the transmission cycle break-points $T_u,\ldots$ are shown, along with the epochs $T^{\prime (2)}_v\ldots$ for Node~2 (the tagged node). After each such epoch, the tagged node samples a new backoff, using its current backoff stage $S_v$. We associate with each $T^{\prime (i)}_v$, three states: (i) $S_v\in\{0,1,\ldots,K\}$, Node~$i's$ new backoff stage, (ii) $X_v\in\{0_s,0_c,\pm 1,\ldots, \pm m\}$, Node~$i's$ \emph{relative} misalignment w.r.t the other nodes at $T^{\prime (i)}_v$, where $X_v=+k$ means Node~$i$ will start backoff at $T^{\prime (i)}_v+k$, and $X_v=-k$ means Node~$i$ starts backoff at $T^{\prime (i)}_v$, while all the others start at $T^{\prime (i)}_v+k$.  Observe that $S_v > 0\Rightarrow X_v\neq 0_s$, since a successful transmission by Node~$i$ would have reset $S_v$ to zero. (iii) $N_v\in\{1,\ldots,n\}$, number of nodes (including the tagged Node~$i$) that attempted in the just concluded \emph{transmission cycle}. For $n=2$ and arbitrary $m$, $N_v$ is completely determined by $X_v$ (e.g., $X_v = 0_s \Rightarrow N_v = 1$), thus reducing the state space. On the other hand, for $m=0$ and arbitrary $n$, $X_v$ is completely determined by $N_v$ (e.g., $N_v>1\Rightarrow X_v=0_c$), thus again reducing the state space; this is, in fact, what was done in Section~\ref{subsec:tagged-node-evolution}.   

Notice from Figure~\ref{fig:bo-cycle-explain-del} that as before (Section~\ref{subsec:tagged-node-evolution}), \emph{transmission cycles} are common to the entire system, whereas \emph{backoff cycles} are defined for each node. Each backoff cycle of a node \emph{comprises one or more transmission cycles of the system}. \emph{The backoff cycle of a tagged node can comprise several successful transmissions and/or collisions by the other nodes}, and ends at the end of a transmission cycle in which the tagged node transmits. 

In the same vein as Approximations~(A3) and (A4) in Section~\ref{subsec:tagged-node-evolution}, we make the following approximations.

\noindent
 \textbf{(A3)} Node~$i$ samples its successive back-offs from a uniform distribution, as in the standard. When a new backoff cycle starts for Node~$i$, if $X_v=0_s$, the other nodes, conditioned on being in backoff, attempt independently in each slot with probability $\bd$ \emph{until the end of the first transmission cycle within this backoff cycle}. If $X_v\neq 0_s$, $N_v-1$ of the nodes, conditioned on being in backoff, attempt independently in each slot with probability $\bc$, and the remaining $n-N_v$ nodes, conditioned on being in backoff, attempt independently in each slot with probability $\bd$ \emph{until the end of the first transmission cycle within this backoff cycle}.\hfill \Square

\noindent
\textbf{(A4)} If Node~$i$ is interrupted within a backoff cycle due to attempts by $n_a$ other nodes ($1\leq n_a\leq n-1$), thus freezing its backoff (see Figure~\ref{fig:bo-cycle-explain-del}), then in the next transmission cycle within this backoff cycle, Node~$i$ resumes its residual backoff countdown, all the $n-1-n_a$ nodes (excluding Node~$i$) that did not attempt in the previous transmission cycle attempt independently in each slot with probability $\bd$, conditioned on being in backoff, while the $n_a$ nodes that attempted in the previous transmission cycle attempt with probability $\beta_c$ or $\beta_s$ (depending on whether the previous transmission cycle ended in collision or success, i.e., whether $n_a > 1$ or $n_a=1$) in each slot, conditioned on being in backoff.\hfill \Square

% \noindent
% \textbf{(A3)} Node~$i$ samples its successive back-offs from a uniform distribution, as in the standard. When a new backoff cycle starts for Node~$i$, if $X_v=0_s$ (respectively $X_v\neq 0_s$), the other node, say Node~$j$, conditioned on being in backoff, attempts independently in each slot with probability $\bd$ (respectively, $\bc$) \emph{until the end of the first transmission cycle within this backoff cycle}.\hfill \Square
% 
% \noindent
% \textbf{(A4)} If Node~$i$ is interrupted within a backoff cycle, thus freezing its backoff in the first transmission cycle within the backoff cycle (see Figure~\ref{fig:bo-cycle-explain-del}), then from the point of interruption until the backoff completion of Node~$i$, Node~$j$ attempts independently in each slot with probability $\bs$, conditioned on being in backoff\footnote{(A3) and (A4) have to be slightly modified for $n>2$. See Section~VIII in \cite{techreport}. Under the modified assumptions, $\{(S_v,X_v,M_v),T^{\prime}_v\}$ is again an MRP, and the analysis can be carried out.}.\hfill \Square 

Under assumptions~(A3)-(A4), observe that the process $\{(S_v,X_v,N_v), T^{\prime (i)}_v\}$ is a Markov Renewal process (MRP), with the state space of the embedded Markov chain being $\{0,\ldots,K\}\times\{0_s,0_c,\pm 1,\ldots,\pm m\}\times\{1,\ldots,n\}$. We shall develop the details here for $n=2$ and arbitrary $m$.

In this case, the process $\{(S_v,X_v), T^{\prime (i)}_v\}$ is a Markov Renewal process (MRP) with state space of the embedded Markov chain being $\{0,\ldots,K\}\times\{0_s,0_c,\pm 1,\ldots,\pm m\}$. We now proceed to derive the transition structure of the embedded Markov chain. We denote by $W_s$, the contention window size for backoff stage $s$, $s\in\{0,1,\ldots,K\}$. We denote the tagged node as Node~$i$, and the only other node as Node~$j$.

\subsubsection{Transition structure of the embedded Markov chain for $n=2$ and arbitrary $m$}
\label{subsubsec:tagged-node-embedded-dtmc-transition}

Denote by $Q_I[(s_2,x_2)|(s_1,x_1)]$ (respectively, $P_{nI}[(s_2,x_2)|(s_1,x_1)]$) the probability that Node~$i$ is (respectively, is not) interrupted in a backoff cycle starting in state $(s_1,x_1)$, \emph{and} its backoff completion results in state $(s_2,x_2)$.  

Let $Q$ denote the transition probability matrix of the embedded DTMC at the epochs $T^{\prime (i)}_v$. Then, we can write, for any $s\in\{0,\ldots,K\}$, any $x\in\{0_s,0_c,\pm 1,\ldots,\pm m\}$, and any $x^{\prime}\in\{0_c,\pm 1,\ldots,\pm m\}$, 

\begin{align}
 Q((s,x),(0,0_s)) &= P_{nI}[(0,0_s)|(s,x)] + Q_I[(0,0_s)|(s,x)]\label{eqn:Q-0-0s}\\
Q((s,x),((s+1)mod(K+1),x^{\prime})) &= P_{nI}[((s+1)mod(K+1),x^{\prime})|(s,x)] + Q_I[((s+1)mod(K+1),x^{\prime})|(s,x)]\label{eqn:Q-s-plus1-xprime}
\end{align}
All other entries in $Q$ are zero; since we embedded after transmissions of the tagged node, there are only two possibilities: success or collision of the tagged node's transmission. 

We next compute the probabilities $Q_{I}[(\cdot,\cdot)|(\cdot,\cdot)]$, and $P_{nI}[(\cdot,\cdot)|(\cdot,\cdot)]$. 

\subsubsection{Computation of $Q_{I}[(\cdot,\cdot)|(\cdot,\cdot)]$}
\label{subsubsec:compute-prob-int}

Define $h(b,x)$ as the probability that Node~$i's$ subsequent backoff completion leads to a relative misalignment of $x\in\{0_s,0_c,\pm 1,\ldots,\pm m\}$ w.r.t the other node, given that Node~$i$ started with a residual backoff of $b$ after an interruption. These probabilities can be computed recursively as follows:

Note that when Node~$i$ is interrupted in a backoff cycle, (A4) is in force; thus Node~$j$ attempts w.p $\bs$ in each slot, conditioned on being in backoff. Let us compute $h(b,+k)$. When Node~$i$ starts with a residual backoff $b$ after interruption, there are two possibilities:

\noindent

1. Node~$i$ completes its backoff without further interruption, and ends up with relative misalignment of $+k$ w.r.t Node~$j$. From left panel of Figure~\ref{fig:misalignment-explain}, this can happen only if Node~$j$ attempts $k$ slots after Node~$i$, i.e., with probability $(1-\bs)^{b+k-1}\bs$. 

\noindent

2. Node~$i$ is interrupted again; this happens if Node~$j$ attempts at some slot $w$ such that $1\leq w\leq (b-m-1)$, provided $b\geq (m+2)$, so that Node~$i$ hears from Node~$j$ at slot $(w+m)\leq (b-1)$. Thus, the residual backoff of Node~$i$ following the interruption will be $b-(w+m)$.

Combining these two possibilities, we can write, for all $1\leq b\leq W_K-1$, and for all $k\in\{1,\ldots,m\}$
\begin{align}
 h(b,+k) &= (1-\bs)^{b+k-1}\bs + [\sum_{w=1}^{b-m-1}(1-\bs)^{w-1}\bs\nonumber\\
& \times h(b-(w+m),+k)]\ind_{b\geq m+2}\label{eqn:h-b-plus-k}
\end{align}

By similar arguments, we also have, for all $1\leq b\leq W_K-1$, and for all $k\in\{0_c,1,\ldots,m\}$

\begin{align}
 h(b,-k) &= (1-\bs)^{b-k-1}\bs \ind_{b\geq k+1} + [\sum_{w=1}^{b-m-1}(1-\bs)^{w-1}\bs\nonumber\\
&\times h(b-(w+m),-k)]\ind_{b\geq m+2}\label{eqn:h-b-minus-k}\\
h(b,0_s) &= (1-\bs)^{b+m} + [\sum_{w=1}^{b-m-1}(1-\bs)^{w-1}\bs\nonumber\\
& \times h(b-(w+m),0_s)]\ind_{b\geq m+2}\label{eqn:h-b-0s}
\end{align}

Now, we can compute $Q_I[(\cdot,\cdot)|(\cdot,\cdot)]$ in terms of the above probabilities. Suppose, a backoff cycle starts with state $(s,+k)$, for some $s\in\{0,\ldots,K\}$, and $k\in\{0_c,1,\ldots,m\}$. Note that (A3) is in force. Suppose Node~$i$ samples (uniformly from $\{1,\ldots, W_s\}$) a backoff of $l$ slots. Then, Node~$i$ is slated to attempt at slot $l+k$, and it will be interrupted if Node~$j$ attempts somewhere between slots $1$ and $l+k-m-1$, provided $l\geq (m-k+1)$ so that Node~$i$ hears from Node~$j$ by slot $l+k-1$, and freezes its backoff. Suppose Node~$j$ attempts at slot $w$, $1\leq w\leq l-(m-k+1)$. Then the residual backoff of Node~$i$ following the interruption will be $l+k-(w+m)$, and the conditional probability that the subsequent backoff completion leads to state $(s^{\prime},x^{\prime})$ will be $h(l+k-(w+m),x^{\prime})$, where $s^{\prime}=0$ if $x^{\prime} = 0_s$, and $s^{\prime}=(s+1)mod(K+1)$ otherwise. Thus, we have, for all $s\in\{0,\ldots,K\}$, for all $k\in\{0_c,1,\ldots,m\}$, and for all $x^{\prime}\in\{0_s,0_c,\pm 1,\ldots,\pm m\}$,

\begin{align}
 Q_I[(s^{\prime},x^{\prime})|(s,+k)] &= \frac{1}{W_s}\sum_{l=m-k+1}^{W_s}\sum_{w=1}^{l-(m-k+1)}(1-\bc)^{w-1}\bc\nonumber\\
& \times h(l+k-(w+m),x^{\prime})
\end{align}
with $s^{\prime}=0$ if $x^{\prime} = 0_s$, and $s^{\prime}=(s+1)mod(K+1)$ otherwise.

Using similar arguments, we also have, for all $s\in\{0,\ldots,K\}$, for all $k\in\{1,\ldots,m\}$, and for all $x^{\prime}\in\{0_s,0_c,\pm 1,\ldots,\pm m\}$,
\begin{align}
 Q_I[(s^{\prime},x^{\prime})|(s,-k)] &= \frac{1}{W_s}\sum_{l=m+k+1}^{W_s}\sum_{w=1}^{l-(m+k+1)}(1-\bc)^{w-1}\bc\nonumber\\
&\times h(l-(w+k+m),x^{\prime})\\
Q_I[(s^{\prime},x^{\prime})|(0,0_s)] &= \frac{1}{W_0}\sum_{l=m+1}^{W_0}\sum_{w=1}^{l-(m+1)}(1-\bd)^{w-1}\bd\nonumber\\
& \times h(l-(w+m),x^{\prime})
\end{align}
with $s^{\prime}=0$ if $x^{\prime} = 0_s$, and $s^{\prime}=(s+1)mod(K+1)$ otherwise.

\subsubsection{Computation of $P_{nI}[(\cdot,\cdot)|(\cdot,\cdot)]$}
\label{subsubsec:compute-prob-no-int}
Again, let us start with the states with no misalignment.

\noindent
\emph{Computation of $P_{nI}[(\cdot,\cdot)|(0,0_s)]$ and $P_{nI}[(\cdot,\cdot)|(s,0_c)]$}

First observe that starting from state $(0,0_s)$, transition probability to any state with backoff stage other than 0 or 1, is zero. Similarly, starting from state $(s,0_c)$, transition probability to any state with backoff stage other than 0 or $(s+1)mod(K+1)$, is zero.

Suppose the backoff cycle starts with state $(0,0_s)$. Suppose Node~$i$ samples (uniformly from $\{1,\ldots, W_0\}$) a backoff of $l$ slots. 
\begin{enumerate}
 \item Node~$i$ will complete its backoff without interruption, and the resulting attempt will be successful if Node~$j$ does not attempt between slot 1 and slot $l+m$, both inclusive; this happens with probability $(1-\bd)^{(l+m)}$. Thus, we have
\begin{equation}
 P_{nI}[(0,0_s)|(0,0_s)] = \frac{1}{W_0}\sum_{l=1}^{W_0}(1-\bd)^{(l+m)}\label{eqn:p-ni-0-0s-to-0-0s}
\end{equation}
 
\item Node~$i$ will complete its backoff without interruption, and the resulting attempt will encounter a collision leading to state $(1,0_c)$ if Node~$j$ also attempts exactly at the end of slot $l$; this happens with probability $(1-\bd)^{l-1} \bd$. Thus, we have
\begin{equation}
 P_{nI}[(1,0_c)|(0,0_s)] = \frac{1}{W_0}\sum_{l=1}^{W_0}(1-\bd)^{l-1} \bd\label{eqn:p-ni-0-0s-to-1-0c}
\end{equation}

\item Node~$i$ will complete its backoff without interruption, and the resulting attempt will encounter a collision leading to state $(1,+k)$ if Node~$j$ attempts $k$ slots later than Node~$i$ in the current cycle (recall Figure~\ref{fig:misalignment-explain}, and the associated explanation in Section~\ref{subsec:misalignment}), i.e., at slot $l+k$. This happens with probability $(1-\bd)^{(l+k-1)} \bd$. Thus, we have, for all $k\in\{1,\ldots,m\}$,
\begin{equation}
 P_{nI}[(1,+k)|(0,0_s)] = \frac{1}{W_0}\sum_{l=1}^{W_0}(1-\bd)^{(l+k-1)} \bd\label{eqn:p-ni-0-0s-to-1-plus-k}
\end{equation}

\item Finally, using similar arguments, for all $k\in\{1,\ldots,m\}$,
\begin{equation}
 P_{nI}[(1,-k)|(0,0_s)] = \frac{1}{W_0}\sum_{l=k+1}^{W_0}(1-\bd)^{(l-k-1)} \bd\label{eqn:p-ni-0-0s-to-1-minus-k}
\end{equation}
\end{enumerate}

When the backoff cycle starts with state $(s,0_c)$, for any $s\in\{0,\ldots, K\}$, we can use identical arguments as before to write, for any $k\in\{1,\ldots,m\}$,
\begin{align}
P_{nI}[(0,0_s)|(s,0_c)] &= \frac{1}{W_s}\sum_{l=1}^{W_s}(1-\bc)^{(l+m)}\label{eqn:p-ni-s-0c-to-0-0s}\\
P_{nI}[((s+1)mod(K+1),0_c)|(s,0_c)] &= \frac{1}{W_s}\sum_{l=1}^{W_s}(1-\bc)^{l-1} \bc\label{eqn:p-ni-s-0c-to-s-plus1-0c}\\
P_{nI}[((s+1)mod(K+1),+k)|(s,0_c)] &= \frac{1}{W_s}\sum_{l=1}^{W_s}(1-\bc)^{(l+k-1)} \bc\label{eqn:p-ni-s-0c-to-s-plus1-plus-k}\\
P_{nI}[((s+1)mod(K+1),-k)|(s,0_c)] &= \frac{1}{W_s}\sum_{l=k+1}^{W_s}(1-\bc)^{(l-k-1)} \bc\label{eqn:p-ni-s-0c-to-s-plus1-minus-k} 
\end{align}

\noindent
\emph{Computation of $P_{nI}[(\cdot,\cdot)|(s,+k)]$ and $P_{nI}[(\cdot,\cdot)|(s,-k)]$}

When the backoff cycle starts with state $(s,+k)$, the tagged Node~$i$ will start its backoff countdown after $k$ slots, while the other node, i.e., Node~$j$ starts its backoff immediately. Suppose Node~$i$ samples (uniformly from $\{1,\ldots, W_s\}$) a backoff of $l$ slots. Thus, Node~$i$ is supposed to make an attempt at slot $l+k$.
\begin{enumerate}
 \item Node~$i$ will not be interrupted, and its resulting attempt will be successful if Node~$j$ does not attempt until slot $l+k+m$ (starting from slot 1); this hapens with probability $(1-\bc)^{(l+m+k)}$. Thus, we have, for any $s\in\{0,\ldots,K\}$, and any $k\in\{1,\ldots,m\}$,
\begin{equation}
 P_{nI}[(0,0_s)|(s,+k)] = \frac{1}{W_s}\sum_{l=1}^{W_s}(1-\bc)^{(l+m+k)}\label{eqn:p-ni-s-plus-k-to-0-0s}
\end{equation}

\item Node~$i$ will not be interrupted, and its attempt will encounter a collision leading to the state $((s+1)mod(K+1),-k^{\prime})$, for any $k^{\prime}\in\{1,\ldots,m\}$, if Node~$j$ attempts $k^{\prime}$ slots earlier than Node~$i$ in the current cycle, i.e., Node~$j$ attempts at slot $l+k-k^{\prime}$; this happens with probability $(1-\bc)^{(l+k-k^{\prime}-1)}\bc$, provided $l\geq (k^{\prime}-k+1)$. Thus, we have, for any $s\in\{0,\ldots,K\}$, any $k\in\{1,\ldots,m\}$, and any $k^{\prime}\in\{1,\ldots,m\}$,
\begin{align}
 P_{nI}[((s+1)mod(K+1),-k^{\prime})|(s,+k)] &= \frac{1}{W_s}\sum_{l=\max\{1,k^{\prime}-k+1\}}^{W_s}(1-\bc)^{(l+k-k^{\prime}-1)}\bc \label{eqn:p-ni-s-plus-k-to-s-plus1-minus-kprime}
\end{align}

\item Using similar arguments, we also have, for any $s\in\{0,\ldots,K\}$, any $k\in\{1,\ldots,m\}$, and any $k^{\prime}\in\{0_c,1,\ldots,m\}$,
\begin{align}
 P_{nI}[((s+1)mod(K+1),+k^{\prime})|(s,+k)] &= \frac{1}{W_s}\sum_{l=1}^{W_s}(1-\bc)^{(l+k+k^{\prime}-1)}\bc \label{eqn:p-ni-s-plus-k-to-s-plus1-plus-kprime}
\end{align}  
\end{enumerate}

Finally, when the backoff cycle starts with state $(s,-k)$, we can use very similar arguments as before to obtain, for any $s\in\{0,\ldots,K\}$, any $k\in\{1,\ldots,m\}$, and any $k^{\prime}\in\{0_c,1,\ldots,m\}$,
\begin{align}
 P_{nI}[(0,0_s)|(s,-k)] &= \frac{1}{W_s}\sum_{l=1}^{W_s}(1-\bc)^{(l+m-k)}\label{eqn:p-ni-s-minus-k-to-0-0s}\\
P_{nI}[((s+1)mod(K+1),+k^{\prime})|(s,-k)] &= \frac{1}{W_s}\sum_{l=\max\{1,k-k^{\prime}+1\}}^{W_s}(1-\bc)^{(l+k^{\prime}-k-1)}\bc \label{eqn:p-ni-s-minus-k-to-s-plus1-plus-kprime}\\
P_{nI}[((s+1)mod(K+1),-k^{\prime})|(s,-k)] &= \frac{1}{W_s}\sum_{l=k^{\prime}+k+1}^{W_s}(1-\bc)^{(l-k-k^{\prime}-1)}\bc \label{eqn:p-ni-s-minus-k-to-s-plus1-minus-kprime}
\end{align}

This completes the derivation of the transition structure of the embedded DTMC at the epochs $T^{\prime (i)}_v$. It is easy to observe that the embedded DTMC is finite, irreducible (from any state, the state $(0,0_s)$ can be reached in one step, and from $(0,0_s)$, any state can be reached, provided the attempt rates are such that the transition probabilities given by Eqns.~\ref{eqn:Q-0-0s} and \ref{eqn:Q-s-plus1-xprime} are positive), and hence \emph{positive recurrent}. We denote by $\psi$, the stationary distribution of this Markov chain, which can be obtained as the unique solution to the system of equations $\psi = \psi Q$, subject to $\psi$ being a probability distribution.  

Our objective from this exercise was to obtain the mean attempt rates $\bd$, $\bs$, and $\bc$, which we proceed to do next.

Recall that $\bs$ and $\bc$ are the mean attempt rates of a node in a transmission cycle after it resumes backoff following a succeessful transmission, and a collision, respectively, while $\bd$ is the mean attempt rate of a node in a transmission cycle after it resumes backoff following an interruption. Thus, observe that in a backoff cycle of a tagged node, the contributions to $\bs$ and $\bc$ come from only the first transmission cycle within the backoff cycle (i.e., until the point of first interruption of the tagged node within the backoff cycle), whereas the remainder (if any) of the backoff cycle (i.e., from the point of first interruption until backoff completion) contributes towards $\bd$. 

\subsubsection{Computation of $\bd$ for $n=2$, arbitrary $m$}
\label{subsubsec:bd} 

Proceeding along the same lines as in Section~\ref{subsubsec:bd-general-n}, we have
\begin{equation}
 \bd = \frac{\sum_{(s,x)}\psi(s,x)P_I(s,x)}{\sum_{(s,x)}\psi(s,x)E\Bcal_r(s,x)}\: a.s\label{eqn:bd-expression}
\end{equation}
where, $P_{I}(s,x)$ is the probability that Node~$i$ is interrupted when the backoff cycle starts in state $(s,x)$, and $E\Bcal_r(s,x)$ is the mean residual backoff counted by Node~$i$ from its first interruption until its backoff completion in a backoff cycle that started with state $(s,x)$; they can be computed as follows.

\noindent

\emph{Computation of $P_{I}(\cdot,\cdot)$:}
%\label{subsubsec:compute-prob-interruption}
Let us first consider the states with no misalignment.

\noindent
\emph{Computation of $P_{I}(0,0_s)$ and $P_I(s,0_c)$:}

Consider first, the state $(0,0_s)$. Suppose Node~$i$ samples (uniformly from $\{1,2,\ldots, W_0\}$) a backoff of $l$ slots. To be interrupted, it must hear a transmission from Node~$j$ within slot $(l-1)$. Thus, Node~$j$ must make an attempt between slots 1 to $(l-1-m)$, both inclusive, which happens with probability $1-(1-\bd)^{(l-m-1)}$, provided $l>m+1$. Thus, we have
\begin{equation}
 P_{I}(0,0_s) = \frac{1}{W_0}\sum_{l=m+2}^{W_0}[1-(1-\bd)^{(l-m-1)}]\label{eqn:p-i-0-0s}
\end{equation}

By exactly same arguments, we also have
\begin{equation}
 P_I(s,0_c) = \frac{1}{W_s}\sum_{l=m+2}^{W_s}[1-(1-\bc)^{(l-m-1)}]\:\forall s\in\{0,1,\ldots,K\}\label{eqn:p-i-s-0c}
\end{equation}

\noindent
\emph{Computation of $P_{I}(s,+k)$ and $P_I(s,-k)$:}

When the state at the start of the cycle is $(s,+k)$, Node~$i$ will start its backoff $k$ slots later, while Node~$j$ starts its backoff immediately. Suppose Node~$i$ samples (uniformly from $\{1,\ldots, W_s\}$) a backoff of $l$ slots. Then Node~$i$ is supposed to make an attempt at slot $l+k$. To be interrupted, therefore, it must hear from Node~$j$ by slot $l+k-1$, which in turn requires Node~$j$ to make an attempt by slot $l+k-1-m$; this happens with probability $1-(1-\bc)^{(l-(m-k+1))}$, provided $l>(m-k+1)$. Thus, we have, for all $s\in\{0,1,\ldots,K\}$, and for all $k\in\{1,\ldots,m\}$,
\begin{equation}
 P_{I}(s,+k) = \frac{1}{W_s}\sum_{l=m-k+2}^{W_s}[1-(1-\bc)^{(l-(m-k+1))}]\label{eqn:p-i-s-plus-k}
\end{equation}

Using similar arguments, for all $s\in\{0,1,\ldots,K\}$, and for all $k\in\{1,\ldots,m\}$,
\begin{equation}
 P_{I}(s,-k) = \frac{1}{W_s}\sum_{l=m+k+2}^{W_s}[1-(1-\bc)^{(l-(m+k+1))}]\label{eqn:p-i-s-minus-k}
\end{equation}

\noindent
\emph{Computation of $E\Bcal_r(s,x)$:}

Consider a backoff cycle starting with state $(s,+k)$. Suppose Node~$i$ samples (uniformly from $\{1,\ldots, W_s\}$) a backoff of $l$ slots. As was explained earlier, to interrupt Node~$i$, Node~$j$ must make an attempt by slot $l+k-1-m$, provided $l\geq (m-k+2)$. Suppose Node~$j$ makes an attempt at slot $w$, $1\leq w\leq l+k-1-m$; this happens with probability $(1-\bc)^{w-1}\bc$. Thus, Node~$i$ hears from Node~$j$ at slot $(w+m)$, and freezes its backoff. Thus, the residual backoff of Node~$i$ is $l+k-(w+m)$. Thus, we have, for any $k\in\{0_c,1,\ldots,m\}$, and any $s\in\{0,\ldots,K\}$,
\begin{align}
 E\Bcal_r(s,+k) &= \frac{1}{W_s}\sum_{l=m-k+2}^{W_s}\sum_{w=1}^{l-(m-k+1)}(1-\bc)^{w-1}\bc (l+k-(w+m))\label{eqn:EB-r-s-plus-k}
\end{align}

By similar arguments, we also have, for any $k\in\{1,\ldots,m\}$, and any $s\in\{0,\ldots,K\}$,
\begin{align}
 E\Bcal_r(s,-k) &= \frac{1}{W_s}\sum_{l=m+k+2}^{W_s}\sum_{w=1}^{l-(m+k+1)}(1-\bc)^{w-1}\bc (l-(w+k+m))\label{eqn:EB-r-s-minus-k}\\
E\Bcal_r(0,0_s) &= \frac{1}{W_0}\sum_{l=m+2}^{W_0}\sum_{w=1}^{l-(m+1)}(1-\bd)^{w-1}\bd (l-(w+m))\label{eqn:EB-r-0-0s}
\end{align}

\subsubsection{Computation of $\bs$}
\label{subsubsec:bs}
Looking at the backoff evolution of the tagged Node~$i$, we can define $\bs$ more formally as
\begin{equation*}
 \bs = \lim_{t\to\infty}\frac{\sum_{k=1}^{N_s(t)}\ind_{\{\text{Node~$i$ was not interrupted in backoff cycle $k$}\}}}{\sum_{k=1}^{N_s(t)}B_{s,k}}
\end{equation*}
 where, $N_s(t)$ is the number of backoff cycles until time $t$ that start with the state $(0,0_s)$ (implying that Node~$i$ was successful in the previous transmission cycle), and $B_{s,k}$ is the backoff counted by Node~$i$ \emph{in the transmission cycle that started along with backoff cycle $k$}; in other words, $B_{s,k}$ is the backoff counted by Node~$i$ until it gets interrupted, or completes its backoff, whichever is earlier. Thus, the denominator is the total backoff counted by Node~$i$ until time $t$, in those transmission cycles that followed a successful transmission by Node~$i$. Similarly, the numerator is the total number of attempts by Node~$i$ until time $t$ in those transmission cycles that followed a successful transmission by Node~$i$. 

Denote by $\Bcal_s$, the random variable representing the backoff counted by Node~$i$ in the first transmission cycle within a backoff cycle starting in state $(0,0_S)$. Then, by Markov regenerative theory, it follows that
\begin{equation}
 \bs = \frac{1-P_I(0,0_s)}{E\Bcal_s(0,0_s)} \: a.s.
\end{equation}
where, $E\Bcal_s(0,0_s)$ is the mean time spent in backoff by Node~$i$ until it gets interrupted, or completes its backoff in the backoff cycle starting in state $(0,0_s)$, and can be computed as follows.

Suppose Node~$i$ samples (uniformly from $\{1,\ldots, W_0\}$) a backoff of $l$ slots. As explained earlier, to interrupt Node~$i$, the other node must attempt within slot $(l-1-m)$, which is possible only if $l\geq (m+2)$. Now there are three possibilities:
\begin{enumerate}
 \item $l<(m+2)$. Node~$i$ cannot be interrupted; its backoff count is $l$.
 \item $l\geq (m+2)$, but Node~$j$ does not attempt up to $(l-1-m)$. Then again, Node~$i$ does not get interrupted, and its backoff count is $l$.
 \item $l\geq (m+2)$, and Node~$j$ attempts at slot $w$, $1\leq w\leq l-1-m$. Then, Node~$i$ is interrupted, and its backoff counted until interruption is $w+m$.
\end{enumerate}

Combining all of these together,
\begin{align}
 E\Bcal_s(0,0_s) &= \frac{1}{W_0}\frac{(m+1)(m+2)}{2} + \frac{1}{W_0}\sum_{l=m+2}^{W_0}[(1-\bd)^{(l-m-1)}l\nonumber\\
&+\sum_{w=1}^{l-(m+1)}(1-\bd)^{w-1}\bd(w+m)]\label{eqn:EB-s}
\end{align}

\subsubsection{Computation of $\bc$}
\label{subsubsec:bc}

Proceeding along the same lines as in Section~\ref{subsubsec:bc-general-n}, we have
\begin{equation}
 \bc = \frac{\sum_{(s,x)\neq (0,0_s)}\psi(s,x)(1-P_I(s,x))}{\sum_{(s,x)\neq (0,0_s)}\psi(s,x)E\Bcal_c(s,x)}\: a.s
\end{equation}
where, $E\Bcal_c(s,x)$ is the mean time spent in backoff by Node~$i$ until it gets interrupted, or completes its backoff in the backoff cycle starting in state $(s,x)$, and can be computed as follows.

Consider a backoff cycle starting with state $(s,+k)$. Suppose Node~$i$ samples (uniformly from $\{1,\ldots, W_s\}$) a backoff of $l$ slots. As explained earlier in Section~\ref{subsubsec:bd}, to interrupt Node~$i$, Node~$j$ must make an attempt by slot $l+k-1-m$, provided $l\geq (m-k+2)$. Now, there are three possibilities:
\begin{enumerate}
 \item $l<(m-k+2)$. Node~$i$ cannot be interrupted, and its backoff count is $l$.
 \item $l\geq (m-k+2)$, but Node~$j$ does not attempt up to $l-(m-k+1)$. Again, Node~$i$ does not get interrupted, and its backoff count is $l$.
 \item $l\geq (m-k+2)$, and Node~$j$ attempts at slot $w$, $1\leq w\leq l-(m-k+1)$. Then, Node~$i$ is interrupted, and its backoff count until interruption is $(w+m-k)$ (recall that when the backoff cycle starts with state $(\cdot,+k)$, Node~$i$ starts its backoff process after slot $k$).
\end{enumerate}

Combining these together, we have, for any $k\in\{0_c,1,\ldots,m\}$, and any $s\in\{0,\ldots,K\}$,
\begin{align}
 E\Bcal_c(s,+k) &= \frac{1}{W_s}\sum_{l=1}^{m-k+1}l + \frac{1}{W_s}\sum_{l=m-k+2}^{W_s}[(1-\bc)^{l-(m-k+1)}l\nonumber\\
&+ \sum_{w=1}^{l-(m-k+1)}(1-\bc)^{w-1}\bc(w+m-k)]\label{eqn:EB-c-s-plus-k}
\end{align}

By similar arguments, we also have, for any $k\in\{1,\ldots,m\}$, and any $s\in\{0,\ldots,K\}$,
\begin{align}
 E\Bcal_c(s,-k) &= \frac{1}{W_s}\sum_{l=1}^{m+k+1}l + \frac{1}{W_s}\sum_{l=m+k+2}^{W_s}[(1-\bc)^{l-(m+k+1)}l\nonumber\\
&+ \sum_{w=1}^{l-(m+k+1)}(1-\bc)^{w-1}\bc(w+m+k)]\label{eqn:EB-c-s-minus-k}
\end{align}

Equations~\ref{eqn:h-b-plus-k}-\ref{eqn:EB-c-s-minus-k} together form a system of vector fixed point equations in $(\bd,\bc)$ (observe from Eqns.~\ref{eqn:p-i-0-0s} and \ref{eqn:EB-s} that $\bs$ is a deterministic function of $\bd$ alone), which can be solved using an iterative procedure until convergence to obtain the attempt rates $\bd$, $\bs$, and $\bc$.

\subsubsection{Computation of the average attempt rate, $\beta$, over all backoff time}
\label{subsubsec:beta}
The backoff cycle analysis can be used to obtain the long run average attempt rate, $\beta$, averaged over all backoff time (irrespective of system state). 

To obtain $\beta$, note that each backoff cycle contains exactly one attempt by the tagged node, and the backoff counted by the tagged node in the entire backoff cycle contributes towards $\beta$. In a backoff cycle starting in state $(s,x)$, the mean backoff counted by the tagged node is clearly $(W_s+1)/2$. Thus, using Markov regenerative analysis, we have
\begin{equation}
 \beta = \frac{1}{\sum_{(s,x)}\psi(s,x)\frac{W_s+1}{2}}\label{eqn:beta-wifi}
\end{equation}

\subsection{Discussion on the existence and uniqueness of the fixed point}
\label{subsec:existence-del}
\begin{theorem}
There exists a fixed point for the system of equations~\ref{eqn:h-b-plus-k}-\ref{eqn:EB-c-s-minus-k} in the set $\mathbf{C}=[1/W_K,1]\times[1/W_K,1]$. 
\end{theorem}

\begin{proof}
The proof follows along exactly the same lines as that of Theorem~\ref{thm:existence-general-n}.
\end{proof}

We do not have proof of uniqueness of the fixed point. However, in our numerical experiments, the iterations always converged to the same solutions even when starting with different initial values. 

\section{Model Validation Through Simulations}
\label{sec:numerical}

To validate our analytical model, we performed extensive simulations on a topology with 2 transmitter-receiver pairs with saturated transmit queues; we assumed equal propagation delay $\Delta$ among all nodes, and varied $\Delta$ across simulations. We used the default backoff parameters of IEEE~802.11b. 

We used the method of simulating the detailed stochastic system model, described in Section~\ref{sec:exact_model_del}, since it is much faster compared to detailed ``off-the-shelf'' event-driven simulation tools such as Qualnet, and gives excellent accuracy (as was demonstrated in Figure~\ref{fig:exact-model-gamma-validation} in Section~\ref{sec:exact_model_del}), while providing more flexibility in examining the finer details of the system evolution (e.g., it is considerably harder to obtain the conditional attempt rates such as $\bd$ from a Qualnet simulation).

\begin{figure}[htpb]
%\footnotesize
\begin{center}
\includegraphics[scale=0.4]{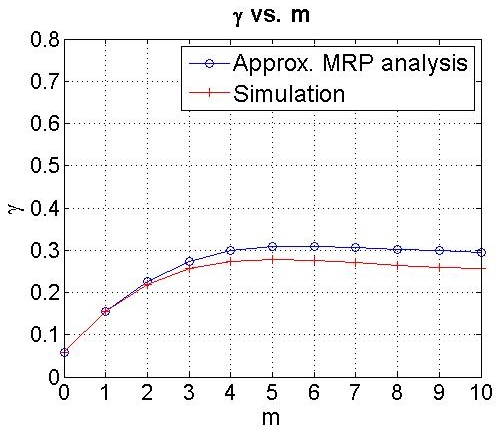}
\caption{Comparison of collision probabilities obtained from the approximate MRP analysis in Section~\ref{sec:mrp-state-dependent-del} against simulations.}
\label{fig:approx-model-gamma-validation}
\vspace{-5mm}
\end{center}
\normalsize
%\vspace{-6mm}
\end{figure}
\begin{figure}[htpb]
%\footnotesize
\begin{center}
\includegraphics[scale=0.3]{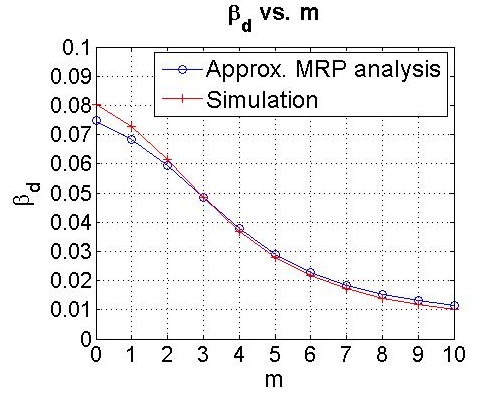}
\hspace{1mm}
\includegraphics[scale=0.3]{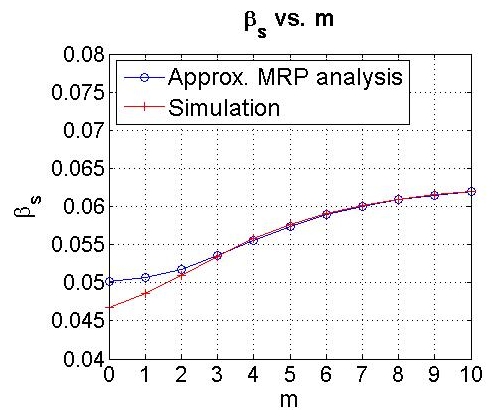}
\hspace{1mm}
\includegraphics[scale=0.3]{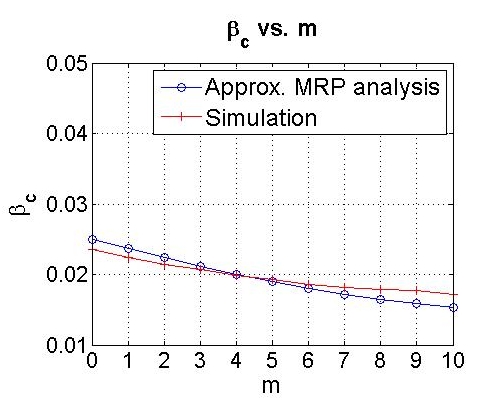}
\vspace{2mm}
\includegraphics[scale=0.32]{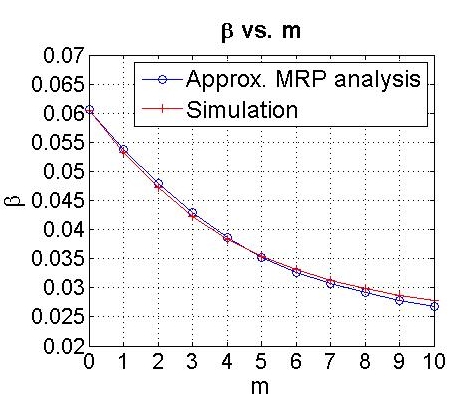}
\hspace{1mm}
\includegraphics[scale=0.3]{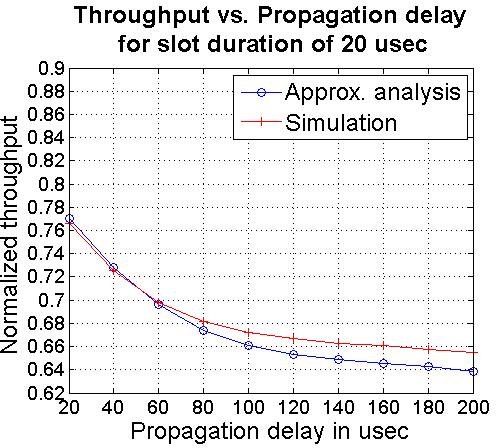}
\caption{Comparison of attempt rates and throughputs obtained from the approximate MRP analysis in Section~\ref{sec:mrp-state-dependent-del} against simulations for $n=2$, arbitrary $m$.}
\label{fig:approx-model-beta-theta-validation}
\vspace{-5mm}
\end{center}
\normalsize
%\vspace{-6mm}
\end{figure}

We first compared the collision probabilities obtained from our approximate analytical model against those obtained from simulations for a range of values of $m$ (the integer ratio of propagation delay and slot duration). Figure~\ref{fig:approx-model-gamma-validation} summarizes the results. The relative errors in the analytical values compared to simulations are no more than 8\%. Also observe that the trend of the collision probability as a function of $m$ is captured well by the approximate analysis.

We also compared the attempt rates, and throughputs obtained from our approximate analytical model against those obtained from simulations. Note that while collision probabilities depend only on $m=\lfloor\frac{\Delta}{\sigma}\rfloor$, the propagation delay in integer multiples of slots (see Section~\ref{sec:exact_model_del}), throughput depends on the actual ratio $\frac{\Delta}{\sigma}$, since it involves computing the actual lengths of the transmission cycle, and the data duration. We compared the throughput obtained from the approximate analysis against simulation results for a range of values of $\Delta$, under default backoff parameters of IEEE~802.11b with $T_d=1028$ bytes (4112 $\mu$secs at 2 Mbps rate), and $T_o=10$ $\mu$secs; the results are summarized in Figure~\ref{fig:approx-model-beta-theta-validation}.

% \begin{figure*}[t]
% %\footnotesize
% \begin{center}
% \includegraphics[scale=0.3]{plots/throughput_vs_prop_delay_sigma20_K6_cwmax10.jpg}
% \hspace{0.1mm}
% \includegraphics[scale=0.3]{plots/throughput_vs_prop_delay_sigma60_K6_cwmax10.jpg}
% \hspace{0.1mm}
% \includegraphics[scale=0.3]{plots/throughput_vs_slot_durn_Delta60_K6_cwmax10.jpg}
% \hspace{0.1mm}
% \includegraphics[scale=0.3]{plots/throughput_vs_slot_durn_Delta120_K6_cwmax10.jpg}
% \hspace{0.1mm}
% \includegraphics[scale=0.3]{plots/throughput_vs_slot_durn_Delta160_K6_cwmax10.jpg}
% \caption{Comparison of throughputs obtained from the approximate analytical model against simulations for different combinations of propagation delay, and slot duration.}
% \label{fig:approx-model-theta-validation}
% \vspace{-5mm}
% \end{center}
% \normalsize
% %\vspace{-6mm}
% \end{figure*} 

From these plots, we can make the following observations:

\noindent
\textbf{Observations:}

\noindent

 1. The errors in the approximate analysis compared to simulations are at most 2-4\%, and 2-3\% respectively in predicting the attempt rates, and throughput, thus validating the accuracy of the analysis.

\noindent

 2. As $m$ increases, $\bs$ monotonically increases, $\bd$, and $\bc$ monotonically decrease. An intuition behind this follows from the intuitive explanation of the short term unfairness property provided in the discussion at the end of Section~\ref{subsec:long-distance-wifi}. At higher propagation delays, due to the high collision probability, the backoff difference of the nodes is stochastically larger, and hence, after a successful transmission in the system, the residual backoff of the interrupted node is also stochastically large. This causes $\bd$ to decrease with increasing $m$. The same argument will also see the successful node (which samples its next backoff from the smallest contention window) attempt again without interruption with a higher likelihood, thus causing $\bs$ to increase with increasing $m$. Since at higher propagation delays, due to the high collision probability, the nodes after a collision sample backoffs from  stochastically larger contention windows (compared to those at lower $m$), $\bc$ decreases with increasing $m$. Also, the overall attempt rate, $\beta$, decreases with increasing propagation delay. This is also intuitive, since due to the higher collision probability, the nodes are likely to spend more time in larger backoff stages, thus increasing the denominator in Eqn.~\ref{eqn:beta-wifi}. 

\noindent
 
3. At higher $m$, $\bs\gg \bd$, \emph{which is a reflection of the short term unfairness property} demonstrated in Section~\ref{subsec:long-distance-wifi}.

\noindent

4. As $m$ increases, \emph{the collision probability $\gamma$ increases at first, but then gradually flattens out}. This can be intuitively explained as follows. For simplicity, consider the case when the backoffs of the two nodes are aligned at the start of a transmission cycle; the conclusions from the other cases are similar. Suppose, $B_1$, and $B_2$ are the backoffs sampled by the two nodes and assume, without loss of generality, $B_1 < B_2$. Suppose further, for simplicity, that $B_1$ and $B_2$ were sampled from the same contention window, say $W_s$. A collision happens when $B_2 \leq B_1 + m$, i.e., $B_2-B_1 \leq m$. Now, (i) Clearly, for a fixed $W_s$, the probability of this event is increasing in $m$, thus causing an increase in collision probability. (ii) However, as collision probability increases with $m$, the nodes tend to sample backoff from a higher contention window, i.e., $W_s$ becomes stochastically larger. Further, it can be shown by an elementary analysis that as $W_s$ increases, the random variable $B_2-B_1$ becomes stochastically larger, and hence the probability of the concerned event decreases. These two opposing effects cause the collision probability to saturate at higher values of $m$.

\noindent
5. On a Linux based machine with 8 GB RAM, the running time of the approximate analysis is at most a few seconds, while that of the stochastic simulation is of the order of several minutes; it takes hours to run the Qualnet simulation, especially when the short term unfairness is severe.  

\newpage
\begin{center}
 \Large{Part~III: Implications for the Protocol}
\end{center}
Apart from providing an accurate prediction of the system performance in the presence of short term unfairness using a parsimonious state representation, the approximate analysis proposed above has several applications, and implications for the protocol, some of which we proceed to illustrate next.

\section{Optimizing Slot Duration for Throughput Maximization}
\label{sec:opt-slot-duration}
Since the approximate analysis is very accurate, we can use this instead of computationally expensive simulations to choose system parameters for performance optimization. In this section, we use the analysis to choose the optimal slot duration for a given propagation delay to maximize system throughput, $\Theta$. 

Observe that $\sigma$ very small $\Rightarrow m=\lfloor\frac{\Delta}{\sigma}\rfloor$ is large $\Rightarrow\gamma$ is high $\Rightarrow$ nodes attempt less frequently; the number of nodes is fixed, the attempt rate per node reduces, while $\gamma$ increases, thus reducing $\Theta$. On the other hand, $\sigma$ very large $\Rightarrow$ backoff durations are large $\Rightarrow$ too much idle time compared to data duration $\Rightarrow$ reduced $\Theta$. Hence, there is an optimal value of $\sigma$ for a given $\Delta$ to maximize $\Theta$.

Further observe that for a fixed $\Delta$, there could be several values of $\sigma$ that can give rise to the same value of $m$; e.g., for $\Delta = 60$ $\mu$secs, any slot duration between 21 and 30 $\mu$secs result in $m=2$, and hence they all lead to the same probability of collision. However, as $\sigma$ increases (e.g., from 21 $\mu$secs to 30 $\mu$secs) keeping $m$ same ($m=$2 in example), $\Theta$ will decrease, since $\gamma$ stays same, and idle time increases.

With the above observation, we adopted the following strategy for obtaining the throughput as a function of $m$ for any fixed $\Delta$. For each $m$, and each fixed $\Delta$ in $\mu$secs, we computed the \emph{least} slot duration in $\mu$secs required to achieve that $m$ for that $\Delta$; this can be easily seen to be $\lfloor\frac{\Delta}{m+1}\rfloor+1$ $\mu$secs. This slot duration was used to compute the throughput for that $(m,\Delta)$ combination. The results are summarized in 4 sets of plots in Figure~\ref{fig:theta-vs-m-approx-model}, where we have plotted the throughput as a function of $m$ for several different values of $\Delta$, keeping other parameters of the protocol fixed at their default values under IEEE~802.11b. From these plots, one can read off, for each $\Delta$, the optimum $m$, and hence the optimum slot duration that maximizes throughput for that $\Delta$. 

\begin{figure*}[t]
%\footnotesize
\begin{center}
\includegraphics[scale=0.32]{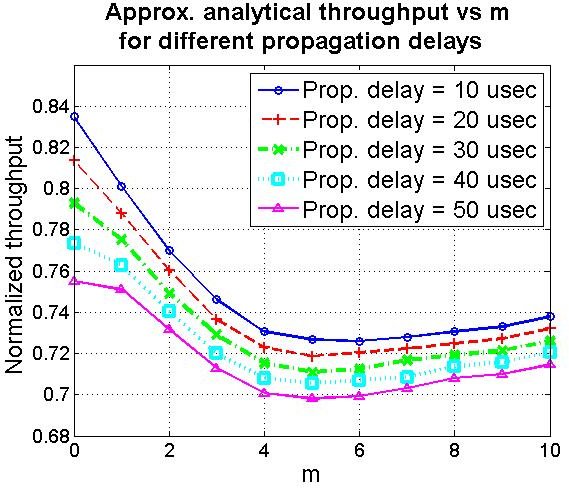}
\hspace{0.1mm}
\includegraphics[scale=0.32]{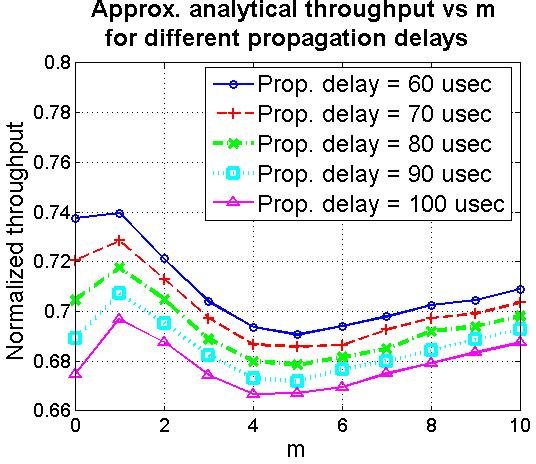}
\hspace{0.1mm}
\includegraphics[scale=0.32]{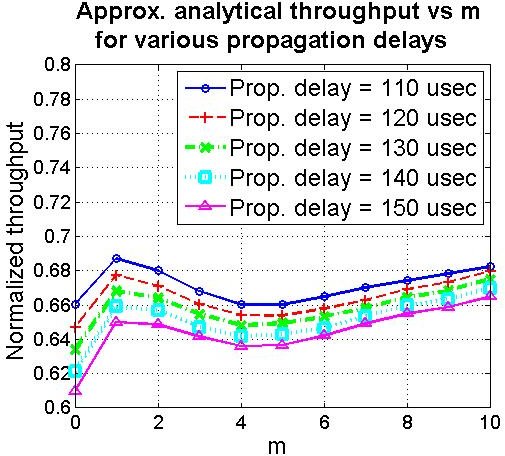}
\hspace{0.1mm}
\includegraphics[scale=0.32]{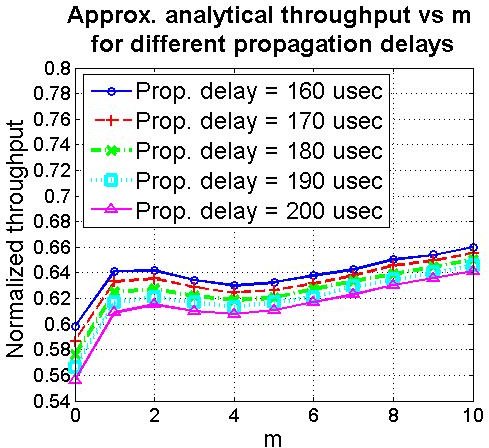}
\caption{Throughput as a function of $m$ for different propagation delays, obtained from the approximate analysis; for each propagation delay, the optimum slot duration can be read off from the plots.}
\label{fig:theta-vs-m-approx-model}
\vspace{-5mm}
\end{center}
\normalsize
%\vspace{-6mm}
\end{figure*} 

From the plots in Figure~\ref{fig:theta-vs-m-approx-model}, we can make the following observations:

\noindent
\textbf{Observations:}

\noindent
 1. For $\Delta\leq 110$ $\mu$secs, $\Theta$ is maximum at $m=0$ or $m = 1$. However, for $\Delta\geq 120$ $\mu$secs, $\Theta$ is maximized at $m = 10$ or beyond. Thus, at lower propagation delays, collision probability dominates throughput, while at higher $\Delta$, slot duration takes over as the dominant factor. This also means that in general, it is not necessarily throughput optimal to make the slot duration comparable to the propagation delay, unlike what has been suggested in some previous literature (see, for example, \cite{radamacher14wild,raman-chebrolu07wifi-rural}).

\noindent
2. Beyond $m = 2$, all the plots exhibit a convex pattern. This can be explained from the $\gamma$ vs. $m$ plot in Panel~1 of Figure~\ref{fig:approx-model-gamma-validation}. As $m$ increases, $\gamma$ increases, causing $\Theta$ to decrease (see our observation at the beginning of this section); but since $\gamma$ gradually flattens out, and $\sigma$ decreases with increasing $m$, the rate of decrease in $\Theta$ also starts decreasing, and eventually $\Theta$ starts increasing with $m$.

\noindent
3. Below $m=2$, we see that for lower propagation delays, $\Theta$ decreases as $m$ increases, but at higher propagation delays, $\Theta$ increases with $m$. One possible explanation for this behavior is as follows. For lower propagation delays, the reduction in slot duration as we go from $m=0$ to $m=1$ is not significant enough to ameliorate the effect of the increase in collision probability. However, at higher propagation delays, the reduction in slot duration as we go from $m=0$ to $m=1$ is considerably large, which causes significant reduction in the system idle time, and more than makes up for the increase in collision probability. Note that this explanation is also consistent with our Observation~1 above.  

\section{Quantifying the Extent of Short Term Unfairness}
\label{sec:stu-measure}
Once we know the attempt rates $\bd,\bc,\bs$ for a system using the procedure described in Section~\ref{subsec:tagged-node-evolution}, we can use the state dependent Bernoulli attempt process model introduced in Section~\ref{subsec:system-evolution-mrp} to quantify the extent of short term unfairness in the system. This is an important measure which can be used for tuning protocol parameters as we shall see later, and is not easy to obtain using state-of-the-art simulation tools such as Qualnet (and cannot at all be obtained using the standard fixed point analysis). We define below, two possible measures of short term unfairness, and show how we can obtain them using our state dependent attempt rate model. 

\subsection{A Throughput Fairness Index for $m=0$, arbitrary $n$}
\label{subsec:fairness-ix}

As we saw in Section~\ref{subsec:system-evolution-mrp}, Figure~\ref{fig:success-processes}, the impact of short term unfairness is to skew the success process in favor of an already successful node, thus introducing high correlation in the success process. With that in mind, we proceed to define a measure of short term unfairness as follows. 

Fix a node, say Node~$1$. Define a \emph{frame} as a block of $L$ consecutive transmission cycles \emph{following a successful transmission by Node~$1$}. Our aim is to compare the average throughputs obtained by all the nodes $1,\ldots,n$, over a frame. When the system has short term unfairness, the average throughput of Node~1 in a frame will be higher than the other nodes, even for moderately large values of the frame length $L$. We can make this intuition more formal as follows.

Define
\begin{description}
 \item $M(t):$ Number of frames completed until time $t$.
 \item $R_{k,i}:$ Number of successes of Node~$i$ in frame $k$, $i=1,\ldots,n$, $k=1,2,\ldots$; observe that under the assumption of Bernoulli attempt processes ((A1) and (A2)), $\{R_{k,i}\}, k=1,2,\ldots$ are i.i.d. for each $i$; however, the vector across $i$ is not independent. Let $ER_i$ denote the mean.
 \item $U_i(t)\define \sum_{k=1}^{M(t)}R_{k,i}:$ Total number of successes of Node~$i$ in the $M(t)$ frames, $i=1,\ldots,n$.
 \item $X_k:$ Duration of the $k^{th}$ frame. Observe that $\{X_k\}, k=1,2,\ldots$ are i.i.d. Let $EX$ denote the mean. 
\end{description}
Note that all the above quantities depend on the frame length $L$. We have omitted $L$ to ease the notational burden. 

Now the average normalized throughput of Node~$i$ over a frame is given by

\begin{align}
 \theta_i(L) &= \lim_{t\to\infty}\frac{U_i(t)\times T_d}{\sum_{k=1}^{M(t)}X_k}\nonumber\\
&= T_d\times \lim_{t\to\infty}\frac{U_i(t)/t}{\frac{1}{t}\sum_{k=1}^{M(t)}X_k}\nonumber\\
&= T_d\times \lim_{t\to\infty}\frac{\frac{1}{t}\sum_{k=1}^{M(t)}R_{k,i}}{\frac{1}{t}\sum_{k=1}^{M(t)}X_k}
\end{align}
for all $i=1,\ldots,n$.

By our definition of a frame, and by (A1) and (A2), it can be seen that the beginnings of frames are renewal instants, and the mean renewal cycle length is finite. Moreover, it can be verified that $ER_i(L)<\infty$, and $EX(L)<\infty$. Thus, by Renewal Reward Theorem, we have
\begin{align}
 \theta_i(L) &= T_d\times \frac{ER_i(L)}{EX(L)}\nonumber\\
&= \frac{T_d}{EX(L)}\times ER_i(L)
\end{align}

Then, the Jain's fairness index \cite{fairness-ix} for $\{\theta_i(L)\}_{i=1}^n$ can be computed as
\begin{align}
 J(\underline{\theta}(L)) &= \frac{(\sum_{i=1}^n\theta_i(L))^2}{n\sum_{i=1}^n\theta_i^2(L)}\nonumber\\
&= \frac{(\sum_{i=1}^nER_i(L))^2}{n\sum_{i=1}^n(ER_i(L))^2}
\end{align}

This can be taken as a measure of short term fairness of the system. For a given $L$, the closer this value is to 1, the fairer is the system. Also as $L\to\infty$, $\theta_i(L)\to\Theta/n$, the long run average throughput, and $J(\underline{\theta}(L))\to 1$. 

It still remains to compute $ER_i(L)$, $i=1,\ldots,n$. We proceed to do this next. 

Consider the tuple $\{N_u, I_u\}$ embedded at the epochs $T_u$ (starts of transmission cycles; recall from Section~\ref{sec:exact_model}). Here, $N_u\in\{1,\ldots,n\}$ denotes the number of nodes that attempted in the last transmission cycle, and $I_u\in\{0,1\}$ indicates whether Node~1 attempted or not in the last transmission cycle ($I_u=1$ if Node~1 attempted). Moreover, $I_u=0\Rightarrow N_u < n$; thus the size of the state space is $(2n-1)$. It is easy to see that under (A1) and (A2), $\{N_u, I_u\}$ is a DTMC. We provide the transition structure of this DTMC below. Denote by $P((n_a,z),(n_a^\prime,z^\prime))$ the transition probability from state $(n_a,z)$ to state $(n_a^\prime,z^\prime)$.

\subsubsection{Computation of transition probabilities $P((n_a,z),(n_a^\prime,z^\prime))$}
\label{subsubsec:fairness-ix-trans-prob-compute}

\noindent
\emph{From states $(n_a,0)$:}

\noindent
When the state is $(n_a,0)$, we know $n_a$ of the nodes transmitted in the last cycle, and Node~1 did not transmit. Thus, in the current cycle, $n_a$ nodes attempt in each slot w.p. $\bx$ ($\bx=\bc$ if $n_a>1$, and $\bx=\bs$ if $n_a=1$), and the remaining $(n-n_a)$ nodes including Node~1 attempt in each slot w.p. $\bd$. Now three types of events can happen.

\begin{enumerate}
 \item None of the nodes attempt in the next backoff slot. This happens with probability $(1-\bx)^{n_a}(1-\bd)^{n-n_a}$. Due to the assumption of Bernoulli attempt processes, this results in a renewal with state $(n_a,0)$, and the transition probabilities from there onwards remain the same.

 \item Exactly $n_a^\prime$ nodes attempt in the next backoff slot, but Node~1 does not attempt. It can be verified that this happens with probability 
\begin{align}
 q((n_a,0),(n_a^\prime,0))&= (1-\bd)\sum_{(i,j)\in G(n_a,n_a^\prime)}\bigg[\bx^i\bd^j\nonumber\\
&\times \binom{n_a}{i}(1-\bx)^{n_a-i}\nonumber\\
&\times \binom{n-1-n_a}{j}(1-\bd)^{n-1-n_a-j}\bigg]
\end{align}
Recall the definition of the sets $G(\cdot,\cdot)$ from Section~\ref{subsec:tagged-node-evolution}, Case~2. In this case, the system goes to the state $(n_a^\prime,0)$.

\item Exactly $n_a^\prime$ nodes including Node~1 attempt in the next backoff slot. It can be verified that this happens with probability 
\begin{align}
 q((n_a,0),(n_a^\prime,1))&= \bd\sum_{(i,j)\in G(n_a,n_a^\prime-1)}\bigg[\bx^i\bd^j\nonumber\\
&\times \binom{n_a}{i}(1-\bx)^{n_a-i}\nonumber\\
&\times \binom{n-1-n_a}{j}(1-\bd)^{n-1-n_a-j}\bigg]
\end{align}
In this case, the system goes to the state $(n_a^\prime,1)$.
\end{enumerate}
Combining all these, we have

\begin{align}
P((n_a,0),(n_a^\prime,0))&= \frac{q((n_a,0),(n_a^\prime,0))}{1-(1-\bx)^{n_a}(1-\bd)^{n-n_a}}\\
P((n_a,0),(n_a^\prime,1)) &= \frac{q((n_a,0),(n_a^\prime,1))}{1-(1-\bx)^{n_a}(1-\bd)^{n-n_a}}
\end{align}
for all $n_a,n_a^\prime\in\{1,\ldots,n\}$. 

\noindent
\emph{From states $(n_a,1)$:}

\noindent
When the state is $(n_a,1)$, we know $n_a$ of the nodes including Node~1 transmitted in the last cycle. Thus, in the current cycle, $n_a$ nodes including Node~1 attempt in each slot w.p. $\bx$ ($\bx=\bc$ if $n_a>1$, and $\bx=\bs$ if $n_a=1$), and the remaining $(n-n_a)$ nodes attempt in each slot w.p. $\bd$. Now three types of events can happen.

\begin{enumerate}
 \item None of the nodes attempt in the next backoff slot. This happens with probability $(1-\bx)^{n_a}(1-\bd)^{n-n_a}$. Due to the assumption of Bernoulli attempt processes, this results in a renewal with state $(n_a,1)$, and the transition probabilities from there onwards remain the same.

 \item Exactly $n_a^\prime$ nodes attempt in the next backoff slot, but Node~1 does not attempt. It can be verified that this happens with probability 
\begin{align}
 q((n_a,1),(n_a^\prime,0))&= (1-\bx)\sum_{(i,j)\in G(n_a-1,n_a^\prime)}\bigg[\bx^i\bd^j\nonumber\\
&\times \binom{n_a-1}{i}(1-\bx)^{n_a-1-i}\nonumber\\
&\times \binom{n-n_a}{j}(1-\bd)^{n-n_a-j}\bigg]
\end{align}
In this case, the system goes to the state $(n_a^\prime,0)$.

\item Exactly $n_a^\prime$ nodes including Node~1 attempt in the next backoff slot. It can be verified that this happens with probability 
\begin{align}
 q((n_a,1),(n_a^\prime,1))&= \bx\sum_{(i,j)\in G(n_a-1,n_a^\prime-1)}\bigg[\bx^i\bd^j\nonumber\\
&\times \binom{n_a-1}{i}(1-\bx)^{n_a-1-i}\nonumber\\
&\times \binom{n-n_a}{j}(1-\bd)^{n-n_a-j}\bigg]
\end{align}
In this case, the system goes to the state $(n_a^\prime,1)$.
\end{enumerate}
Combining all these, we have

\begin{align}
P((n_a,1),(n_a^\prime,0))&= \frac{q((n_a,1),(n_a^\prime,0))}{1-(1-\bx)^{n_a}(1-\bd)^{n-n_a}}\\
P((n_a,1),(n_a^\prime,1)) &= \frac{q((n_a,1),(n_a^\prime,1))}{1-(1-\bx)^{n_a}(1-\bd)^{n-n_a}}
\end{align}
for all $n_a,n_a^\prime\in\{1,\ldots,n\}$. This completes the derivation of the transition probabilities of the DTMC $\{N_u, I_u\}$. We next show how to compute the expectations $ER_i(L)$ using this DTMC.

\subsubsection{Computation of $ER_i(L)$, $i=1,\ldots,n$ when the frame starts after a success by Node~1}
\label{subsubsec:ER-L-compute}

Define
\begin{description}
 \item $ES_i(L;(n_a,z)):$ Expected number of successful transmissions by Node~$i$, $i=1,\ldots,n$, in a block of $L$ transmission cycles given that the block started with the state $(n_a,z)$, where $n_a\in\{1,\ldots,n\}$, $z\in\{0,1\}$.
\end{description}

We can make the following observations.

\begin{enumerate}
 \item $ER_i(L)=ES_i(L;(1,1))$ for all $i=1,\ldots,n$. Note that state $(1,1)$ implies that the block started with a successful attempt by Node~1.
 \item Starting in state $(1,1)$, the evolution of the success processes of all nodes except Node~1 are statistically identical. Thus, $ES_2(L;(1,1))=\cdots = ES_n(L;(1,1))$, i.e., $ER_2(L)=\cdots=ER_n(L)$. This is because starting in state $(1,1)$, in the next transmission cycle, Node~1 attempts at rate $\bs$, while all the other nodes attempt at rate $\bd$.
 \item Consider the state $(1,0)$, i.e., some node other than Node~1 succeeded in the last transmission cycle. The evolution of the success process of Node~1 starting from this state is statistically identical to the success process evolution of any Node~$i\neq 1$ starting from the state $(1,1)$. Thus,
\begin{align}
 ES_1(L;(1,0))&=ES_i(L;(1,1))\nonumber\\
&= ER_i(L)
\end{align}
for all $i=2,\ldots,n$. Hence, it suffices to compute $ES_1(L;(1,0))$ and $ES_1(L;(1,1))$. 
\end{enumerate}

For all $n_a\in\{1,\ldots,n\}$ and $z\in\{0,1\}$, $ES_1(L;(n_a,z))$ can be computed recursively as follows:
\begin{align}
 ES_1(1;(n_a,z)) &= P((n_a,z),(1,1))\nonumber\\
ES_1(L;(n_a,z)) &= \sum_{(n_a^\prime,z^\prime)}P((n_a,z),(n_a^\prime,z^\prime))\bigg[\ind_{n_a^\prime=1,z^\prime=1}\nonumber\\
&+ ES_1(L-1;(n_a^\prime,z^\prime))\bigg] \:\forall L > 1
\end{align}

\subsection{Mean Success Run Length}
\label{subsec:success-run}
In this subsection, we propose another alternative measure of short term unfairness. Let us define $r_{11}$ as the probability that the next successful transmission in the system is by Node~1 \emph{given} that the current successful transmission is by Node~1. Define $EU_1$ as the \emph{mean number of consecutive successes by Node~1 before any other node succeeds}. It is easy to see that $EU_1=\frac{1}{1-r_{11}}$. Then, $EU_1$, or equivalently, $r_{11}$, can be taken as a measure of short term unfairness in the system. The larger the value of $EU_1$ (and $r_{11}$), the more biased is the success process in favor of the currently successful node. We now explain how to compute $r_{11}$ from our approximate model.

\subsubsection{Computation of $r_{11}$ for $m=0$, arbitrary $n$}
\label{subsubsec:success-run-general-n}

Consider the Markov chain $\{N_u,I_u\}$ embedded at the epochs $T_u$ introduced in Section~\ref{subsubsec:ER-L-compute}. Define, for all $(n_a,z)\neq (1,0),(1,1)$ (i.e., all \emph{collision} states),

\noindent
\begin{description}
\item $r((n_a,z),(1,1)):$ Probability that the next \emph{success} state is due to Node~1 (i.e., $(1,1)$) given that the current state is $(n_a,z)$
\end{description}

\noindent
Then, for all $(n_a,z)\neq (1,0),(1,1)$, $r((n_a,z),(1,1))$ can be obtained as the solution to the following system of linear equations ($(2n-3)$ linear equations in $(2n-3)$ variables):
\begin{align}
 r((n_a,z),(1,1)) &= P((n_a,z),(1,1))\nonumber\\
&+\sum_{(n_a^\prime,z^\prime)\neq (1,0),(1,1)}\bigg[P((n_a,z),(n_a^\prime,z^\prime))\nonumber\\
&\times r((n_a^\prime,z^\prime),(1,1))\bigg]\:\forall (n_a,z)\neq (1,0),(1,1)
\end{align}
where $P((n_a,z),(n_a^\prime,z^\prime))$ are as derived in Section~\ref{subsubsec:fairness-ix-trans-prob-compute}. 

The above expression can be explained as follows: the next success state can be due to Node~1 if either (i) Node~1 succeeds in the next transmission cycle; probability of this event is given by the first term on the R.H.S.; or (ii) the next transmission cycle results in a collision leading to some state $(n_a^\prime,z^{\prime})$, and starting from that state, the next success state is due to Node~1; the second term on the R.H.S gives the probabilities of these events. 

Finally, $r_{11}$ can be computed using the same argument as above, and is given by
\begin{align}
 r_{11} &= P((1,1),(1,1)) \nonumber\\
&+ \sum_{(n_a^\prime,z^\prime)\neq (1,0),(1,1)}P((1,1),(n_a^\prime,z^\prime))r((n_a^\prime,z^\prime),(1,1))
\end{align}

\subsubsection{Computation of $r_{11}$ for $n=2$, arbitrary $m$}
\label{subsubsec:success-run-long-distance}
Consider a Markov chain $\{Y_u\}\in\{0_{s,1},0_{s,2},0_c,\pm 1,\ldots,\pm m\}$ embedded at the epochs $T_u$. This Markov chain keeps track of the misalignment of the backoff counter of Node~1 w.r.t. Node~2 (in case of a collision), as well as the successful Node Id (in case of a success). The state values can be interpreted as follows:

\noindent
$0_{s,1}:$ Node~1 was successful in the last cycle

\noindent
$0_{s,2}:$ Node~2 was successful in the last cycle

\noindent
$0_c:$ There was a collision in the last cycle, but the backoff counters of the nodes are aligned, i.e., both start counting down their backoffs at $T_u$

\noindent
$+k:$ There was a collision, and Node~1's backoff is deferred by $k$ slots, i.e., Node~1 will start backoff countdown at $T_u+k$, for all $k=1.\ldots,m$

\noindent
$-k:$ There was a collision, and Node~2's backoff is deferred by $k$ slots, i.e., Node~1 starts backoff at $T_u$, Node~2 starts backoff at $T_u+k$, for all $k=1,\ldots,m$

\noindent
It is easy to see that under (A1) and (A2), $\{Y_u\}$ is a DTMC. Denote by $P(y,y^\prime)$, the transition probability from state $y$ to state $y^\prime$ in this DTMC. 

Define, for all $y\neq 0_{s,1},0_{s,2}$, 

\noindent
$r(y,0_{s,1}):$ Probability that the next success state is due to Node~1 given that the current system state is $y$

\noindent
Then, using the same arguments as in Section~\ref{subsubsec:success-run-general-n}, $r(y,0_{s,1})$ for all $y\neq 0_{s,1},0_{s,2}$ can be obtained as the solution to a system of linear equations ($(2m+1)$ linear equations in $(2m+1)$ variables) as follows:
\begin{align}
 r(y,0_{s,1}) &= P(y,0_{s,1}) + \sum_{y^\prime \neq 0_{s,1},0_{s,2}}P(y,y^\prime)r(y^\prime,0_{s,1})\:\forall y\neq 0_{s,1},0_{s,2}
\end{align}

Finally, $r_{11}$ can be computed as follows:
\begin{align}
 r_{11} &= P(0_{s,1},0_{s,1}) + \sum_{y\neq 0_{s,1},0_{s,2}}P(0_{s,1},y)r(y,0_{s,1})
\end{align}

The transition probabilities $P(y,y^\prime)$ can be computed using the same renewal arguments used for computing the transition probabilities in Section~\ref{subsec:mrp-analysis-given-beta}, Case~1. We omit the details for brevity, and directly write down the expressions for the transition probabilities used in the above derivation.

\begin{align}
 P(0_{s,1},0_{s,1}) &= \frac{\bs(1-\bd)^{m+1}}{1-(1-\bd)(1-\bs)}\\
P(0_{s,1},+k) &= \frac{\bs(1-\bd)^k\bd}{1-(1-\bd)(1-\bs)},\:0\leq k\leq m\\
P(0_{s,1},-k) &= \frac{\bd(1-\bs)^k\bs}{1-(1-\bd)(1-\bs)},\:1\leq k\leq m\\
P(0_c,0_{s,1}) &= \frac{\bc(1-\bc)^{m+1}}{1-(1-\bc)^2}\\
P(0_c,+k) &= \frac{\bc^2(1-\bc)^k}{1-(1-\bc)^2},\:0\leq k\leq m\\
P(0_c,-k) &= \frac{\bc^2(1-\bc)^k}{1-(1-\bc)^2},\:1\leq k\leq m\\
P(+k,0_{s,1}) &= (1-\bc)^k P(0_c,0_{s,1}),\:1\leq k\leq m\\
P(-k,0_{s,1}) &= (1-\bc)^k P(0_c,0_{s,1})\nonumber\\
&+\sum_{j=1}^k (1-\bc)^{j-1}\bc (1-\bc)^{j+m-k},\:1\leq k\leq m
\end{align}

For all $k,k^\prime \in\{1,\ldots,m\}$,
\begin{align}
 P(+k,+k^\prime) &= (1-\bc)^k P(0_c,+k^\prime)=P(-k,-k^\prime)\\
P(+k,-k^\prime) &= (1-\bc)^k P(0_c,-k^\prime)\nonumber\\
&+\sum_{j=\max\{k+1-k^\prime,1\}}^k \bigg[(1-\bc)^{j-1}\bc\nonumber\\
&\times (1-\bc)^{j+k^\prime-k-1}\bc\bigg] = P(-k,+k^\prime)
\end{align}

\section{Optimizing the Backoff Sequence for Throughput and Fairness}
\label{sec:optimization-throughput-fairness}
Intuitively, an unfair system may actually achieve higher system throughput than a fair system, since in the former, one node or the other will have unhindered access to the channel over extended periods, whereas in the latter, there will be more contention. However, a high long run average system throughput does not yield the desired quality of experience in the presence of significant short term unfairness. Now that we have developed methods to quantify the extent of short term unfairness in a system, we can use these measures to tune the protocol parameters to achieve desired throughput and fairness objectives. In particular, our interest is in maximizing system throughput subject to some constraint on the extent of short term unfairness. We demonstrate with an example how we can do this using our analytical methods for the case of $n=2$, and large propagation delay, $m$ (in slots). For the purposes of this example, we use the mean success run length, $EU_1$ as the measure of short term unfairness. The advantage of this over the throughput fairness index measure is that if we use throughput fairness index as the fairness measure, then we need to specify two values, namely, the value of $L$, as well as the target fairness index\footnote{A fairness index of 1 is achievable only as $L\to\infty$. For any finite $L$, we need to specify a target value $1-\epsilon$.} to specify the optimization problem, whereas if we use $EU_1$, we need to specify only the target value for $EU_1$ (i.e., an upper bound). 

\subsection{Throughput maximization subject to short term fairness: an example}
\label{subsec:throughput-stu-optimization}
Consider a system with $n=2$, and propagation delay of $m=10$ backoff slots. The system uses the IEEE~802.11 backoff expansion framework with default values for $p,K,$ and maximum backoff exponent, $maxBE$, namely, $p=2$, $K=6$, and $maxBE=10$, i.e., the maximum backoff a node can take is $2^{maxBE}=1024$. For the purposes of optimization, we treat the minimum backoff exponent, $minBE$ as the free variable. Recall that the initial backoff window of a node is $[1,2^{minBE}]$. Our aim is to choose $minBE$ to maximize system throughput subject to the fairness constraint that $EU_1(minBE)<3$. 

To this end, we proceed as follows. We first compute $EU_1$ as a function of $minBE$ for the given system for $0\leq minBE\leq 10$. The results are shown in Panel~1 of Figure~\ref{fig:throughput-fairness-optimization}. Also shown is the target fairness objective. As can be seen from the plot, $minBE\geq 6$ achieves the fairness objective of $EU_1(minBE)<3$.  
\begin{figure}[t]
%\footnotesize
\begin{center}
\includegraphics[scale=0.37]{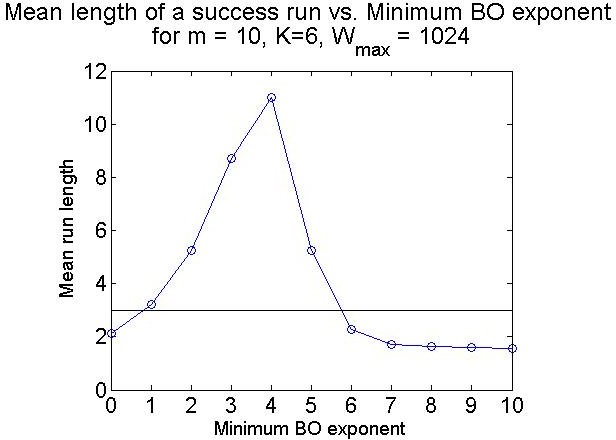}
\hspace{0.1mm}
\includegraphics[scale=0.35]{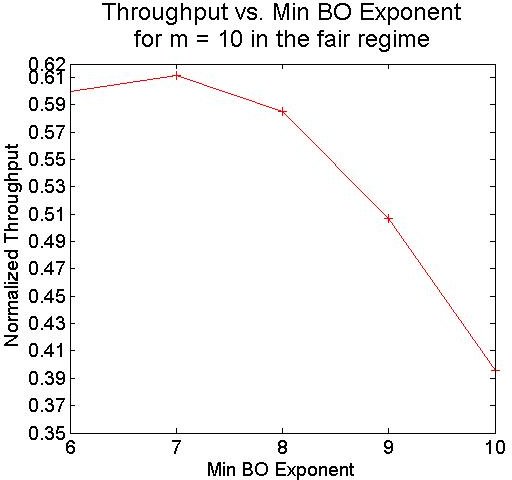}
\caption{Throughput maximization subject to short term fairness constraint for $n=2, m=10$. (Panels are numbered row-wise, from left to right) Panel~1: Mean success run length as a function of $minBE$; the flat line is the target fairness objective. Panel~2: Throughput as a function of $minBE$ in the fair regime.}
\label{fig:throughput-fairness-optimization}
\vspace{-5mm}
\end{center}
\normalsize
%\vspace{-6mm}
\end{figure}
We next compute the system throughput as a function of $minBE$ in this ``fair regime'', i.e., for $6\leq minBE\leq 10$. The results are shown in Panel~2 of Figure~\ref{fig:throughput-fairness-optimization}. It can be seen from the plot that $minBE=7$ achieves the maximum throughput for this system subject to the fairness constraint. 

\noindent
\remarks

\begin{enumerate}
 \item We see from Panel~1 of Figure~\ref{fig:throughput-fairness-optimization} that $EU_1$ first increases with $CW_{\min}$, then decreases. This can be explained by looking at the corresponding backoff sequences. When $minBE=0$, the backoff sequence is $[1,2,4,8,16,32,64]$. Thus, the difference between minimum and maximum possible backoff window size is 63. When $minBE=1$, the backoff sequence is $[2,4,8,16,32,64,128]$. Thus, the difference between minimum and maximum possible backoff window size is 126, more than the previous case. This difference (which can be taken as a measure of the backoff variability) keeps on increasing until $minBE=4$, at which point it is 1008. This causes the short term unfairness to increase. Beyond $minBE=4$, the difference starts decreasing, since the maximum backoff window size is clamped at 1024, and the minimum backoff window ($2^{minBE}$) keeps increasing. This causes the short term unfairness to decrease.

\item In the fairness regime, throughput shows a general decreasing trend, since when $minBE$ is already large, further increasing $minBE$ causes an increase in the system idle time, without significantly improving the collision probability.  
\end{enumerate}

\section{Revisiting Bianchi Analysis: Some Observations}
\label{sec:bianchi-revisited}
In this section, we aim to explain the scope and limitations of the standard f.p. analysis due to Bianchi \cite{bianchi00performance} using our generalized approximate system model of Section~\ref{sec:mrp-state-dependent}. We start by reviewing the system model and assumptions in the Bianchi analysis.

\subsection{Independence assumption in the Bianchi model}
\label{subsec:bianchi-independence}
In Bianchi's analysis, the system evolution is modeled as follows: in backoff time, in each backoff slot, each node attempts i.i.d. with probability $\beta$. 

In this system evolution, consider a Markov chain embedded at the success epochs; the Markov chain tracks the node id of the successful node at each success epoch, and has state space $\{1,2,\ldots,n\}$. Then, under the above assumption of Bernoulli attempt processes with state independent rates, the transition probabilities of this Markov chain are $p_{i,j}=\frac{1}{n}$, for all $i,j\in\{1,\ldots,n\}$; thus, the underlying assumption in Bianchi's model is that \emph{the success process is i.i.d.}. When is this a good assumption? We aim to provide some partial answers to this question using our generalized system model.

\subsection{The independence assumption in the light of the MRP model}
\label{subsec:mrp-bianchi-relation}

Consider again the Markov chain of successful node ids (embedded at the success epochs) in the generalized system model with state dependent Bernoulli attempt processes introduced in Section~\ref{subsec:system-evolution-mrp}. By symmetry, the transition probabilities of this Markov chain in the generalized model satisfy $p_{i,i}=p_{j,j}$ for all $i\neq j$, and $p_{i,j}=\frac{1-p_{i,i}}{n-1}$ for all $i$, for all $j\neq i$. 

Let us compare the transition probability matrix (t.p.m.) of this Markov chain under the Bianchi model with that under the generalized model. Due to the symmetry property mentioned above, it is enough to compare a single row in the t.p.m.; without loss of generality, consider Row~1. The KL distance between the first rows is easily seen to be $\log n - H(p_{1,1})$, i.e., \emph{the difference between the entropies of the two p.m.fs}. This suggests that the i.i.d. success process assumption in the Bianchi model is accurate \emph{when the entropy of a row of the t.p.m. in the generalized MRP model is close to maximum}, i.e., $\log n$. This is also intuitive, since independence, or a lack of correlation in the success process would imply high level of uncertainty in the system evolution. 

We further explore the implications of this observation for the simplest case of $n=2,m=0$. Note that in this case, achieving $H(p_{1,1})$ close to 1 is equivalent to achieving $p_{1,1}$ (and equivalently, $p_{1,2}$) close to 1/2. We have the following lemma.

\begin{lemma}
\label{lem:bd-bs}
 For any sufficiently small $\epsilon > 0$, to achieve $\frac{1}{2}-\epsilon \leq p_{1,2}\leq \frac{1}{2}+\epsilon$ for a system with $n=2,m=0$, it suffices to have $1-2\epsilon\leq \frac{\bd}{\bs}\leq 1+2\epsilon$.
\end{lemma}
  
\begin{proof}
 It can be shown, using the method described in Section~\ref{subsec:success-run} (Note that $p_{1,1}$ is nothing but $r_{11}$ from Section~\ref{subsec:success-run}), that for $n=2,m=0$,
\begin{align}
 p_{1,2} &= \frac{\bd(1-\frac{1}{2}\bs)}{\bs+\bd(1-\bs)}
\end{align}
Then, simple algebraic manipulations yield that to achieve $p_{1,2}\geq x$ (respectively, $\leq x$), for any $0\leq x\leq 1$, we need $\frac{\bd}{\bs}\geq \frac{x}{1-x+(x-1/2)\bs}$ (respectively, $\frac{\bd}{\bs}\leq \frac{x}{1-x+(x-1/2)\bs}$). Thus, to achieve $\frac{1}{2}-\epsilon \leq p_{1,2}\leq \frac{1}{2}+\epsilon$, we need 
\begin{align}
 \frac{\frac{1}{2}-\epsilon}{\frac{1}{2}+\epsilon-\epsilon\bs}\leq \frac{\bd}{\bs}\leq \frac{\frac{1}{2}+\epsilon}{\frac{1}{2}-\epsilon+\epsilon\bs}
\end{align}
Since $\epsilon(1-\bs)>0$, to achieve the above, it suffices to have $1-2\epsilon \leq \frac{\bd}{\bs}\leq 1 + 2\epsilon$.
\end{proof}

Lemma~\ref{lem:bd-bs} implies that for $n=2,m=0$, the independence assumption in Bianchi's model is accurate when $\frac{\bd}{\bs}$ is close to 1 in the generalized MRP model. This is also intuitively satisfactory, since this makes the attempt processes of the successful node, and the interrupted node indistinguishable. 
 
\section{Conclusion}
\label{sec:conclusion-chap8}
We have considered a class of single-hop networks with saturated, IEEE~802.11 DCF based transmitters and their receivers, where the system exhibits a performance anomaly known as short term unfairness. We have demonstrated with several examples that short term unfairness abounds; it arises for several classes of backoff sequences, as well as when the propagation delays among the nodes are \emph{non-negligible} compared to the slot duration, and the standard fixed point analysis (or simple extensions thereof) do not predict the system performance well in such cases. We then proposed a detailed stochastic model of the system evolution, and developed a novel approximate, yet accurate, analysis of this model. Interestingly, for the case of non-negligible propagation delays, we observed that as propagation delay increases, the collision probability of a node initially increases, but then flattens out, contrary to simple intuition (Figure~\ref{fig:approx-model-gamma-validation}). Moreover, in such systems, after a successful transmission, the attempt rate of the successful node is much higher than the other nodes, a reflection of the short-term unfairness property (see, for example, Figure~\ref{fig:approx-model-beta-theta-validation}). We further explored the use of the approximate analysis for maximizing system throughput; we observed that at lower propagation delays, collision probability dominates throughput, while at higher delays, slot duration takes over as the dominant factor (Figure~\ref{fig:theta-vs-m-approx-model}). We also demonstrated the use of the analytical model to quantify the extent of short term unfairness in the system, and to tune the protocol parameters to achieve desired throughput and fairness objectives (Figure~\ref{fig:throughput-fairness-optimization}). Finally, we also explored an interesting connection between the assumptions in the standard f.p. analysis, and our generalized system model.

\section{Appendix}
\label{sec:appendix-wlan-analysis}
\subsection{Derivation of stationary probabilities of the Markov chain in Section~\ref{subsec:tagged-node-evolution}}
Our goal is to derive the stationary probabilities of the Markov chain $\{S_v,N_v\}\in\{0,1,\ldots,K\}\times\{1,\ldots,n\}$ embedded at the starts of backoff cycles, $T^{\prime(i)}_v$, of the tagged node, Node~$i$. 

We first need to derive the transition structure of the Markov chain. However, note that the tagged node can get interrupted in a backoff cycle due to a success by a single node, or a collision (simultaneous attempts) by several other nodes, and the evolution therefrom depends on the number of attempting nodes at that interruption instant (Recall Approximation~(A4)). Hence in this case, to derive the required stationary probabilities, it is more convenient to embed the concerned Markov chain (and the MRP) within a bigger auxiliary Markov chain (and MRP), namely a Markov chain embedded at the instants $T_u$ (the starts of transmission cycles; see Figures~\ref{fig:tx-cycle-explain} and \ref{fig:bo-cycle-explain}). To construct the auxiliary Markov chain, we associate with each $T_u$, three states, namely, (i) $S_u$, the backoff stage of Node~$i$ at $T_u$, (ii) $N_u$, the number of nodes that attempted in the just concluded transmission cycle, (iii) $B_u$, the residual backoff of Node~$i$ at $T_u$. Under Approximations~(A3) and (A4), it is easy to observe that $(S_u, N_u, B_u)$ is a DTMC embedded at the instants $\{T_u\}$ (and $\{(S_u,N_u,B_u), T_u\}$ is a Markov Renewal Process), with state space $\subset\{0,\ldots,K\}\times\{1,\ldots,n\}\times\{0,\ldots,W_K\}$. To see that this auxiliary Markov chain contains within it, the concerned Markov chain, simply observe that the set of states with $B_u=0$ is exactly the set of states in the original Markov chain (that was embedded at $\{T^{\prime(i)}_v\}$). Note that in the auxiliary chain, $B_u\neq 0\Rightarrow$ Node~$i$ was interrupted in the previous transmission cycle by $N_u$ other nodes. This facilitates tracking the evolution from an interruption instant of Node~$i$. 

We make the following simple observations about the state space of the auxiliary chain.
\begin{enumerate}
 \item $B_u>0\Rightarrow$ Node~$i$ was interrupted in the last transmission cycle $\Rightarrow N_u < n$.
 \item If Node~$i$ was interrupted in backoff stage $S_u=k$, then $B_u\in\{1,\ldots, W_k-1\}$, $k=0,\ldots,K$.
 \item $N_u=1, B_u=0\Rightarrow$ Node~$i$ transmitted successfully in the last transmission cycle $\Rightarrow S_u = 0$.  
\end{enumerate}
With the above observations, it can be verified that the total number of states in the auxiliary chain is $(n-1)\sum_{k=0}^K (W_k-1) + (n-1)(K+1) + 1$, which still grows linearly in the number of nodes. 

We now proceed to derive the transition structure of the auxiliary chain. We start by defining the following sets, which will be useful later in writing the transition probabilities.

Define, for all $0\leq x\leq n-1$ and for all $0\leq y\leq n-1$,

\begin{align}
 G(x,y) = \{(i,j): 0\leq i\leq x,0\leq j\leq n-1-x, i+j = y\}
\end{align}

Let $Q$ be the transition probability matrix of the auxiliary Markov chain, i.e., we denote by $Q((s,n_a,b),(s^\prime,n_a^\prime,b^\prime))$, the transition probability from the state $(s,n_a,b)$ to the state $(s^\prime,n_a^\prime,b^\prime)$ in the auxiliary chain. 

\noindent
\emph{Transition probabilities from states of the form $(s,n_a,0)$:}

\noindent
When the state is $(s,n_a,0)$, we know that Node~$i$ transmitted in the last transmission cycle along with $(n_a-1)$ other nodes, and its current backoff stage is $s$. Then, by our approximation~(A3), Node~$i$ will sample a new backoff uniformly from $[1,W_s]$, while $(n_a-1)$ other nodes will attempt independently w.p. $\bc$ in each backoff slot, and the remaining $(n-n_a)$ nodes will attempt independently w.p. $\bd$ in each backoff slot. Now 3 types of events can occur in the next transmission cycle. 

\begin{enumerate}
 \item Node~$i$ successfully transmits. This happens if Node~$i$ samples a backoff of $l$ slots, $1\leq l\leq W_s$, and all the other nodes remain silent for these $l$ slots. Using the Bernoulli attempt process approximation for the other nodes, the probability of this event is 
\begin{align}
 Q((s,n_a,0),(0,1,0)) &= \frac{1}{W_s}\sum_{l=1}^{W_s}(1-\bd)^{l(n-n_a)}(1-\bc)^{l(n_a-1)}\label{eqn:trans-prob-gen-n-first-eqn}
\end{align}
for all $s\in\{0,\ldots,K\}$, $n_a\in\{1,\ldots,n\}$.

\item Node~$i$ transmits and encounters a collision with $n^\prime_a-1$ other nodes. This happens if Node~$i$ samples a backoff of $l$ slots, $1\leq l\leq W_s$, and among the remaining $(n-1)$ nodes, exactly $(n^\prime_a-1)$ nodes attempt together at the $l^{th}$ slot, while the rest of the nodes remain silent. The probability of this event can be seen to be
 \begin{align}
  Q((s,n_a,0),((s+1)mod(K+1),n_a^\prime,0)) &= \frac{1}{W_s}\sum_{l=1}^{W_s}\sum_{(i,j)\in G(n_a-1,n^\prime_a-1)}((1-\bd)^{l-1}\bd)^j\nonumber\\
&\times((1-\bc)^{l-1}\bc)^i \binom{n_a-1}{i}(1-\bc)^{l(n_a-1-i)}\nonumber\\
&\times \binom{n-n_a}{j}(1-\bd)^{l(n-n_a-j)}
 \end{align}
for all $s\in\{0,\ldots,K\}$, $n_a\in\{1,\ldots,n\}$, $n^\prime_a\in\{2,\ldots,n\}$. The term corresponding to pair $(i,j)$ inside the second summation above is the probability that among the $(n_a-1)$ nodes (excluding Node~$i$) that attmepted in the previous cycle, exactly $i$ nodes attempt together in the $l^{th}$ slot of the current cycle, among the $(n-n_a)$ nodes that did not attempt in the previous cycle, exactly $j$ nodes attempt together in the $l^{th}$ slot in the current cycle, and the remaining $n-n_a^\prime$ nodes remain silent (note that by our definition of sets $G(\cdot,\cdot)$, $i+j= n^\prime_a-1$). 

\item Node~$i$ is interrupted by $n_a^\prime$ nodes, and its residual backoff is $b$, $1\leq b\leq W_s-1$. This can happen only if Node~$i$ samples a backoff $l\geq b+1$, and among the remaining $(n-1)$ nodes, exactly $n_a^\prime$ nodes attempt at the $(l-b)^{th}$ slot, while the rest of the nodes remain silent. Using similar arguments as above, the probability of this event can be verified to be

\begin{align}
 Q((s,n_a,0),(s,n_a^\prime,b)) &= \frac{1}{W_s}\sum_{l=b+1}^{W_s}\sum_{(i,j)\in G(n_a-1,n^\prime_a)}((1-\bd)^{l-b-1}\bd)^j\nonumber\\
&\times((1-\bc)^{l-b-1}\bc)^i \binom{n_a-1}{i}(1-\bc)^{(l-b)(n_a-1-i)}\nonumber\\
&\times \binom{n-n_a}{j}(1-\bd)^{(l-b)(n-n_a-j)}
\end{align}
for all $s\in\{0,\ldots,K\}$, $n_a\in\{1,\ldots,n\}$, $n^\prime_a\in\{1,\ldots,n-1\}$, $b\in\{1,\ldots,W_s-1\}$.
\end{enumerate}

\noindent
\emph{Transition probabilities from states of the form $(s,n_a,b)$ with $b>0$:}

\noindent
When the state is $(s,n_a,b)$ with $b>0$, we know that Node~$i$ was interrupted in the last transmission cycle by transmissions of $n_a$ other nodes, and its current backoff stage and residual backoff are $s$ and $b$ respectively. Then, by our approximation~(A4), Node~$i$ will resume its residual backoff countdown, while $n_a$ other nodes will attempt independently w.p. $\bc$ (respectively $\bs$) in each backoff slot if $n_a>1$ (respectively $n_a=1$), and the remaining $(n-1-n_a)$ nodes will attempt independently w.p. $\bd$ in each backoff slot. Now 3 types of events can occur in the next transmission cycle. 

\begin{enumerate}
 \item Node~$i$ transmits successfully. This happens if none of the other nodes attempt in the next $b$ slots. The probability of this event is

\begin{align}
 Q((s,n_a,b),(0,1,0)) &= (1-\bd)^{b(n-1-n_a)}(1-\bx)^{bn_a}
\end{align}
for all $s\in\{0,\ldots,K\}$, $n_a\in\{1,\ldots,n-1\}$, $b\in\{1,\ldots,W_s-1\}$. Here, $\bx=\bc$ if $n_a > 1$, and $\bx=\bs$ if $n_a=1$. 

\item Node~$i$ transmits and collides with $(n^\prime_a-1)$ other nodes. This happens if exactly $(n^\prime_a-1)$ other nodes attempt at the $b^{th}$ slot, and the rest of the nodes remain silent. Proceeding along similar lines as before, the probability of this event can be obtained as

\begin{align}
 Q((s,n_a,b),((s+1)mod(K+1),n_a^\prime,0)) &= \sum_{(i,j)\in G(n_a,n^\prime_a-1)}((1-\bx)^{b-1}\bx)^i\nonumber\\
&\times ((1-\bd)^{b-1}\bd)^j\binom{n_a}{i}(1-\bx)^{b(n_a-i)}\nonumber\\
&\times \binom{n-1-n_a}{j}(1-\bd)^{b(n-1-n_a-j)}
\end{align}
for all $s\in\{0,\ldots,K\}$, $n_a\in\{1,\ldots,n-1\}$, $b\in\{1,\ldots,W_s-1\}$, $n_a^\prime\in\{2,\ldots,n\}$. Here $\bx$ has the same interpretation as before.

\item Node~$i$ is again interrupted due to transmission by $n_a^\prime$ nodes, and its residual backoff is $b^\prime$. This can happen only if $b^\prime < b$, and exactly $n_a^\prime$ other nodes attempt at the $(b-b^\prime)^{th}$ slot, while the rest of the nodes remain silent. Using similar arguments as before, the probability of this event can be seen to be

\begin{align}
 Q((s,n_a,b),(s,n_a^\prime,b^\prime)) &= \sum_{(i,j)\in G(n_a,n^\prime_a)}((1-\bx)^{b-b^\prime-1}\bx)^i\nonumber\\
&\times ((1-\bd)^{b-b^\prime-1}\bd)^j\binom{n_a}{i}(1-\bx)^{(b-b^\prime)(n_a-i)}\nonumber\\
&\times \binom{n-1-n_a}{j}(1-\bd)^{(b-b^\prime)(n-1-n_a-j)}\label{eqn:transition-prob-gen-n-last-eqn}
\end{align}
for all $s\in\{0,\ldots,K\}$, $n_a\in\{1,\ldots,n-1\}$, $b\in\{1,\ldots,W_s-1\}$, $n_a^\prime\in\{1,\ldots,n-1\}$, $1\leq b^\prime \leq b-1$. Here $\bx$ has the same interpretation as before.
\end{enumerate}

This completes the derivation of the transition structure of the auxiliary Markov chain. All other entries in $Q$ are zero. 

It is easy to observe that the auxiliary DTMC is finite, irreducible (from any state, the state $(0,1,0)$ can be reached in one step, and from $(0,1,0)$, any state can be reached), and hence \emph{positive recurrent}. We denote by $\phi$, the stationary distribution of this Markov chain, which can be obtained as the unique solution to the system of equations $\phi = \phi Q$, subject to $\phi$ being a probability distribution. 

From the stationary distribution $\phi$ of the auxiliary Markov chain, we can obtain the stationary distribution $\psi$ of our original intended Markov chain (embedded at the backoff completion points $T^{\prime (i)}_v$ of the tagged node) as follows:

\begin{align}
 \psi(s,n_a) &= \frac{\phi(s,n_a,0)}{\sum_{(s^\prime,n_a^\prime,0)}\phi(s^\prime,n_a^\prime,0)}
\end{align}
for all $s\in\{0,\ldots,K\}$, $n_a\in\{1,\ldots,n\}$.

%\footnotesize
\bibliography{PhD_thesis_references}
\end{document}